\documentclass[10pt,journal,compsoc]{IEEEtran}
\usepackage{graphicx}

\usepackage{epsfig,endnotes}
\usepackage{todonotes}
\usepackage{latexsym} 
\usepackage{subfig}
\usepackage{algorithm}
\usepackage[noend]{algpseudocode}
\usepackage{listings}
\usepackage{empheq}

\usepackage{amsmath}
\usepackage{amssymb}

\usepackage{amsthm}
\usepackage{amssymb}
\usepackage{amsopn}

\usepackage{multirow}
\usepackage{colortbl}
\makeatletter
\def\BState{\State\hskip-\ALG@thistlm}
\makeatother
\usepackage{hyperref}
\usepackage{mathtools}
\usepackage{paralist}
\usepackage{hyphenat}

\usepackage{tikz}
\usetikzlibrary{shapes}
\usetikzlibrary{calc,shapes.multipart,chains,arrows,positioning}

\newcommand{\chead}{\cellcolor{tableHeadGray}}
\newcommand{\thead}[1]{\textbf{#1}}

\newcommand{\parttitle}[1]{\smallskip\noindent\textbf{#1.}}
\newcommand{\myproofpar}[1]{\smallskip\noindent\underline{{#1}:}}
\newcommand{\cList}{\thickspace \| \thickspace}

\newcommand{\mathtext}[1]{\thickspace\text{#1}\thickspace}

\DeclareMathOperator{\argmin}{argmin}

\DeclareMathOperator{\defas}{:=}
\newcommand{\vExpr}{\mathbf{v}}
\newcommand{\cExpr}{\mathbf{c}}
\newcommand{\constExpr}{\mathbf{const}}
\newcommand{\attrExpr}{\mathbf{attr}}
\newcommand{\compExpr}{\mathbf{cmp}}
\newcommand{\exprEval}{\textsf{eval}}

\newcommand{\aDom}{\mathcal{U}}
\newcommand{\supp}{\textsc{supp}}

\newcommand{\query}{Q}
\newcommand{\qSub}{{\query_{sub}}}

\newcommand{\schema}[1]{\textsc{Sch}(#1)}

\newcommand{\opparents}{\textsf{parents}}

\newcommand{\projection}{\Pi}
\newcommand{\selection}{\sigma}
\newcommand{\aggregation}{\gamma}
\newcommand{\Aggregation}[2]{{}_{#1}\aggregation_{#2}}
\newcommand{\union}{\cup}
\newcommand{\intersection}{\cap}
\newcommand{\difference}{-}

\newcommand{\duplicate}{\delta}

\newcommand{\join}{\bowtie}
\newcommand{\crossprod}{\times}

\newcommand{\win}{\omega}
\newcommand{\Win}[4]{\win_{#1 \to #2, #3\|#4}}

\newcommand{\colsOf}{cols}

\newcommand{\eIf}[3]{\textsf{if}\thickspace #1 \thickspace
  \textsf{then} \thickspace #2 \thickspace \textsf{else} \thickspace #3}

\newcommand{\qH}{h}
\newcommand{\qD}{d}

\newcommand{\ecProp}{ec}
\newcommand{\ecb}{{\ecProp_{B}}}
\newcommand{\ect}{{\ecProp_{T}}}
\newcommand{\keyProp}{key}
\newcommand{\icolsProp}{icols}
\newcommand{\setProp}{set}

\newcommand{\aKey}{k}
\newcommand{\anEC}{E}

\newcommand{\ecClosure}{{\cal E}^*}

\newcommand{\aEquiv}{\simeq}

\newcommand{\curOp}{\Diamond}
\newcommand{\rootOp}{\circledast}
\newcommand{\minKey}{\textsc{Min}}

\definecolor{black}{rgb}{0,0,0}
\definecolor{lgrey}{rgb}{0.9,0.9,0.9}
\definecolor{grey}{rgb}{0.8,0.8,0.8}
\definecolor{red}{rgb}{1,0,0}
\definecolor{green}{rgb}{0,1,0}
\definecolor{darkgreen}{rgb}{0,0.5,0}
\definecolor{darkblue}{rgb}{0,0,0.5}
\definecolor{darkpurple}{rgb}{0.5,0,0.5}
\definecolor{darkdarkpurple}{rgb}{0.3,0,0.3}
\definecolor{blue}{rgb}{0,0,1}
\definecolor{shadegreen}{rgb}{0.95,1,0.95}
\definecolor{shadeblue}{rgb}{0.95,0.95,1}
\definecolor{shadered}{rgb}{1,0.85,0.85}
\definecolor{oddRowGrey}{rgb}{0.95,0.95,0.95}
\definecolor{evenRowGrey}{rgb}{0.85,0.85,0.85}
\definecolor{tableHeadGray}{rgb}{0.85,0.85,0.85}

\newtheorem{Theorem}{Theorem}
\newtheorem{Definition}{Definition}
\newtheorem{Lemma}{Lemma}

\newtheorem{Example}{Example}
\makeatletter
\algrenewcommand{\algorithmiccomment}[1]{\hfill //\,\textit{#1}}
\renewcommand{\ALG@beginalgorithmic}{\footnotesize}
\makeatother

\begin{document}

\definecolor{lstpurple}{rgb}{0.5,0,0.5}
\definecolor{lstred}{rgb}{1,0,0}
\definecolor{lstreddark}{rgb}{0.7,0,0}
\definecolor{lstredl}{rgb}{0.64,0.08,0.08}
\definecolor{lstmildblue}{rgb}{0.66,0.72,0.78}
\definecolor{lstblue}{rgb}{0,0,1}
\definecolor{lstmildgreen}{rgb}{0.42,0.53,0.39}
\definecolor{lstgreen}{rgb}{0,0.5,0}
\definecolor{lstorangedark}{rgb}{0.6,0.3,0}	
\definecolor{lstorange}{rgb}{0.75,0.52,0.005}
\definecolor{lstorangelight}{rgb}{0.89,0.81,0.67}
\definecolor{lstbeige}{rgb}{0.90,0.86,0.45}

\DeclareFontShape{OT1}{cmtt}{bx}{n}{<5><6><7><8><9><10><10.95><12><14.4><17.28><20.74><24.88>cmttb10}{}

\lstdefinestyle{psql}
{
tabsize=2,
basicstyle=\small\upshape\ttfamily,
language=SQL,
morekeywords={PROVENANCE,BASERELATION,INFLUENCE,COPY,ON,TRANSPROV,TRANSSQL,TRANSXML,CONTRIBUTION,COMPLETE,TRANSITIVE,NONTRANSITIVE,EXPLAIN,SQLTEXT,GRAPH,IS,ANNOT,THIS,XSLT,MAPPROV,cxpath,OF,TRANSACTION,SERIALIZABLE,COMMITTED,INSERT,INTO,WITH,SCN,UPDATED},
extendedchars=false,
keywordstyle=\bfseries,
mathescape=true,
escapechar=@,
sensitive=true
}

\lstdefinestyle{psqlcolor}
{
tabsize=2,
basicstyle=\small\upshape\ttfamily,
language=SQL,
morekeywords={PROVENANCE,BASERELATION,INFLUENCE,COPY,ON,TRANSPROV,TRANSSQL,TRANSXML,CONTRIBUTION,COMPLETE,TRANSITIVE,NONTRANSITIVE,EXPLAIN,SQLTEXT,GRAPH,IS,ANNOT,THIS,XSLT,MAPPROV,cxpath,OF,TRANSACTION,SERIALIZABLE,COMMITTED,INSERT,INTO,WITH,SCN,UPDATED},
extendedchars=false,
keywordstyle=\bfseries\color{lstpurple},
deletekeywords={count,min,max,avg,sum},
keywords=[2]{count,min,max,avg,sum},
keywordstyle=[2]\color{lstblue},
stringstyle=\color{lstreddark},
commentstyle=\color{lstgreen},
mathescape=true,
escapechar=@,
sensitive=true
}

\lstdefinestyle{datalog}
{
basicstyle=\footnotesize\upshape\ttfamily,
language=prolog
}

\lstdefinestyle{pseudocode}
{
  tabsize=3,
  basicstyle=\small,
  language=c,
  morekeywords={if,else,foreach,case,return,in,or},
  extendedchars=true,
  mathescape=true,
  literate={:=}{{$\gets$}}1 {<=}{{$\leq$}}1 {!=}{{$\neq$}}1 {append}{{$\listconcat$}}1 {calP}{{$\cal P$}}{2},
  keywordstyle=\color{lstpurple},
  escapechar=&,
  numbers=left,
  numberstyle=\color{lstgreen}\small\bfseries, 
  stepnumber=1, 
  numbersep=5pt,
}

\lstdefinestyle{xmlstyle}
{
  tabsize=3,
  basicstyle=\small\upshape\ttfamily,
  language=xml,
  extendedchars=true,
  mathescape=true,
  escapechar=£,
  tagstyle=\bfseries,
  usekeywordsintag=true,
  morekeywords={alias,name,id},
  keywordstyle=\color{lstred}
}

\lstdefinestyle{xmlstyle-color}
{
  tabsize=3,
  basicstyle=\small\upshape\ttfamily,
  language=xml,
  extendedchars=true,
  mathescape=true,
  escapechar=£,
  tagstyle=\color{keywordpurple},
  usekeywordsintag=true,
  morekeywords={alias,name,id},
  keywordstyle=\color{lstred}
}

 \lstset{style=psql}

\title{Heuristic and Cost-based Optimization for Diverse Provenance Tasks\\\textbf{(extended version)}}

\author{
Xing~Niu, 
Raghav~Kapoor, 
Boris~Glavic,
Dieter~Gawlick,
Zhen~Hua~Liu,
Vasudha~Krishnaswamy,
Venkatesh~Radhakrishnan

\IEEEcompsocitemizethanks{
  \IEEEcompsocthanksitem X.Niu, R.Kapoor and B. Glavic, Department of Computer Science, Illinois Institute of Technology, Chicago, IL 60616, USA.\protect\\
E-mail: \{xniu7, rkapoor7\}@hawk.iit.edu, bglavic@iit.edu.
\IEEEcompsocthanksitem D. Gawlick, Z.H.Liu and V. Krishnaswamy, Oracle, Rewood City, CA 94065, USA.\protect\\
E-mail: \{dieter.gawlick, zhen.liu, vasudha.krishnaswamy\}@oracle.com.

\IEEEcompsocthanksitem V. Radhakrishnan, YugoByte, Sunnyvale, CA 94085
, USA.\protect\\ 
E-mail: venkatesh@yugabyte.com.
}
}

\IEEEtitleabstractindextext{
  \begin{abstract}
A well-established technique for capturing database provenance as annotations on data is to \emph{instrument} queries to propagate such annotations. However, even sophisticated query optimizers often fail to produce efficient execution plans for instrumented queries. 
We develop provenance-aware 
optimization techniques to address this problem.
Specifically, we study algebraic equivalences targeted at instrumented queries 
and alternative ways of instrumenting queries for provenance capture.
Furthermore, we present an extensible heuristic and cost-based optimization framework utilizing these 
optimizations. Our experiments confirm that these optimizations are highly effective, improving performance by several orders of magnitude for diverse provenance tasks.
\end{abstract}

}

\maketitle

\section{Introduction}\label{sec:intro}
Database provenance, information about the origin of data and the queries and/or updates that produced it, is critical for debugging queries, auditing, establishing trust in data, and many other use cases.
The de facto standard for database provenance~\cite{KG12,GA12} is to model provenance as annotations on data and define a query semantics that determines how annotations propagate. 
Under such a semantics, each output tuple $t$ of a query $Q$ is annotated with its provenance, i.e.,
a combination  of input tuple annotations that explains how these inputs were used by $Q$ to derive $t$.

Database provenance systems such as  Perm~\cite{glavic2013using}, GProM~\cite{AF18}, DBNotes~\cite{bhagwat2005annotation}, LogicBlox~\cite{GA12}, 
declarative Datalog debugging~\cite{KL12},   ExSPAN~\cite{ZS10}, and many others use a relational encoding of provenance annotations. These systems typically compile queries with annotated semantics into relational queries
that produce this encoding of provenance annotations following the process outlined in Fig.~\ref{fig:general-rewrite-approach}. We refer to this reduction from annotated to standard 
relational semantics as \textit{provenance instrumentation} or \textit{instrumentation} for short.
The example below introduces a relational encoding of provenance polynomials~\cite{KG12} and the instrumentation approach for this model implemented in Perm~\cite{glavic2013using}.

\begin{figure}[t]
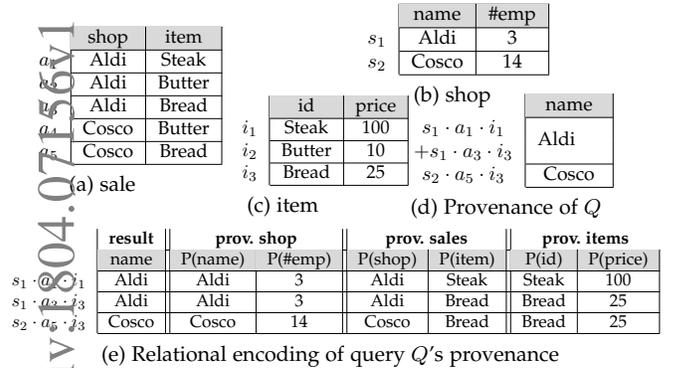

\centering
\begin{minipage}{0.30\linewidth}
  \centering

  \subfloat[sale]{
    \centering
          \resizebox{0.8\linewidth}{!}{
            \begin{minipage}{1.0\linewidth}
              \centering
    \begin{tabular}{c|c|c|} \cline{2-3} 
     & \chead shop& \chead item\\ \cline{2-3}
     $a_1$ & Aldi & Steak\\ \cline{2-3}
     $a_2$ & Aldi & Butter\\ \cline{2-3}
     $a_3$ & Aldi & Bread\\ \cline{2-3}
     $a_4$ & Cosco & Butter \\ \cline{2-3} 
     $a_5$ & Cosco & Bread \\ \cline{2-3}
    \end{tabular}
  \end{minipage}
  }
  }
  
\end{minipage}
\begin{minipage}{0.65\linewidth}
  \centering
  \hspace{0.1\linewidth}
  \begin{minipage}{0.8\linewidth}
    \centering
    \subfloat[shop]{
      \centering
      \resizebox{0.8\linewidth}{!}{
        \begin{minipage}{1.0\linewidth}
          \centering
    \begin{tabular}{c|c|c|} \cline{2-3}
 & \chead name& \chead \#emp\\\cline{2-3}
   $s_1$ & Aldi & 3\\ \cline{2-3}
   $s_2$ & Cosco & 14 \\ \cline{2-3}
    \end{tabular}
  \end{minipage}
  }
  }
\end{minipage}\\[-2mm]

\begin{minipage}{0.33\linewidth}
  \centering
\subfloat[item]
{
  \centering
            \resizebox{0.8\linewidth}{!}{
  \begin{minipage}{1.0\linewidth}
        \centering
    \begin{tabular}{c|c|c|} \cline{2-3} 
 & \chead id& \chead price\\ \cline{2-3}
     $i_1$ & Steak & 100\\ \cline{2-3} 
     $i_2$ & Butter & 10\\ \cline{2-3}
     $i_3$ & Bread & 25\\ \cline{2-3}
    \end{tabular}
  \end{minipage}
  }
  }
\end{minipage}\hspace{0.1cm}
\begin{minipage}{0.63\linewidth}
  \centering
  \subfloat[Provenance of $Q$]{\label{fig:Example-database-result}
    \centering
    \resizebox{0.8\linewidth}{!}{
      \begin{minipage}{1.0\linewidth}
        \centering
    \begin{tabular}{c|c|} \cline{2-2}
 &\chead name\\ \cline{2-2} 
      $s_1 \cdot a_1 \cdot i_1 $  & \multirow{2}{1.1cm}{Aldi} \\
      $+ s_1 \cdot a_3 \cdot i_3$ &\\
      \cline{2-2}
  $s_2 \cdot a_5 \cdot i_3$& Cosco \\  \cline{2-2}
    \end{tabular}
  \end{minipage}
  }
  }
\end{minipage}
\end{minipage}\\[1mm]

\subfloat[Relational encoding of query $Q$'s provenance ]{\label{fig:provenance-result-example-database}
\resizebox{1\linewidth}{!}{
  \begin{tabular}{c|c||c|c||c|c||c|c|} 
&\multicolumn{1}{c||}{\bf result} & \multicolumn{2}{c||}{\bf prov. shop} & \multicolumn{2}{c||}{\bf prov. sales} & \multicolumn{2}{c}{\bf prov. items}\\ \cline{2-8} 
 & \chead  name& \chead P(name)& \chead P(\#emp)& \chead P(shop)& \chead P(item)& \chead P(id)& \chead P(price)\\ \cline{2-8} 
 $s_1 \cdot a_1 \cdot i_1$  &  Aldi  & Aldi & 3 & Aldi & Steak & Steak & 100\\ \cline{2-8} 
$s_1 \cdot a_3 \cdot i_3 $&  Aldi  & Aldi & 3 & Aldi & Bread & Bread & 25\\ \cline{2-8} 
 $s_2 \cdot a_5 \cdot i_3$&  Cosco  & Cosco & 14 & Cosco & Bread & Bread & 25\\ \cline{2-8} 
  \end{tabular}
}
}\\[-2mm]

\caption{Provenance annotations and relational encoding}
\label{fig:Example-database}

\end{figure}
 
\begin{Example}\label{ex:simple-prov-ex}
Consider a query  
over the database in Fig.~\ref{fig:Example-database} returning shops that sell items which cost more than \$20:
\\[-4mm]
$$
\projection_ {name} (shop \join_{name=shop} sale \join_{item=id} \selection_{price > 20}(item))
$$\\[-6mm]
The query's result is shown in Fig.~\ref{fig:Example-database-result}. 
Using provenance polynomials
to represent provenance, tuples in the database are annotated with  variables
encoding tuple identifiers (shown to the left of each tuple). Each query result is annotated with
a polynomial 
that explains how the tuple was derived by combining input
tuples. Here, addition corresponds to alternative use of
tuples (e.g., union) and multiplication to
conjunctive use (e.g., a join). For example, the tuple \emph{(Aldi)} is derived by joining tuples $s_1$, $a_1$, and $i_1$ ($s_1 \cdot a_1 \cdot i_1$) or
alternatively by joining tuples $s_1$, $a_3$, and $i_3$.
Fig.~\ref{fig:provenance-result-example-database}
shows a relational encoding of these annotations as supported by the
Perm~\cite{glavic2013using} and GProM~\cite{AF18} systems: 
variables are represented by the tuple
they are annotating, multiplication is represented by concatenating the encoding
of the factors, and addition is represented by encoding
each summand as a separate tuple (see~\cite{glavic2013using}). 
This encoding is computed by compiling the input query with annotated semantics into  
 relational algebra. The resulting \emph{instrumented} query is shown below. It adds
input relation attributes to the final projection and renames them (represented as $\to$)
to denote that they store provenance.\\[-2mm]  
\resizebox{1\linewidth}{!}{
  \begin{minipage}{1.0\linewidth}
\begin{align*}
Q_{join} &= shop \join_{name=shop} sale \join_{item=id} \selection_{price > 20}(item)\\
Q&= \projection_{name,name \to P(name),numEmp \to P(numEmp), \ldots} (Q_{join})
\end{align*}
\end{minipage}
}\\[1mm]
The instrumentation we are using here is defined for any SPJ (Select-Project-Join) query (and beyond) based on a set of algebraic rewrite rules (see~\cite{glavic2013using} for details).
\end{Example}

The present paper extends~\cite{XN17}. Additional details are presented in the  appendix.

\subsection{Instrumentation Pipelines}
\label{sec:instr-pipel}

In this work, we focus on optimizing instrumentation pipelines such as the one from Example~\ref{ex:simple-prov-ex}. These pipelines divide the compilation of a frontend language to a target language into multiple compilation steps using one or more intermediate languages. We now introduce a subset of the pipelines supported by our approach to illustrate the breadth of applications supported by instrumentation. Our approach can be applied to any data management task that can be expressed as instrumentation. Notably, our implementation already supports additional pipelines, e.g., for summarizing provenance and managing uncertainty.

\parttitle{L1. Provenance for SQL Queries}
The pipeline from Fig.~\ref{fig:general-rewrite-approach} 
is applied by many provenance systems, e.g., DBNotes~\cite{bhagwat2005annotation} uses L1 to compute 
Where-provenance~\cite{CC09}.

\parttitle{L2. Provenance for Transactions}
Fig.~\ref{fig:trans-rewrite-approach} shows a pipeline that retroactively captures provenance for transactions~\cite{AG17}. In addition to the steps from Fig.~\ref{fig:general-rewrite-approach}, this pipeline uses a compilation step called \textit{reenactment}. Reenactment translates transactional histories with annotated semantics into equivalent temporal queries with annotated semantics. 

\parttitle{L3. Provenance for Datalog}
This pipeline (Fig.~\ref{fig:DL-rewrite-approach}) produces provenance graphs that explain 
which successful and failed rule derivations of an input Datalog program are relevant for (not) deriving a (missing) query result tuple of interest~\cite{LS16}. 
A provenance request is compiled into a Datalog program that computes the edge relation of the provenance graph. This program is then translated into SQL.

\parttitle{L4. Provenance Export}
This pipeline (Fig.~\ref{fig:expor-rewrite-approach} in Appendix~\ref{sec:pipelines})~\cite{NX15} is an extension of L1 which translates the relational provenance encoding produced by L1 into PROV-JSON, the JSON serialization of the PROV provenance exchange format.
This method~\cite{NX15} adds additional instrumentation on top of a query instrumented for provenance capture  to construct a single PROV-JSON document representing the full provenance of the query. The result of L4 is an SQL query that computes this JSON document.  

\parttitle{L5. Factorized Provenance}
L5 (Fig.~\ref{fig:xml-rewrite-approach} in Appendix~\ref{sec:pipelines}) captures provenance for queries. In contrast to L1, it represents the provenance polynomial of a query result as an XML document. The nested representation of provenance produced by the pipeline is factorized based on the structure of the query. The compilation target of this pipeline is SQL/XML. The generated query directly computes this factorized representation of provenance.

\parttitle{L6. Sequenced Temporal Queries}
This pipeline (Fig.~\ref{fig:temporal-rewrite-approach} in Appendix~\ref{sec:pipelines}) translates temporal queries with sequenced semantics~\cite{DBLP:reference/db/BohlenJ09} into SQL queries over an interval encoding of temporal data. A non-temporal query evaluated over a temporal database under sequenced semantics returns a temporal relation that records how the query result changes over time (e.g., how an employee's salary changes over time). Pipeline L6 demonstrates the use of instrumentation beyond provenance. We describe Pipelines L5 and L6 in more detail in Appendix~\ref{sec:pipel-l5-fact} and~\ref{sec:l6.-sequ-temp}.

\subsection{Performance Bottlenecks of Instrumentation}
\label{sec:motivation}
While instrumentation enables diverse 
provenance features to be implemented on top of DBMS,  
the performance of instrumented queries is often suboptimal. Based on our extensive experience with instrumentation systems~\cite{LS16,NX15,AF18,AG17,glavic2013using} and a preliminary evaluation we have identified bad plan choices by the DBMS backend as a major   
bottleneck.
Since query optimizers have to trade optimization time for query performance, optimizations that do not benefit common workloads are typically not considered.
Thus, most optimizers are incapable of simplifying instrumented queries, will not explore relevant parts of the plan space, or will spend excessive time on optimization. 
We now give an overview of problems we have encountered.

\begin{figure}[t]
  \centering
\subfloat[
Provenance is captured using an annotated version of relational algebra which is first translated into relational algebra over a relational encoding of annotated relations and then into SQL code. 
]{  \label{fig:general-rewrite-approach}
  \begin{minipage}{0.96\linewidth}
  \centering
\includegraphics[width=1\columnwidth]{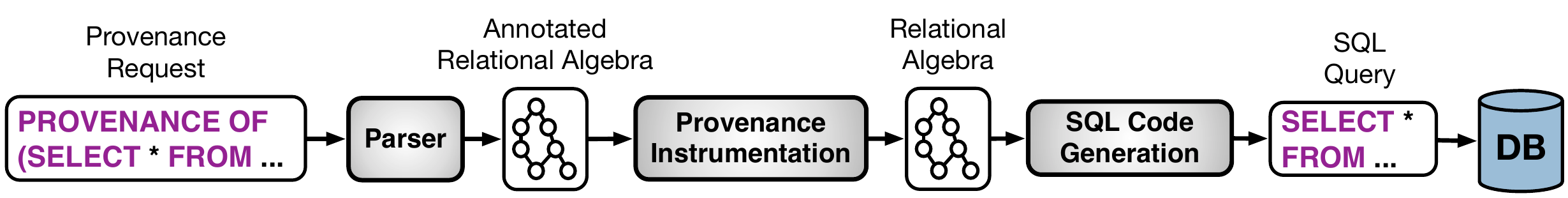}
\end{minipage}
}\\[-0.5mm]
\subfloat[
In addition to the steps of \textbf{(a)}, this pipeline 
uses \emph{reenactment}~\cite{AG17} to compile annotated updates into annotated queries. 
]{\label{fig:trans-rewrite-approach}
  \begin{minipage}{0.98\linewidth}
  \centering
  \includegraphics[width=0.9\columnwidth]{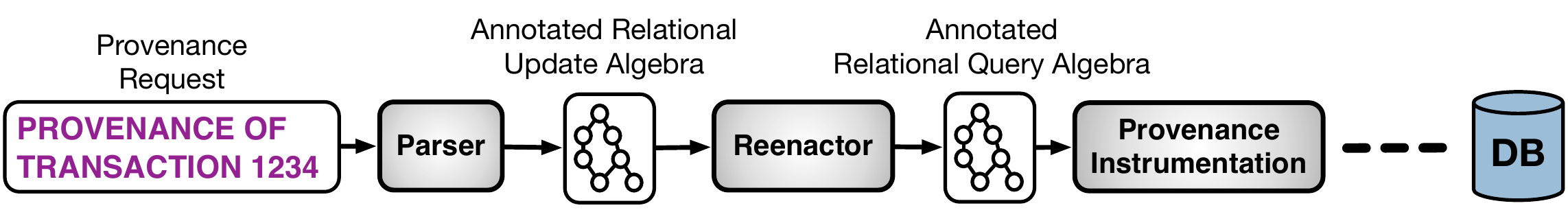}
\end{minipage}
}\\[-0.5mm]
\subfloat[Computing provenance graphs for Datalog queries~\cite{LS16} based on a rewriting called \emph{firing rules}. The instrumented Datalog program is first compiled into relational algebra and then into SQL.]{  \label{fig:DL-rewrite-approach}
\begin{minipage}{0.98\linewidth}
  \centering
\includegraphics[width=1\columnwidth]{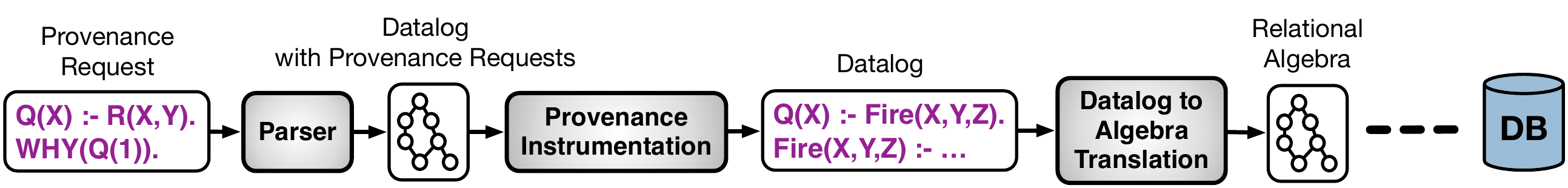}
\end{minipage}  
}\\[-2.2mm]
\caption{Instrumentation: \textbf{(a)} SQL, \textbf{(b)} transactions, \textbf{(c)} Datalog} 

\end{figure}

\parttitle{P1. Blow-up in Expression Size} The instrumentation for transaction provenance~\cite{AG17} shown in Fig.~\ref{fig:trans-rewrite-approach} may produce queries with a large number of query blocks. This can lead to long optimization times in systems that unconditionally pull-up subqueries (such as Postgres) because the subquery pull-up results in \lstinline!SELECT! clause expressions of size exponential in the number of stacked query blocks. While advanced optimizers do not apply this transformation unconditionally, they will at least consider it leading to the same blow-up in expression size during optimization. 

\parttitle{P2. Common Subexpressions} 
Pipeline L3~\cite{LS16} (Fig.~\ref{fig:DL-rewrite-approach}) instruments the input Datalog program to capture rule derivations. Compiling such queries into relational algebra leads to queries with many common subexpressions and 
duplicate elimination operators. Pipeline L4 constructs the PROV output using multiple projections over an instrumented subquery that captures provenance. The large number of common subexpressions in both cases  may significantly increase optimization time. Furthermore, if subexpressions are not reused then this significantly increases the query size. The choice of when to remove duplicates significantly impacts performance for Datalog queries.

\parttitle{P3. Blocking Join Reordering} Provenance instrumentation 
in 
GProM~\cite{AF18} is based on rewrite rules. 
For instance, provenance annotations are propagated through an aggregation by joining the aggregation with the provenance instrumented version of the aggregation's input on the group-by attributes. Such transformations increase  query size  and lead to 
interleaving of joins with operators such as aggregation.  
This interleaving may block optimizers from reordering joins leading to suboptimal join orders. 

\parttitle{P4. Redundant Computations} To capture provenance, systems such as Perm~\cite{glavic2013using} instrument a query one operator at a time using operator-specific rewrite rules. To apply operator-specific rules to rewrite a complex query, the rules have to be generic enough to be applicable no matter how operators are combined. 
This can lead to redundant computations, e.g., an instrumented operator generates a new column that is not needed by 
downstream operators.

 \section{Solution Overview}
\label{sec:solut-overv-contr}

While optimization has been recognized  as an important problem in provenance management, previous work has almost exclusively focused on how to compress provenance to reduce  storage cost, e.g., see~\cite{AB09,CJ08a,wu2013subzero}.
We 
study the orthogonal problem of \textbf{improving the performance of instrumented queries} that capture provenance.
Specifically, we develop heuristic and cost-based optimization techniques to address  the performance bottlenecks of instrumentation.

An important advantage of our approach is that it applies to any database backend and instrumentation pipeline. 
New transformation rules and cost-based choices can be added with ease.
When optimizing a pipeline, 
we can either target one of its intermediate languages or the compilation steps.
As an example for the first type of optimization, consider a compilation step that outputs relational algebra. 
We can optimize the generated algebra expression using algebraic equivalences before passing it on to the next stage of the pipeline.
For the second type of optimization consider the compilation step from pipeline L1 that translates annotated relational algebra (with provenance) into relational algebra.
If we know two equivalent ways of translating an algebra operator with annotated semantics into standard relational algebra, then we can optimize this step by choosing the translation
that maximizes performance. We study both types of optimization. For the first type, we focus on relational algebra since it is an intermediate language used in all of the pipelines from Sec.~\ref{sec:instr-pipel}. We investigate algebraic equivalences that are beneficial for instrumentation, but which are usually not applied by database optimizers. We call this type of optimizations \textit{provenance-specific algebraic transformations} (\textit{PATs}). We refer to optimizations of the second type as \textit{instrumentation choices} (ICs).

\begin{figure}[t]
\centering
\includegraphics[width=1\columnwidth]{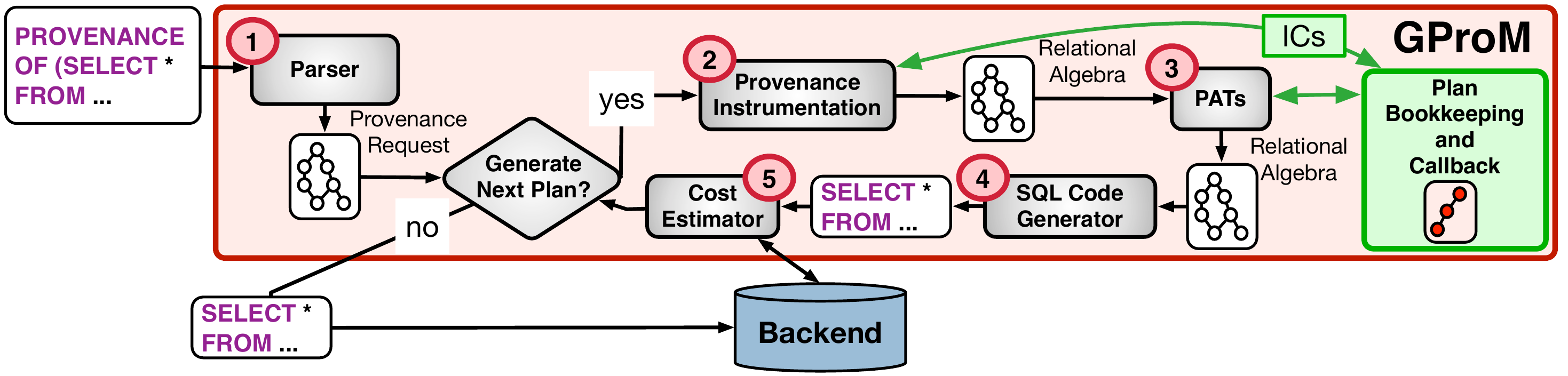}
$,$\\[-7mm]
\caption{GProM with Cost-based Optimizer}
\label{fig:cbo-arch}
\end{figure}

\parttitle{PATs}
We identify algebraic equivalences which 
are effective for speeding up provenance computations. For instance, 
we factor references to attributes to enable merging of projections without blow-off in expression size, pull up projections that 
create provenance annotations, and 
remove unnecessary duplicate elimination and window operators.
We infer local and non-local properties~\cite{grust2010let} such as candidate keys  for the algebra operators of a query. This enables us to define transformations that rely on non-local information.

\parttitle{ICs}
We introduce two ways for instrumenting an aggregation for provenance capture: 1) using a \textit{join}~\cite{glavic2013using} 
to combine the aggregation with the provenance of the aggregation's input; 2) using \textit{window} functions (SQL \lstinline!OVER! clause) to directly compute the aggregation functions over inputs annotated with provenance. 
We also present two ways for pruning tuples that are not in the provenance early-on when computing the provenance of a transaction~\cite{AG17}. Furthermore, we present two options for normalizing the output of a sequenced temporal query (L6). 

Note that virtually all pipelines that we support use relational algebra as an intermediate language. Thus, our PATs are more generally applicable than the ICs which target a compilation step that only is used in some pipelines. This is however an artifact of the pipelines we have chosen. In principle, one could envision ICs that are applied to a compilation step that is common to many pipelines.

\parttitle{CBO for Instrumentation}
Some  PATs are not always beneficial and for some ICs there is no clearly superior choice. Thus, there is a need for \textit{cost-based optimization} (CBO).
Our second contribution is 
a 
CBO framework for instrumentation pipelines that can be applied to any such pipeline no matter what compilation steps and  intermediate languages are used. This is made possible by decoupling the plan space exploration from actual plan generation.
Our optimizer treats the instrumentation pipeline as a blackbox function which it calls repeatedly  to  produce SQL queries (\textit{plans}). Each such plan is sent to the backend database for 
planning and 
cost estimation. 
We refer to one execution of the pipeline as an \textit{iteration}.
It is the responsibility of the pipeline's components to signal to the optimizer the existence of optimization choices (called \textit{choice points}) through the optimizer's \textit{callback API}. The optimizer responds to a call from one of these components by instructing it which of the available \textit{options} to choose. 
We keep track of which choices had to be made, which options exist for each choice point, and which options were chosen. This information is sufficient to iteratively enumerate the plan space by making different choices during each iteration.
Our approach provides great flexibility in terms of supported optimization
decisions, e.g., we can choose whether to apply a PAT or select which ICs to use.
Adding an optimization choice only requires adding a few lines of code (LOC) to the pipeline to inform the optimizer about the availability of options. 
To the best of our knowledge our framework is the first CBO that is \textbf{plan space and query language agonistic}.
Costing a plan (SQL query) requires us to use the DBMS to optimize a query which can be expensive. Thus, we may not be able to explore the full plan space. 
In addition to meta-heuristics, we also support a strategy that balances optimization vs. execution time.

We have implemented these optimizations in GProM~\cite{AF18}, our  provenance middleware that supports multiple DBMS backends and all the instrumentation pipelines discussed in Sec.~\ref{sec:instr-pipel}.
GProM is available as open source (\url{https://github.com/IITDBGroup/gprom}). Using L1 as an example,  Fig.~\ref{fig:cbo-arch} shows  how ICs, PATs, and CBO are integrated into the system. 
We demonstrate experimentally that our optimizations 
improve performance by over 4 orders of magnitude on average compared to unoptimized instrumented queries.
Our approach peacefully coexists with the DBMS optimizer. We use the DBMS optimizer 
where it is effective (e.g., join reordering) and use our optimizer to address the database's shortcomings with respect to provenance computations.

\section{Background and Notation}\label{sec:background}

A relation schema ${\bf R}(a_1, \ldots, a_n)$ consists of a name (${\bf R}$) and a list of attribute names $a_1$ to $a_n$. The arity of a schema is the number of attributes in the schema.
We use the bag semantics version of the relational model.
Let $\cal U$ be a domain of values. An instance $R$ of an n-ary schema ${\bf R}$  is a function $\aDom^n \to \mathbb{N}$ mapping tuples to their multiplicity. Here $R(t)$ denotes applying the function that is $R$ to input $t$, i.e., the multiplicity of tuple $t$
in relation $R$.
We require that relations have finite support $\supp(R) = \{ t \mid R(t) \neq 0\}$.  We use $t^m \in R$ to denote that tuple $t$ occurs with multiplicity $m$, i.e., $R(t) = m$ and $t \in R$ to denote that $t \in \supp(R)$.
An n-ary relation $R$ is contained in a relation $S$, written as $R \subseteq S$, iff $\forall t \in {\cal U}^n: R(t) \leq S(t)$, i.e., each tuple in $R$ appears in $S$ with the same or higher multiplicity.

\begin{table}
  \centering
 \begin{tabular}{|p{1.3cm}|p{6cm}|} \hline 
\rowcolor[gray]{.9}  Operator & Definition\\ \hline 
  $\selection$ & $\selection _\theta (R) = \{ t^n|t^n \in R \wedge t \models \theta  \}$ \\ \hline
  $\projection$ & $\projection_A (R) = \{t^n|n = \sum_{u.A = t} R(u) \} $ \\ \hline
  $\union$ & $R \union S = \{ t^{n+m}|t^n \in R \wedge t^m \in S \} $\\ \hline
  $\intersection$ & $R \intersection S = \{ t^{min(n,m)}|t^n \in R \wedge t^m \in S \} $\\ \hline
  $\difference$ & $R-S = \{ t^{max(n-m,0)}|t^n \in R \wedge t^m \in S \}$ \\ \hline
  $\crossprod$ & $R \crossprod S = \{ (t,s)^{n*m} |t^n \in R \wedge s^m \in S \} $ \\ \hline
  $\aggregation$ & $_{G}\aggregation_{f(a)} (R) = \{ (t.G, f(G_t))^1|t \in R \} $ \\ 
                 & $G_t = \{ (t_1.a)^n |{t_1}^n \in R \wedge t_1.G = t.G \} $ \\ \hline
  $\duplicate$ & $\duplicate (R) = \{ t^{1} |t  \in R \} $ \\ \hline
$\omega$ & $\omega_{f(a) \to x, G\|O}(R) \equiv \{ (t,f(P_t))^n | t^n \in R \}  $\\ 
                 & $P_t = \{ (t_1.a)^n |{t_1}^n \in R \wedge t_1.G = t.G \wedge t_1 \leq_O t \} $ \\ \hline
 \end{tabular}\\[-2mm]
 \caption{Relational algebra operators}
\label{tab:rel-algebra-def}
\end{table}

Table~\ref{tab:rel-algebra-def} shows the definition of the bag semantics version of relational algebra we use in this work. We use  $\schema{Q}$ to denote the schema of the result of query $Q$ and $Q(I)$ to denote the result of evaluating query $Q$ over database instance $I$.
Selection $\selection _\theta (R)$ returns all tuples from relation $R$ which satisfy the condition $\theta$. 
Projection $\projection _A (R)$ projects all input tuples on a list of  projection expressions. Here, $A$ denotes a list of expressions with potential renaming (denoted by $e \rightarrow a$) and $t.A$ denotes applying these expressions to a tuple $t$. The syntax of projection expressions is defined by the grammar shown below where $\constExpr$ denotes the set of constants, $\attrExpr$ denotes attributes, $\cExpr$ defines conditions, and $\vExpr$ defines projection expressions. 
\begin{align*}
  \vExpr &\defas \vExpr + \vExpr \mid \vExpr \cdot \vExpr \mid \constExpr \mid \attrExpr \mid \eIf{\cExpr}{\vExpr}{\vExpr}\\
  \cExpr &\defas \vExpr\, \compExpr\, \vExpr \mid \cExpr \wedge \cExpr \mid \cExpr \vee \cExpr \mid \neg \cExpr\\
  \compExpr &\defas = \mid \neq \mid < \mid \leq \mid \geq \mid >
\end{align*}
For instance, a valid projection expression over schema $R(a,b)$ is $(a + b) \cdot 5$. The expression type $\eIf{\cExpr}{\vExpr}{\vExpr}$ is introduced to support conditional expressions similar to SQL's \lstinline!CASE!. The semantics of projection expressions is defined using a function $\exprEval(t,e)$ which returns the result of evaluating $e$ over $t$. In the following we will often use $t.e$ to denote $\exprEval(t,e)$. 
The definition of $\exprEval$ and an example for how to apply it are shown in Appendix~\ref{sec:supp-back-expr-eval}.

Union $R \union S$ returns the bag union of  tuples from relations $R$ and $S$. Intersection $R \intersection S$ 
returns the tuples which are both in relation $R$ and $S$. Difference $R-S$ returns the tuples in relation $R$ which are not in $S$.   These set operations are only defined for inputs of the same arity.
Aggregation $\Aggregation{G}{f(a)} (R)$ groups tuples according to their values in attributes $G$ and computes the aggregation function $f$ over the bag of values of attribute $a$ for each group. We also allow the attribute storing $f(a)$ to be named explicitly, e.g., $\Aggregation{G}{f(a) \to x}(R)$ renames $f(a)$ as $x$.
Duplicate removal $\duplicate (R)$ removes duplicates. $R \times S$ is the cross product for bags (input multiplicities are multiplied).
For convenience we also define join $R \Join _{\theta} S$ and natural join $R \Join S$ in the usual way. 
For each tuple $t$, the window operator $\Win{f(a)}{x}{G}{O}(R)$ returns $t$ with an additional attribute $x$ storing the result of the aggregation function $f$.
Function $f$ is applied over the window (bag of values from attribute $a$) generated by partitioning the input on  $G \subseteq \schema{R}$ and including only tuples which are smaller than $t$ wrt. their values in attributes $O \subseteq \schema{R}$ where $G \cap O = \emptyset$.  An example is shown in Appendix~\ref{sec:supp-win-op-ex}.  
We use the window operator to express a limited form of SQL's \lstinline!OVER! clause.

We represent algebra expressions as DAGs (Directed Acyclic Graph) to encode reuse of subexpressions. For instance, Figure~\ref{fig:ex-algebra-graph} shows an algebra graph (left) which reuses an expression $\selection_{c<5}(R)$ and the corresponding algebra tree (right). We assume that nodes are uniquely identified within such graphs and abusing notation will use operators to denote nodes in such graphs. 
We use $Q[Q_1 \leftarrow Q_2]$ to denote the result of substituting subexpression (subgraph) $Q_1$ with $Q_2$ in the algebra graph for query $Q$. Again, we assume some way of identifying subgraphs. 
We use $Q = op(Q')$ to denote that operator $op$ is the root of the algebra graph for query $Q$ and that subquery $Q'$ is the input to $op$. 

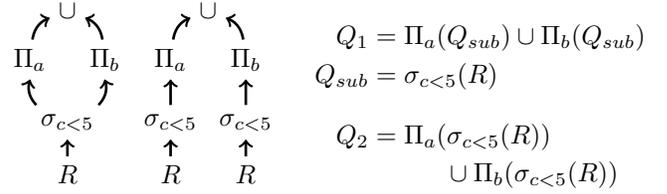
\begin{figure}[t]
\begin{center}
  \begin{minipage}{0.45\linewidth}
\begin{tikzpicture}
[op/.style={anchor=south},
conn/.style={->,line width=1pt}]

\node[op] (j) at (0,2.5) {$\union$};

\node[op] (p1) at (-0.5,1.8) {$\projection_{a}$};
\node[op] (p2) at (0.5,1.8) {$\projection_{b}$};
\node[op] (s) at (0,1) {$\selection_{c < 5}$};
\node[op] (r) at (0,0.3) {$R$};

\draw[conn,bend left] (p1) to (j);
\draw[conn,bend right] (p2) to (j);

\draw[conn,bend left] (s) to (p1);
\draw[conn,bend right] (s) to (p2);

\draw[conn] (r) to (s);
\end{tikzpicture}
\begin{tikzpicture}
[op/.style={anchor=south},
conn/.style={->,line width=1pt}]

\node[op] (j) at (0,2.5) {$\union$};

\node[op] (p1) at (-0.5,1.8) {$\projection_{a}$};
\node[op] (p2) at (0.5,1.8) {$\projection_{b}$};

\node[op] (s1) at (-0.5,1) {$\selection_{c < 5}$};
\node[op] (r1) at (-0.5,0.3) {$R$};

\node[op] (s2) at (0.5,1) {$\selection_{c < 5}$};
\node[op] (r2) at (0.5,0.3) {$R$};

\draw[conn,bend left] (p1) to (j);
\draw[conn,bend right] (p2) to (j);

\draw[conn] (s1) to (p1);
\draw[conn] (s2) to (p2);

\draw[conn] (r1) to (s1);
\draw[conn] (r2) to (s2);
\end{tikzpicture}
\end{minipage}
\begin{minipage}{0.5\linewidth}
  \begin{align*}
      \query_1 &= \projection_a(\query_{sub}) \union \projection_b(\query_{sub})\\
  \query_{sub} &= \selection_{c  < 5}(R)\\[3mm]
    \query_2 &= \projection_a(\selection_{c  < 5}(R))\\ &\hspace{1cm}\union \projection_b(\selection_{c  < 5}(R))
\end{align*} 
\end{minipage}
\end{center}
\caption{Algebra graph ($\query_1$, left), equivalent algebra tree ($\query_2$, middle), and corresponding algebra expressions (right)}
\label{fig:ex-algebra-graph}
\end{figure}

\section{Properties and Inference Rules}\label{sec:properties-inference}

We now discuss how to infer local and non-local properties of operators within the context of a query. 
Similar to  Grust et al.~\cite{grust2010let}, we 
use these properties in preconditions of algebraic rewrites (PATs). PATs are covered in Sec.~\ref{sec:heuristic}.

\subsection{Operator Properties}

\parttitle{keys} 
Property $keys$ is a set of super keys for an operator's output. 
For example, if $keys(R) = \{\{a\},\{b,c\}\}$ for a relation $R(a,b,c,d)$, then   the values of attributes $\{a\}$ and $\{b,c\}$ are unique in $R$.

\begin{Definition} \label{def:def_keys}
Let $Q$ be a query. A set 
$E \subseteq \schema{Q}$ is a \emph{super key} for $Q$ iff for every instance $I$ we have  $\forall t,t' \in Q(I): t.E = t'.E \rightarrow t=t'$ and $\forall t: Q(I)(t) \leq 1$. A super key is called a candidate key if it is minimal. 
\end{Definition}

Since we are using bag semantics, in the above definition we need to enforce that a relation with a superkey cannot contain duplicates. Recall that we defined bag relations as functions, thus, $Q(I)(t)$ denotes the multiplicity of $t$ in the result of $Q$ over $I$. Klug~\cite{K80} demonstrated that computing the set of functional dependencies that hold over the output of a query expressed in relational algebra is undecidable. The problem studied in \cite{K80} differs from  our setting in two aspects: 1) we only consider keys and not arbitrary functional dependencies and 2) we consider a more expressive algebra over bags which includes  generalized projection. As the reader might already expect, the undecidability of the problem caries over to our setting.

\begin{Theorem}\label{theo:keys-undecidable}
Computing the set of candidate keys for the output of a query $\query$ expressed in our bag algebra is undecidable. The problem stays undecidable even if $\query$ consists only of a single generalized projection, i.e., it is of the form $\query = \projection_A(R)$. 
\end{Theorem}
\begin{proof}
  We prove the theorem by a reduction from the undecidable problem of checking whether a multi-variant polynomial over the integers ($\mathbb{Z}$) is injective.  The undecidability of injectivity stems from the fact that this problem can be reduced to Hilbert's tenth  problem~\cite{matiyasevich1993hilbert} (does a Diophantine equation have a solution) which is known to be undecidable for integers. Given such a polynomial function $f(x_1, \ldots, x_n)$ over $\mathbb{Z}$, we define a schema $R(x_1, \ldots, x_n)$ over domain $\mathbb{Z}$ with a candidate key $X = \{x_1, \ldots, x_n\}$ and a query $\query_f = \projection_{f(x_1, \ldots, x_n) \to b}(R)$. Intuitively, the query computes the set of results of $f$ for the set of inputs stored as tuples in $R$.
  For instance, consider the multivariant polynomial $f(x,y) = x^2 + x \cdot y$. We would define an input relation $R$ with schema $\schema{R} = (x,y)$ and query $\query_f= \projection_{(x \cdot x + x \cdot y) \to b}(R)$ which computes $f$. 

  Now for sake of contradiction assume that we have a procedure that computes the set of candidate keys for a query based on keys given for the relations accessed by the query. The result schema of query $\query_f$ for polynomial $f$ consists of a single attribute $(b)$. Thus, it has either a candiate key $\{b\}$ or no candidate key at all. Since $X$ is a candidate key for $R$, $\{b\}$ is a candidate key iff $f$ is injective (we prove this equivalence below).
Thus, the hypothetical algorithm for computing the candidate keys of a query result relation gives us a decision procedure for $f$'s injectivivity. However, deciding whether $f$ is injective is undecidable and, thus, the problem of computing candidate keys for query results has to be undecidable.   

We still need to prove our claim that $\{b\}$ is a candidate key iff $f$ is injective.

\myproofpar{$\Rightarrow$}
For sake of contradiction assume that $\{b\}$ is a candidate key, but $f$ is not injective. Then there have to exist two inputs $I = (i_1, \ldots, i_n)$ and $J = (j_1, \ldots, j_n)$ with $I \neq J$ such that $f(I) = y$ and $f(J) = y$ for some value $y$. Now consider an instance of relation $R$ defined as $\{ I, J \}$. The result of evaluating query $\query_f$ over this instance is clearly $\{ (y)^2 \}$. That is, tuple $(y)$ appears twice in the result. However, this violates the assumption that $\{b\}$ is a candidate key.

\myproofpar{$\Leftarrow$}
For sake of contradiction assume that $f$ is injective, but $\{b\}$ is not a candidate key. Then there has to exists some instance of  $R$ such that $\query_f(R)$ contains a tuple $t$ with multiplicity $n > 1$. Since $X$ is a candiate key of $R$, we know that there are no duplicates in $R$. Thus, based on the definition of projection, the only way $t$ can appear with a multiplicity larger than one is if there are two inputs $t_1$ and $t_2$ in the input such that $f(t_1) = f(t_2)$ which contradicts the assumption that $f$ is injective. 
\end{proof}

Given this negative result, we will focus on computing a set of keys that is not necessarily complete nor is each key in this set guaranteed to be minimal. 
This is unproblematic, since we will only use the existence of keys as a precondition for PATs. That is, we may miss a chance of applying a transformation since our approach may not be able to determine that a key holds, but we will never incorrectly apply a transformation.

\parttitle{set}
Boolean property \emph{set} denotes 
whether the number of duplicates in the result of a subquery $\qSub$ of a query $\query$ is insubstantial for computing
$\query$. We model this condition using query equivalence, i.e., if we apply duplicate elimination to the result of $\qSub$, the resulting query is equivalent to the original query $\query$.   

\begin{Definition}\label{def:def_set}
  Let $\qSub$ be a subquery of a query $\query$. We say $\qSub$ is \emph{duplicate-insensitive} if
  $\query \equiv \query[\qSub \gets \duplicate(\qSub)]$.
\end{Definition}

The \emph{set} property is useful for introducing or removing duplicate elimination operators. However, as the following theorem shows, determining whether a subquery is duplicate-insensitive is undecidable. We, thus, opt for an approach that is sound, but not complete. 

\begin{Theorem}\label{theo:set-is-undecidable}
    Let $\qSub$ be a subquery of a query $\query$. The problem of deciding whether $\qSub$ is duplicate-insensitive is undecidable.
\end{Theorem}
\begin{proof}
See Appendix~\ref{sec:supp-properties-inference-rule}.
\end{proof}

\parttitle{ec}
The \textit{ec} property stores a set of equivalence classes (ECs) with respect to an equivalence relation $\aEquiv$ over attributes and constants. Let $a,b \in (\schema{\qSub} \union \aDom)$ for a subquery $\qSub$ of a query $\query$. We consider $a \aEquiv b$ if to evaluate $\query$ we only need tuples from the result of $\qSub$ where $a = b$ holds. 
We model this condition using query equivalence: if $a \aEquiv b$ for a subquery $\qSub$ of a query $Q$ then $Q \equiv Q[\qSub \gets \selection_{a=b}(\qSub)]$.

\begin{Definition} \label{def:def_ec}
  
  Let $\qSub$ be a subquery of query $\query$ and $a,b \in (\schema{\qSub} \cup \aDom)$. We say $a$ is equivalent to $b$, written as $a \aEquiv b$, if 
$Q \equiv Q[\qSub \leftarrow \selection_{a=b}(\qSub)]$.
  A set $E \subseteq (\schema{\qSub} \cup \aDom)$ is an \emph{equivalence class} (EC) for $\qSub$ if 
we have 
$\forall a,b \in E: a \aEquiv b$. An EC $E$ is maximal if no superset of $E$ is  an EC. 

\end{Definition}

As a basic sanity check we prove that $\aEquiv$ is in fact an equivalence relation. 

\begin{Lemma}
$\aEquiv$ is an equivalence relation.  
\end{Lemma}
\begin{proof}
See Appendix~\ref{sec:supp-properties-inference-rule}.
\end{proof}

Note that our definition of equivalence class differs from the standard definition of this concept. In fact, what is typically considered to be an equivalence class is what we call maximal equivalence class here.  
We consider non-maximal equivalence classes, 
because, as the following theorem shows, we cannot hope to find an algorithm that computes all equivalences that can be enforced for a query using generalized projection.

\begin{Theorem}\label{theo:eq-undecidable}
Let $\qSub$ be a subquery of a query $\query$ and $a,b \in \schema{\qSub}$. Determining whether $a \aEquiv b$ is undecidable.
\end{Theorem}
\begin{proof}
See Appendix~\ref{sec:supp-properties-inference-rule}.
\end{proof}

In the light of this undecidability result, we develop inference rules for property  $\ecProp$ (Section~\ref{sec:prop-inference}) that are sound, but not complete. That is, 
 all inferred equivalences hold, but there is no guarantee that we infer all equivalences that hold. Put differently, the equivalence classes computed using these rules may not be maximal.

\parttitle{\icolsProp{}}
This property records a set of attributes that are sufficient for evaluating the ancestors of an operator. By sufficient, we mean that if we remove other attributes this will not affect the result of the query.

\begin{Definition} \label{def:def_icols}
Let $Q$ be a query and $Q_{sub}$ be a subquery of $Q$, a 
set of attributes $E \subseteq \schema{Q_{sub}}$ is called  \emph{sufficient} in $Q_{sub}$ wrt. $Q$ if $Q \equiv Q[Q_{sub} \leftarrow \projection_{E}(Q_{sub})]$.
\end{Definition}

For example, attribute $d$ in $\projection_{a} (\projection_{a, b+c \to d}(R))$ is not needed to evaluate $\projection_a$. Note that there exists at least one trivial set of sufficient attributes for any query $\qSub$ which is $\schema{\qSub}$. Ideally, we would like to find sufficient attribute sets of minimal size to be able to reduce the tuple size of intermediate results and to remove operations that generate attributes that are not needed. Unfortunately, it is undecidable to determine a minimal sufficient  set of attributes.

\begin{Theorem}\label{theo:icols-undecidable}
Let $\qSub$ be a subquery of a query $\query$ and let $E \subset \schema{\qSub}$. The problem of determining whether $E$ is sufficient is undecidable.   
\end{Theorem}
\begin{proof}
See Appendix~\ref{sec:supp-properties-inference-rule}.
\end{proof}

The \icolsProp{} property we infer for an operator is guaranteed to be a sufficient set of attributes for the query rooted at this operator, but may not represent the smallest such set.

\subsection{Property Inference}\label{sec:prop-inference}
We infer properties for operators 
through 
traversals of the algebra graph of an input query. 
During a \textit{bottom-up traversal} the property $P$ for an operator $op$ is computed based on the values of $P$ for the operator's children. Conversely, during a \textit{top-down traversal} the property $P$ of an operator $op$ is initialized to a fixed value and is then updated based on the value of $P$ for one of the parents of $op$. We use $\curOp$ to denote the operator for which we are inferring a property (for bottom-up inference) or for a parent of this operator (for top-down inference). Thus, a top-down rule  $P(R) = P(R) \cup P(\curOp)$ has to be interpreted as update property $P$ for $R$ as the union of the current value of $P$ for $R$ and the current value of $P$ for operator $\curOp$ which is a parent of $R$.
We use $\rootOp$ to denote the root of a query graph. 
Because of space limitations we only show the inference rules of property $\setProp$ here (Table~\ref{tab:top-down-set}).
We show the inference rules for the remaining properties ($\ecProp{}$, $icols$ and $key$) in Appendix~\ref{sec:supp-inference-rule}.
In the following when to referring to properties such as the sufficient set of attributes of an operator we will implicitly understand this to refer to the property of the subquery rooted at this operator.
We prove these rules to be correct in Appendix~\ref{sec:supp-correct-proof}.
Here by correct we mean that $\keyProp(op)$ is a set of superkeys for $op$ which is not necessarily complete nor does it only contain candidate keys, $\ecProp(op)$ is a set of equivalence classes for $op$ which may not be maximal, if the $\setProp(op) = true$ than $op$ is duplicate-insensitive (but not necessarily vice versa), and finally $\icolsProp(op)$ is a sufficient set of attributes for  $op$.

\begin{table}
\centering
\begin{minipage}{1\linewidth}
\centering

  \resizebox{1\linewidth}{!}{
  \begin{minipage}{1.15\linewidth}
  \begin{tabular}{|c|c|l|} \hline 
\rowcolor[gray]{.9} Rule & Operator $\curOp$ & Inferred property \textit{\setProp{}} for the input(s) of $\curOp$\\ \hline 
  1,2&  $\circledast$ or $_{G}\aggregation _{F(a)}(R)$ & $\setProp{}(\circledast) = false$, $ \setProp{}(R) = false$
  
  \\ \hline
  
 3,4 & $\selection_{\theta}(R)$ or $\projection_{A}(R)$ & $ \setProp{}(R) = \setProp{}(R) \wedge \setProp{}(\curOp) $ 
  
  \\ \hline
                     
 5 &  $\duplicate(R)$ & $  \setProp{}(R) = \setProp{}(R) $ \\ \hline
 6-9 &  $R \join_ {a=b} S$ or $R \crossprod S$ or  & $ \setProp{}(R) = \setProp{}(R) \wedge \setProp{}(\curOp) $ \\ 
  & $R \union S$ or $R \intersection S$ or  $R \difference S$ & $ \setProp{}(S) = \setProp{}(S) \wedge \setProp{}(\curOp)  $ \\ \hline
10 & $R \difference S$ & $ \setProp{}(R) = false $ \\ 
     &                 & $\setProp{}(S) = false  $ \\ \hline
11 & $\omega_{f(a) \to x, G\|O}(R)$ & $\setProp{}(R) = false$ \\ \hline
  \end{tabular}  
\end{minipage}
}
\end{minipage}\\[-2mm]
\caption{Top-down inference of Boolean property \textit{\setProp{}}}
\label{tab:top-down-set}
\end{table}

\parttitle{Inferring the set Property}
We compute set in a top-down traversal  (Tab.~\ref{tab:top-down-set}). 
We initialize this property to \textsf{true} for all operators. As mentioned above our inference rules for this property are sound (if $\setProp(op) = true$ then the operator is duplicate-insensitive), but not complete.
We set  $\setProp(\rootOp)$ for the root operator ($\rootOp$)  to false (rule 1) since the final query result will differ if duplicates are eliminated from the output of $\rootOp$.  
Descendants of a duplicate elimination operator are duplicate-insensitive, because the duplicate elimination operator will remove any duplicates that they produce. The exception are descendants of operators such as aggregation and the window operator which may produce different result tuples if duplicates are removed. These conditions are implemented by the inference rules as follows: 1) $\setProp(op) = true$ if $op$ is the child of a duplicate elimination operator (Rule 5); 2) $\setProp(op) = false$ if $op$ is the child of a window, difference, or aggregation operator (Rules 11, 2, and 10); and otherwise 3) $\setProp(op)$ is true if $\setProp(\curOp)$ is true for all parents of the operator (Rules 1, 3, 4, 6-10).   

\begin{Example}\label{ex:set-op-inference}
Consider the algebra graph shown in Fig.~\ref{fig:inferring-set}. We show $\setProp$ for each operator as red annotations. For the root operator we set $\setProp(\union) = false$.  Since the root operator is a union, both children of the root inherit $\setProp(op) = false$.
We set $\setProp(\projection_b) = true$ since $\projection_b$ is a child of a duplicate elimination operator. This propagates to the child of this projection. The selection's set property is false, because even though it is below a duplicate elimination operator, it also has a parent for which $\setProp$ is false. Thus, the result of the query may be affected by eliminating duplicates from the result of the selection. Finally, operator $R$ inherits the set property from its parent which is a selection operator.

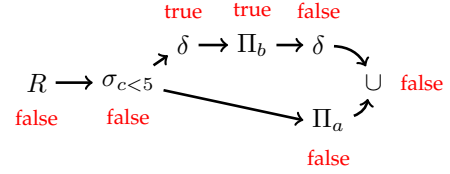
\begin{figure}
  \centering
  \begin{tikzpicture}
[op/.style={anchor=west},
pro/.style={red,font=\footnotesize},
conn/.style={->,line width=1pt}]

\node[op] (j) at (3.5,0) {$\union$};
\node[pro,right] (js) at (j.east) {false};

\node[op] (p1) at (2.8,-0.5) {$\projection_{a}$};
\node[pro,below] (p1s) at (p1.south) {false};

\node[op] (d) at (2.8,0.5) {$\duplicate$};
\node[pro,above] (ds) at (d.north) {false};

\node[op] (p2) at (1.8,0.5) {$\projection_{b}$};
\node[pro,above] (p2s) at (p2.north) {true};

\node[op] (d2) at (1,0.5) {$\duplicate$};
\node[pro,above] (d2s) at (d2.north) {true};

\node[op] (s) at (0,0) {$\selection_{c < 5}$};
\node[pro,below] (ss) at (s.south) {false};

\node[op] (r) at (-1,0) {$R$};
\node[pro,below] (rs) at (r.south) {false};

\draw[conn,bend right] (p1) to (j);
\draw[conn,bend left] (d) to (j);

\draw[conn] (p2) to (d);

\draw[conn] (d2) to (p2);

\draw[conn] (s) to (p1.west);
\draw[conn] (s) to (d2);

\draw[conn] (r) to (s);
\end{tikzpicture}
\caption{Inferring property $\setProp$}
\label{fig:inferring-set}
\end{figure}
\end{Example}

\begin{figure*}[t]
  \centering
\begin{minipage}{0.46\linewidth}
\begin{equation}\label{eq:pulling-up-provenance-projections}
  \frac{ a \subseteq \schema{\Diamond(\projection_{A}(R))} 
  }{\Diamond(\projection_{A,a \to b}(R)) \to  
     \projection_{\schema{\Diamond(\projection_{A}(R))},a \to b}(\Diamond(\projection_A (R)))}
\end{equation}
\end{minipage}
\hspace{1cm}
\begin{minipage}{0.15\linewidth}  
\begin{equation}\label{eq:duplicate-remove}
    \frac{keys(R) \neq \emptyset}{\duplicate (R) \rightarrow R} 
\end{equation}
\end{minipage}
\begin{minipage}{0.15\linewidth}  
\begin{equation}\label{eq:duplicate-remove-set}
    \frac{set(\duplicate(R))}{\duplicate (R) \rightarrow R} 
\end{equation}
\end{minipage}
\begin{minipage}{0.16\linewidth}  
\begin{equation}\label{eq:remove-redundant-columns1}
 \frac{A=icols(R)}{R \rightarrow \projection_A (R)}
\end{equation}
\end{minipage}
\hspace{1cm}
\begin{minipage}{0.41\linewidth}  
\begin{equation}\label{eq:add-duplicate-removal}
 \frac{G \subseteq \schema R}{\Aggregation{G}{}(R \join_ {b=c} S)  \rightarrow  \Aggregation{G}{}(\Aggregation{G,b}{} (R) \join_{b=c} S)}
\end{equation}
\end{minipage}
\hspace{1cm}
\begin{minipage}{0.45\linewidth}  
\begin{equation}\label{eq:attribute-factoring}
 \frac{e_1 = \eIf{\theta}{a + c}{a}}
{\projection_{e_1,...,e_m}(R) \to \projection_{ a +  \eIf{\theta}{c}{0}, e_2,...,e_m}(R)}
\end{equation}
\end{minipage}\\[1mm]
\begin{minipage}{0.28\linewidth}  
\begin{equation}\label{eq:window-function}
 \frac{x \not\in icols(\Win{f(a)}{x}{G}{O}(R))}{\Win{f(a)}{x}{G}{O}(R) \rightarrow R}
\end{equation}
\end{minipage}
\begin{minipage}{0.67\linewidth}  
\begin{equation}\label{eq:group-by-push-down}
 \frac{a \in \schema R \wedge a \not\in (G \union \{b, c\}) \wedge b \in G \wedge G \subseteq \schema R \wedge \{c\} \in keys(S)}{_{G} \aggregation _{f(a)}(R \join_ {b=c} S)  \rightarrow \Aggregation{G}{f(a)}(R) \join_{b=c} S}
\end{equation}
\end{minipage}

  \caption{Provenance-specific transformation (PAT) rules}
  \label{fig:algebraic-rules}
\end{figure*}

\section{PATs}\label{sec:heuristic}

We now introduce a subset of our PAT rules (Fig.~\ref{fig:algebraic-rules}), prove their correctness, and then discuss how these rules address the performance bottlenecks discussed in Sec.~\ref{sec:motivation}. A rule  $\frac{pre}{q \rightarrow q'}$ has to be read as ``If condition $pre$ holds, then $q$ can be rewritten as $q'$''.
Note that we also implement standard optimization rules such as selection move-around, and merging of adjacent projections, because these rules may help us to fulfill the preconditions of PATs (see  Appendix~\ref{sec:supp-heuristic}).

\parttitle{Provenance Projection Pull Up}
Provenance instrumentation~\cite{AF18,glavic2013using} 
seeds provenance annotations by duplicating attributes of input relations using projection. This increases the size of tuples in intermediate results. 
We can delay this duplication of attributes if the attribute we are replicating is still available in ancestors of the projection. 
In Rule~\eqref{eq:pulling-up-provenance-projections}, 
$b$ is an attribute storing provenance generated by duplicating attribute $a$. If $a$ is available in the schema of $\Diamond(\projection_{A}(R))$ ($\Diamond$ can be any operator) and $b$ is not needed to compute $\Diamond$,
then  we can pull the projection on $a \to b$ through operator $\Diamond$. 
For example, consider a query $Q = \selection_{a<5}(R)$ over relation $R(a,b)$.
Provenance instrumentation yields: 
$\selection_{a<5}(\projection_{a,b,a \to P(a), b \to P(b)}(R))$. This projection can be pulled up: 
$
\projection_{a,b,a \to P(a), b \to P(b)}(\selection_{a<5}(R))
$.

\parttitle{Remove Duplicate Elimination}
Rules~\eqref{eq:duplicate-remove} and~\eqref{eq:duplicate-remove-set} remove duplicate elimination operators. If a relation $R$  has at least one super key, then it cannot contain any duplicates. Thus, a duplicate elimination applied to  $R$ can be safely removed (Rule~\eqref{eq:duplicate-remove}). Furthermore, if the output of a duplicate elimination $op$ is again subjected to duplicate elimination further downstream and the operators on the path between these two operators are not sensitive to duplicates (property \textit{set} is true for $op$), then $op$ can be removed (Rule~\eqref{eq:duplicate-remove-set}).

\parttitle{Remove Redundant Attributes}
Recall that $icols(R)$ is a set of attributes from relation $R$ which is sufficient to evaluate ancestors of $R$.
If $icols(R) = A$, 
 then we use Rule~\eqref{eq:remove-redundant-columns1} to remove all other attributes by projecting $R$ on $A$.
Operator $\Win{f(a)}{x}{G}{O} (R)$ extends each tuple $t \in R$ by adding a new attribute $x$ that stores the aggregation function result $f(a)$.
Rule~\eqref{eq:window-function} removes $\win$ if $x$ is not needed by ancestors of $\Win{f(a)}{x}{G}{O}(R)$.

\parttitle{Attribute Factoring}\label{sec:PAT-factor-attrs}
Attribute factoring restructures projection expressions 
such 
that adjacent projections can be merged without blow-up in expression size. 
For instance, merging projections $\projection_{b+b+b \to c}(\projection_{a+a+a\to b}(R))$  increases the number of references to $a$ to 9 (each mention of $b$ is replaced with $a+a+a$). This blow-up can occur when computing the provenance of transactions where multiple levels of \lstinline!CASE! expressions are used. Recall that we represent  \lstinline!CASE! as $\eIf{\theta}{e_1}{e_2}$ in projection expressions. 
For example, update \lstinline!UPDATE R SET a = a + 2 WHERE b = 2! would be expressed as $\projection_{\eIf{b=2}{a+2}{a},b}(R)$ which can be rewritten as $\projection_{a + \eIf{b=2}{2}{0}, b}(R)$, reducing the references to $a$ by 1.  
We define analog rules for any arithmetic operation which has a neutral element (e.g.,  multiplication).

\parttitle{Aggregation Push Down}\label{sec:PAT-push-down-agg}
Pipeline L5 encodes the provenance (provenance polynomial) of a query result as an XML document. Each polynomial is factorized based on the structure of the query.
We can reduce the output's size by rewriting the query using algebraic equivalences to choose a beneficial factorization~\cite{OZ11}. For example, $a \cdot b + a \cdot c + a \cdot d$ can be factorized as $a \cdot (b+c+d)$. 
For queries with aggregation, this factorization can be realized by pushing aggregations through joins.  
Rule~(\ref{eq:add-duplicate-removal}) and (\ref{eq:group-by-push-down}) push down aggregations based on the equivalences introduced in~\cite{chaudhuri1994including}. 
Rule~\ref{eq:group-by-push-down} pushes an aggregation to a child of a join operator if the join is cardinality-preserving and all attributes needed to compute the aggregation are available in that child.  For instance, consider $\Aggregation{b}{f(a)}(R \join_ {b=c} S)$ where $\{c\}$ is a key of $S$. Since R is joined with S on $b=c$, pushing down the aggregation to R does not affect the cardinality of the aggregation's input. Since also  $\{a,b\} \in \schema R$,  we can rewrite this query into $\Aggregation{b}{f(a)}(R) \join_{b=c} S$.
Rule~\ref{eq:add-duplicate-removal} redundantly pushes an aggregation without aggregation functions (equivalent to a duplicate elimination) to create a pre-aggregation step.

\begin{Theorem}
The PATs from Fig.~\ref{fig:algebraic-rules} are equivalence preserving.  
\end{Theorem}
\begin{proof}
See Appendix~\ref{sec:supp-heuristic}. 
\end{proof}

\subsection{Addressing Bottlenecks through PATs}
\label{sec:rule-problem-address}

Rule~\eqref{eq:attribute-factoring} is a preprocessing step that helps us to avoid a blow-up in expression size when merging projections (Sec.~\ref{sec:motivation} \textbf{P1}).
Rules~\eqref{eq:duplicate-remove} and~\eqref{eq:duplicate-remove-set} can be used to remove unnecessary duplicate elimination operators (\textbf{P2}). 
Bottleneck \textbf{P3} is addressed by removing operators that block join reordering:  Rules~\eqref{eq:duplicate-remove}, \eqref{eq:duplicate-remove-set}, and~\eqref{eq:window-function} remove such operators. Even if such operators cannot be removed, Rules~\eqref{eq:pulling-up-provenance-projections} and~\eqref{eq:remove-redundant-columns1} remove attributes that are not needed which reduces the schema size of intermediate results. 
\textbf{P4} can be addressed by using Rules~\eqref{eq:duplicate-remove}, \eqref{eq:duplicate-remove-set}, and \eqref{eq:window-function} to remove redundant operators. Furthermore, Rule~\eqref{eq:remove-redundant-columns1} 
removes unnecessary columns. 
Rule~(\ref{eq:add-duplicate-removal}) and~(\ref{eq:group-by-push-down})
factorize nested representations of provenance (Pipeline L5) to reduce its size by pushing aggregations through joins.
In addition to the rules discussed so far, we apply standard equivalences, because our transformations often benefit from these equivalences and they also allow us to further simplify a query. For instance, we apply \textit{selection move-around} (which benefits from the $\ecProp$ property), merge selections and projections (only if this does not result in a significant increase in expression size), and remove redundant projections (projections on all input attributes). These additional PATs are discussed in Appendix~\ref{sec:supp-heuristic}.

 \section{Instrumentation Choices}
\label{sec:transf-appl-during}

\parttitle{Window vs. Join}\label{sec:window-vs-join}
The \textit{Join} method for instrumenting an aggregation operator for provenance capture was first used by Perm~\cite{glavic2013using}. To propagate provenance from the input of the aggregation to produce results annotated with provenance, the original aggregation is computed and then joined with the provenance of the aggregation's input on the group-by attributes. This will match the aggregation result for a group with the provenance of tuples in the input of the aggregation that belong to that group (see~\cite{glavic2013using} for details). For instance, consider a query $_{b}\aggregation_{sum(a) \to x}(R)$ with $\schema{R} = (a,b)$. This query would be rewritten into $\projection_{b,x,P(a),P(b)}(\Aggregation{G}{sum(a) \to x}(R) \join_{b = b'} \projection_{b\to b', a \to P(a), b \to P(b)}(R))$. Alternatively, the aggregation can be computed over the input with provenance using the window operator $\win$ by turning  
the group-by into a partition-by. The rewritten expression is $\projection_{b,x,P(a),P(b)}(\Win{sum(a)}{x}{b}{}(R)(\projection_{a,b, a \to P(a), b \to P(b)}(R)))$.
The \textit{Window} method has the advantage that 
no additional joins are introduced. However, as we will show in Sec.~\ref{sec:experiments}, the \textit{Join} method is superior in some cases and, thus, the choice between these alternatives should be cost-based.

\parttitle{FilterUpdated vs. HistJoin}
Our approach for capturing the provenance of a transaction $T$~\cite{AG17} only returns the provenance of tuples that were affected by $T$. We consider two alternatives for achieving this. 
The first method is called \textit{FilterUpdated}.  Consider a transaction $T$ with $n$ updates and let $\theta_i$ denote the condition (\lstinline!WHERE!-clause) of the $i^{th}$ update. Every tuple updated by the transaction has to fulfill at least one $\theta_i$. Thus, this set of tuples can be computed by applying a selection on condition $\theta_1 \vee \ldots \vee \theta_n$ to the input of reenactment. The alternative called \textit{HistJoin} uses time travel to determine based on the database version at transaction commit which tuples where updated by the transaction. It then joins this set of tuples with the version at transaction start to recover the original inputs of the transaction. For a detailed description see~\cite{AG17}.
\textit{FilterUpdated} is typically superior, because it avoids the join applied by \textit{HistJoin}. However, for transactions with a large number of operations, 
the cost of \textit{FilterUpdated}'s selection can be higher than the join's cost. 

\parttitle{Set-coalesce vs. Bag-coalesce}\label{sec:set-vs-bag}
The result of a sequenced temporal query~\cite{DBLP:reference/db/BohlenJ09} can be encoded in multiple, equivalent ways using intervals. Pipeline L6 applies a normalization step to ensure a unique encoding of the output. 
Coalescing~\cite{BS96}, the standard method for normalizing interval representations of temporal data under set semantics, is not applicable for bag semantics. 
We introduce a version that also works for bags.  However, this comes at the cost of additional overhead. If we know that a query's output does not contain any duplicates, then we can use the cheaper set-coalescing method. We use Property $\keyProp$ to determine whether should we can apply \textit{set-coalesce} (see Appendix~\ref{sec:l6.-sequ-temp}).

 \begin{algorithm}[t]
  \caption{CBO}
  \label{alg:cbo-skeleton}
  \begin{algorithmic}[1]
    \Procedure{CBO}{$Q$}
      \State $T_{best} \gets \infty$, $T_{opt} \gets 0.0$
      \While {$\Call{hasMorePlans}{ } \wedge \Call{continue}$}
        \State $t_{before} \gets \Call{currentTime}{ }$
        \State $P \gets \Call{generatePlan}{Q}$ 
        \State $T \gets \Call{getCost}{P}$
        \If {$T<T_{best}$}
          \State $T_{best} \gets T, P_{best} \gets P$
        \EndIf
        \State \Call{genNextIterChoices}{ }
        \State $T_{opt} = T_{opt} + (\Call{currentTime}{ } - t_{before})$
      \EndWhile
      \State \Return $P_{best}$
    \EndProcedure
        
  \end{algorithmic}
\end{algorithm}

\section{Cost-based Optimization}\label{sec:cbo}

Our CBO algorithm (Alg.~\ref{alg:cbo-skeleton}) consists of a main loop that is executed until the whole plan space has been explored (function $\Call{hasMorePlans}{}$) or until a stopping criterion has been reached (function $\Call{Continue}{}$). 
In each iteration, function $\Call{generatePlan}{}$ takes the output of the parser and runs it through the instrumentation pipeline (e.g, the one shown in Fig.~\ref{fig:cbo-arch}) to produce an SQL query. The pipeline components inform the optimizer about choice points using function $\Call{makeChoice}{}$. The resulting plan $P$ is then costed. 
If the cost $T$ of the current plan $P$ is less than the cost $T_{best}$ of the best plan found so far, then we set $P_{best} = P$. Finally, we decide which optimization choices to make in the next iteration using function $\Call{genNextIterChoices}{}$.
Our optimizer is plan space agnostic. New choices are discovered at runtime when a step in the pipeline informs the optimizer about an optimization choice. This enables the optimizer to enumerate all plans for a blackbox instrumentation pipeline.

\parttitle{Costing} Our default cost estimation implementation uses the DBMS to create an optimal execution plan for $P$ and estimate its cost. 
This ensures that we get the estimated cost for the plan that would be executed by the backend instead of estimating cost based on the properties of the query alone. 

\parttitle{Search Strategies}
Different strategies for exploring the plan space 
are implemented as different versions  of the $\Call{continue}{}$, $\Call{genNextIterChoices}{}$, and $\Call{makeChoice}{}$ functions. 
The default 
setting guarantees that the whole search space will be explored ($\Call{continue}{}$ returns true).

\subsection{Registering Optimization Choices}

We want to make the optimizer aware of choices available in a pipeline without having to significantly change existing code.
Choices are registered by calling the optimizer's $\Call{makeChoice}{}$ function.
This callback interface has two purposes: 1) inform the optimizer that a choice has to be made and how many alternatives to choose from and 2) allowing it to control which options are chosen.   
We refer to a point in the code where a choice is enforced as 
a \textit{choice point}. 
A choice point 
has a fixed 
number of \textit{options}. 
The return value of $\Call{makeChoice}{}$ instructs the caller to take a particular option.

\begin{Example}
  Assume we want to make a cost-based decision on whether to use the \emph{Join} or \emph{Window} method (Sec.~\ref{sec:transf-appl-during}) to instrument an aggregation.
We add a call $\Call{makeChoice}{2}$ to register a choice with two options to choose from. 
The optimizer responds 
with a number ($0$ or $1$) encoding the option to be chosen. 
$\,$\\[-5mm]
\begin{lstlisting}[language=c] 
if (makeChoice(2) == 0) Window(Q) else Join(Q)
\end{lstlisting}
$\,$\\[-10mm]
\end{Example}

A code fragment containing a call to $\Call{makeChoice}{}$ may be executed several times during one iteration.
Every call is treated as an independent choice point, e.g., 4 possible combinations of the \emph{Join} and \emph{Window} methods will be considered for instrumenting a query with two aggregations.

\begin{figure}[t]
  \centering
  $\,$\\[-7mm]
\resizebox{!}{0.3\columnwidth}{
\begin{tikzpicture}
[every node/.style={circle,draw,fill=blue!50,label position={east},inner sep={1mm}},
el/.style={draw=none,fill=none,inner sep={1mm}},
level distance=7mm,
level 1/.style={sibling distance=55mm},
level 2/.style={sibling distance=25mm},
level 3/.style={sibling distance=13mm}
]
  \node[label={[label distance=0.4cm]{\textbf{Window vs. Join}}}] {} 
        child {node[label=left:{\textbf{Reorder}}] {} 
            child {node[label={[label distance=-0.3cm]below:{\textcolor{red}{[0,0]}}}] {} 
              edge from parent node[left,el] {0}
            }
            child {node[label={[label distance=-0.3cm]below:{\textcolor{red}{[0,1]}}}] {} 
              edge from parent node[right,el] {1}
            }
            edge from parent node[left,el,inner sep={4mm}] {0}
        }
        child {node[label={\textbf{Reorder}}] {}
          child {node[label=left:{\textbf{Reorder}}] {} 
            child {node[label={[label distance=-0.3cm]below:{\textcolor{red}{[1,0,0]}}}] {} 
              edge from parent node[left,el] {0}
            }
            child {node[label={[label distance=-0.3cm]below:{\textcolor{red}{[1,0,1]}}}] {} 
              edge from parent node[right,el] {1}
            }
            edge from parent node[left,el] {0}
          }
          child {node[label=right:{\textbf{Reorder}}] {} 
            child {node[label={[label distance=-0.3cm]below:{\textcolor{red}{[1,1,0]}}}] {} 
              edge from parent node[left,el] {0}
            }
            child {node[label={[label distance=-0.3cm]below:{\textcolor{red}{[1,1,1]}}}] {} 
              edge from parent node[right,el] {1}
            }
            edge from parent node[right,el] {1}            
          }
          edge from parent node[right,el,inner sep={4mm}] {1}
        }
  ;
\end{tikzpicture}
}\\[-5mm]
\caption{Plan space tree example}
\label{fig:plan-tree-example}
\end{figure}
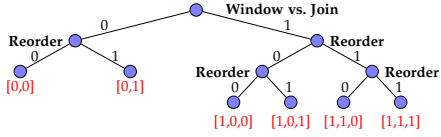

\subsection{Plan Enumeration}
\label{sec:search-space}

During one iteration we may hit any number of choice points and each choice made 
may affect what other choices have to be made in the remainder of this iteration. 
We use a data structure called \textit{plan tree} that models the plan space shape. In the plan tree each intermediate node represents a choice point, outgoing edges from a node are labelled with options and children represent choice points that are hit next. 
A path from the root of the tree to a leaf node represents a particular sequence of choices that results in the plan represented by this leaf node. 

\begin{Example}
  Assume we use two choice points: 1) Window vs. Join; 2) reordering join inputs. The second choice point can only be hit if a join operator exist, e.g., if we choose to use the \emph{Window} method then the resulting algebra expression may not have any joins and this choice point would never be hit.
Consider a query which is an aggregation over the result of a join.   
Fig.~\ref{fig:plan-tree-example} shows the corresponding plan tree. When instrumenting the aggregation, we have to decide whether to use the \textit{Window} (0) or the \textit{Join}
method (1). If we choose 
(0), then we have to decide wether to reorder the inputs of the join. If we 
choose (1), then there is an additional join for which we have to decide whether to reorder its input. 
The tree is asymmetric, i.e., the number of choices to be made in each iteration (path in the tree) is not constant. 
\end{Example}

While the plan space tree encodes all possible plans for a given query and set of choice points, it would not be feasible to materialize it, because its size can be exponential in the maximum number of choice points that are  hit during one iteration (the depth $d$ of the plan tree).  
Our default implementation of the $\Call{generateNextPlan}{}$ and $\Call{makeChoice}{}$ functions 
explores the whole plan space using ${\cal O}(d)$ space. 
As long as we know which path was taken in the previous iteration (represented as a list of choices as shown in Fig.~\ref{fig:plan-tree-example}) and for each node (choice point) on this path the number of available options, then we can determine what choices should be made in the next iteration to reach the leaf node (plan) immediately to the right of the previous iteration's plan. We call this traversal strategy \textit{sequential-leaf-traversal}. We have implemented an alternative strategy that approximates a binary search over the leaf nodes. We opt for an approximation, because the structure of subtrees is not known upfront. This strategy called \textit{binary-search-traversal} is described in more detail in Appendix~\ref{sec:binary-search-traversal}. The rationale for supporting this strategy is that if time constraints prevent us from exploring the full search space, then we would like to increase the diversity of explored plans by traversing different sections of the plan tree.

\begin{Theorem}
Let $Q$ be input query. Algorithm~\ref{alg:cbo-skeleton} iterates over all plans that can be created for the given choice points.
\end{Theorem}

\subsection{Alternative Search Strategies} \label{sec:traversal-strategies}

Metaheuristics
are applied in query optimization to deal with large search spaces. 
We discuss an implementation of a metaheuristic in our framework in Appendix~\ref{sec:supp-sty}.

\parttitle{Balancing Optimization vs. Runtime}\label{sec:balance-opt}
The strategies discussed  so far 
do not adapt the effort spend on optimization based on how expensive the query is. Obviously, spending more time on optimization than on execution is undesirable (assuming that provenance requests are ad hoc). Ideally, we would like to minimize the sum of the  optimization time ($T_{opt}$) and execution time of the best plan $T_{best}$ by stopping optimization once a cheap enough plan has been found. This is an online problem, i.e., after each iteration we have to decide whether to execute the current best plan or continue to produce more plans. 
The following stopping condition results in a 2-competitive algorithm, i.e.,  $T_{opt} + T_{best}$ is  less than 2 times the minimal achievable cost: stop optimization once $T_{best} = T_{opt}$. Note that even though we do not know the length of an iteration upfront, we can still ensure $T_{best} = T_{opt}$ by stopping mid iteration.

\begin{Theorem}
The algorithm outlined above is 2-competitive.
\end{Theorem}
\begin{proof}
See Appendix~\ref{sec:supp-sty}. 
\end{proof}

 \section{Related Work}\label{sec:relatedwork}

Our work is related to optimizations that sit on top of standard CBO, to compilation of non-relational languages into SQL, and to provenance capture and storage optimization.

\parttitle{Cost-based Query Transformation}
State-of-the-art DBMS apply transformations such as decorrelation of nested subqueries~\cite{SP96} 
in addition to (typically exhaustive) join enumeration and choice of physical operators. Often such transformations are integrated with CBO~\cite{AL06} by iteratively rewriting the input query through transformation rules and then finding the best plan for each rewritten query. 
Typically, metaheuristics (randomized search) are applied to deal with the large search space. Extensibility of query optimizers has been studied in, e.g.,~\cite{graefe1993volcano}. 
While our CBO framework 
is also applied on-top of standard database optimization, we 
 can  turn any choice (e.g., ICs) within an instrumentation pipeline into a cost-based decision. 
Furthermore, our framework 
has the advantage  
that new optimization choices can be added without modifying the optimizer.

\parttitle{Compilation of Non-relational Languages into SQL}\label{sec:rel-non-rel-to-rel}
Approaches that compile non-relational languages (e.g., XQuery~\cite{grust2010let,LK05}) or extensions of relational languages (e.g., temporal~\cite{SJ01} and nested collection models~\cite{CL14}) into SQL face similar challenges as we do.  
Grust et al.~\cite{grust2010let} optimize the compilation of XQuery into SQL. The approach heuristically applies algebraic transformations  to cluster join operations 
with the goal to produce an SQL query that can successfully be optimized by a relational database. 
We adopt their idea of inferring properties over algebra graphs.  
However, to the best of our knowledge we are the first to integrate these ideas with CBO and to consider ICs.

\parttitle{Provenance Instrumentation}
Several systems such as \textit{DBNotes}~\cite{bhagwat2005annotation}, \textit{Trio}~\cite{aggarwal2009trio}, 
\textit{Perm}~\cite{glavic2013using}, \textit{LogicBox}~\cite{GA12}, \textit{ExSPAN}~\cite{ZS10}, and \textit{GProM}~\cite{AF18} model provenance as annotations on data and capture provenance by propagating annotations. Most systems apply the \textit{provenance instrumentation} approach described in the introduction by compiling provenance capture and queries into a relational query language (typically SQL).
Thus, the techniques we introduce in this work are applicable to a wide range  of systems.

\parttitle{Optimizing Provenance Capture and Storage}
Optimization of provenance has mostly focused on minimizing the storage size of provenance. Chapman et al.~\cite{CJ08a} introduce several techniques for compressing provenance information, e.g., by replacing repeated elements with references and discuss how to maintain such a storage representation under updates. Similar techniques have been applied to reduce the storage size of provenance for workflows that exchange data as nested collections~\cite{AB09}. A cost-based framework for choosing between reference-based provenance storage 
and propagating full provenance 
was introduced in the context of declarative networking~\cite{ZS10}. 
This idea of storing just enough information to be able to reconstruct provenance through instrumented replay, has also been adopted for computing the provenance for transactions~\cite{AF18,AG17} 
and in the Subzero system~\cite{wu2013subzero}. 
Subzero switches between different provenance storage representations in an adaptive manner to optimize the cost of provenance queries.  
Amsterdamer et al.~\cite{AD12} demonstrate how to rewrite a query into an equivalent query with 
provenance of minimal size. 
Our work is orthogonal 
in that we focus on minimizing execution time of provenance capture and retrieval.

\section{Experiments}\label{sec:experiments}

Our evaluation focuses on measuring 1) the effectiveness of CBO in choosing the most efficient ICs and PATs, 2) the effectiveness of  heuristic application of PATs, 3) the overhead of heuristic and cost-based optimization, and 4) the impact of CBO search strategies on optimization and execution time.
All experiments were executed on a  machine with 2 AMD Opteron 4238 CPUs, 128GB RAM, and a hardware RAID with  4 $\times$ 1TB 72.K HDs in RAID 5 running commercial DBMS X (name omitted due to licensing restrictions).

To evaluate the effectiveness of our CBO 
vs.  
heuristic optimization choices, 
we compare the performance of instrumented queries generated by the CBO (denoted as \textbf{\textit{Cost}}) against queries generated by selecting a predetermined option for each choice point. Based on a preliminary study we have selected 3 choice points:
 1) using the \textbf{Window} or  \textbf{Join} method; 2) using \textbf{Filter\-Updated} or \textbf{HistJoin} and 3) choosing whether to apply PAT rule~\eqref{eq:duplicate-remove-set} (remove duplicate elimination). If CBO is deactivated, then we always remove such operators if possible. 
The application of the remaining PATs introduced in Sec.~\ref{sec:heuristic} turned out to be always beneficial in our experiments. 
Thus, these PATs are applied as long as their precondition is fulfilled.
We consider two variants for each method: activating heuristic application of the remaining PATs (suffix \textbf{Heu}) or deactivating them (\textbf{NoHeu}).  Unless noted otherwise, results were averaged over 100 runs.

\subsection{Datasets \& Workloads}

\parttitle{Datasets}
\underline{TPC-H}:
We have generated TPC-H benchmark datasets of size 10MB, 100MB, 1GB, and 10GB (SF0.01 to SF10). 
\underline{Synthetic}:
For the transaction provenance experiments we use a 1M tuple relation with uniformly distributed numeric values. 
We vary the size of the transactional history. 
Parameter $HX$
indicates $X\%$ of history, e.g., $H10$ represents $10\%$ history (100K tuples).
\underline{DBLP}:
This dataset consistes of 8 million co-author pairs 
extracted from DBLP (\url{http://dblp.uni-trier.de/xml/}). 
\underline{MESD}: The temporal MySQL employees sample dataset has 6 tables and contains 4M records  (\url{https://dev.mysql.com/doc/employee/en/}).

\parttitle{Simple aggregation queries}
This workload computes the provenance of queries consisting solely of aggregations using Pipeline L1 which applies the rewrite rules for aggregation pioneered in Perm~\cite{glavic2013using} and extended in GProM~\cite{AF18}. A  query consists of $i$ aggregations where each
aggregation operates on the result of the previous aggregation. The leaf operation accesses the TPC-H \texttt{part} table. Every aggregation groups the input on a range of PK values such that the last step returns the same number of results independent of $i$. 

\parttitle{TPC-H queries}
We select 11 
out of  the 22 TPC-H queries to evaluate optimization of provenance capture for complex queries. The technique~\cite{GA09a} we are using supports all TPC-H queries, but instrumentations for nested subqueries have not been implemented in GProM yet.

\parttitle{Transactions}
We use the \textit{reenactment} approach of GProM~\cite{AG17} to compute provenance for transactions executed under isolation
level~\texttt{SERIALIZABLE}. 
The transactional workload is run upfront (not included in the measured execution time) and provenance is computed retroactively.
We vary the number of
updates per transaction, e.g., $U10$ is a transaction with 10 updates. 
The tuples to be updated are selected randomly using the PK of
the relation.

\parttitle{Provenance export}
We use the approach from~\cite{NX15} to translate a relational encoding of provenance (see Sec.~\ref{sec:intro}) into PROV-JSON. 
We export the provenance for a foreign key join across  TPC-H relations nation, customer, and orders.

\parttitle{Provenance for Datalog queries}
We use the approach described in~\cite{LS16} (Pipeline L3). 
The input is a non-recursive Datalog 
query 
$Q$ and a set of (missing) query result tuples of interest.  
We use the DBLP co-author dataset for this experiment and the following queries. \textbf{Q1}: Return authors which have co-authors that have co-authors.
\textbf{Q2}: Return authors that are co-authors, but not of themselves (while semantically meaningless, this query is useful for testing negation).
\textbf{Q3}: Return pairs of authors that are indirect co-authors, but are not direct co-authors.
\textbf{Q4}: Return start points of paths of length 3 in the co-author graph. For each query we consider multiple why questions that specify the set of results for which provenance should be generated. 
We use Qi.j to denote the $j^{th}$ why question for query Qi.

\parttitle{Factorizing Provenance}
We use Pipelines L1 and L5 to evaluate the performance of nested versus ``flat'' provenance under different factorizations (applying the aggregation push-down PATs).
We use the following queries over the TPC-H dataset. 
\textbf{Q1}: An aggregation over a join of tables customer and nation.
\textbf{Q2}: Joins the result of Q1 with the table supplier and adds an additional aggregation. 
\textbf{Q3}: An aggregation over a join of tables nation, customer, and supplier.

\parttitle{Sequenced temporal queries}
We use Pipeline L6 to test the IC which replaces \textit{bag-coalesce} with \textit{set-coalesce} for queries that do not return duplicates. We use the following queries over the temporal MESD  dataset. 
\textbf{Q1}: Return the average salary of employees per department.
\textbf{Q2}: Return the salary and department for every employee (3-way join).

\begin{figure*}[p]

\begin{minipage}{1\linewidth}
\hspace{-10mm}
{
  \resizebox{0.6\linewidth}{!}{
  \begin{minipage}{0.93\linewidth}
  \centering
  \begin{tabular}{|c|r|r|r|r|r|} \hline 
\rowcolor[gray]{.9}  \textbf{Queries} & \textbf{Join+NoHeu} & \textbf{Join+Hue} & \textbf{Window+NoHeu} &\textbf{Window+Heu} &\textbf{Cost+Heu}\\ \hline 
SAgg 1G & 4.79 & 20.21 & 4.38 & 2.69 & \textbf{0.81} \\ \hline 
SAgg 10G & 44.06 & 524.78 & 42.62& 27.47 & \textbf{7.65} \\ \hline 
    \hline
TPC-H 1G & $+$173,053.17 & \textbf{199.62} & 173,041.27 & 250.18 &  235.79 \\ \hline 
    TPC-H 10G & $+$175,371.02 & \textbf{2,033.71} & 175,530.53 & 2,247.39 & 2,196.01 \\ \hline
  \end{tabular}
\end{minipage}
}
}
\hspace{-20mm}
{
  \resizebox{0.6\linewidth}{!}{
  \begin{minipage}{0.93\linewidth}
  \centering
  \begin{tabular}{||c|r|r|r|r|r|} \hline 
\rowcolor[gray]{.9}  \textbf{Queries} & \textbf{Join+NoHeu} & \textbf{Join+Heu} & \textbf{Window+NoHeu} & \textbf{Window+Heu} & \textbf{Cost+Heu}\\ \hline 
SAgg 1G & 1 & 3.927 & 0.946 & 0.600 & \textbf{0.261} \\ \hline 
SAgg 10G & 1 & 9.148 & 0.984 & 0.655 & \textbf{0.265} \\ \hline 
\hline
    TPC-H 1G & 1  & \textbf{0.187}  & 0.955 & 0.220 &  0.203\\ \hline 
TPC-H 10G & 1  & 0.198 & 0.975  & 0.180 & \textbf{0.174} \\ \hline
  \end{tabular}
\end{minipage}
}
}
\caption{Total (\textbf{Left}) and average runtime per query  (\textbf{Right}) relative to Join+NoHeu for \textit{SimpleAgg} and \textit{TPC-H} workloads}
\label{tab:overview-sum-avg-sagg-tpch}
\end{minipage}
\begin{minipage}{0.66\linewidth}
  \begin{minipage}[b]{0.5\linewidth}
  \includegraphics[width=1\linewidth,trim=0 50pt 0 100pt, clip]{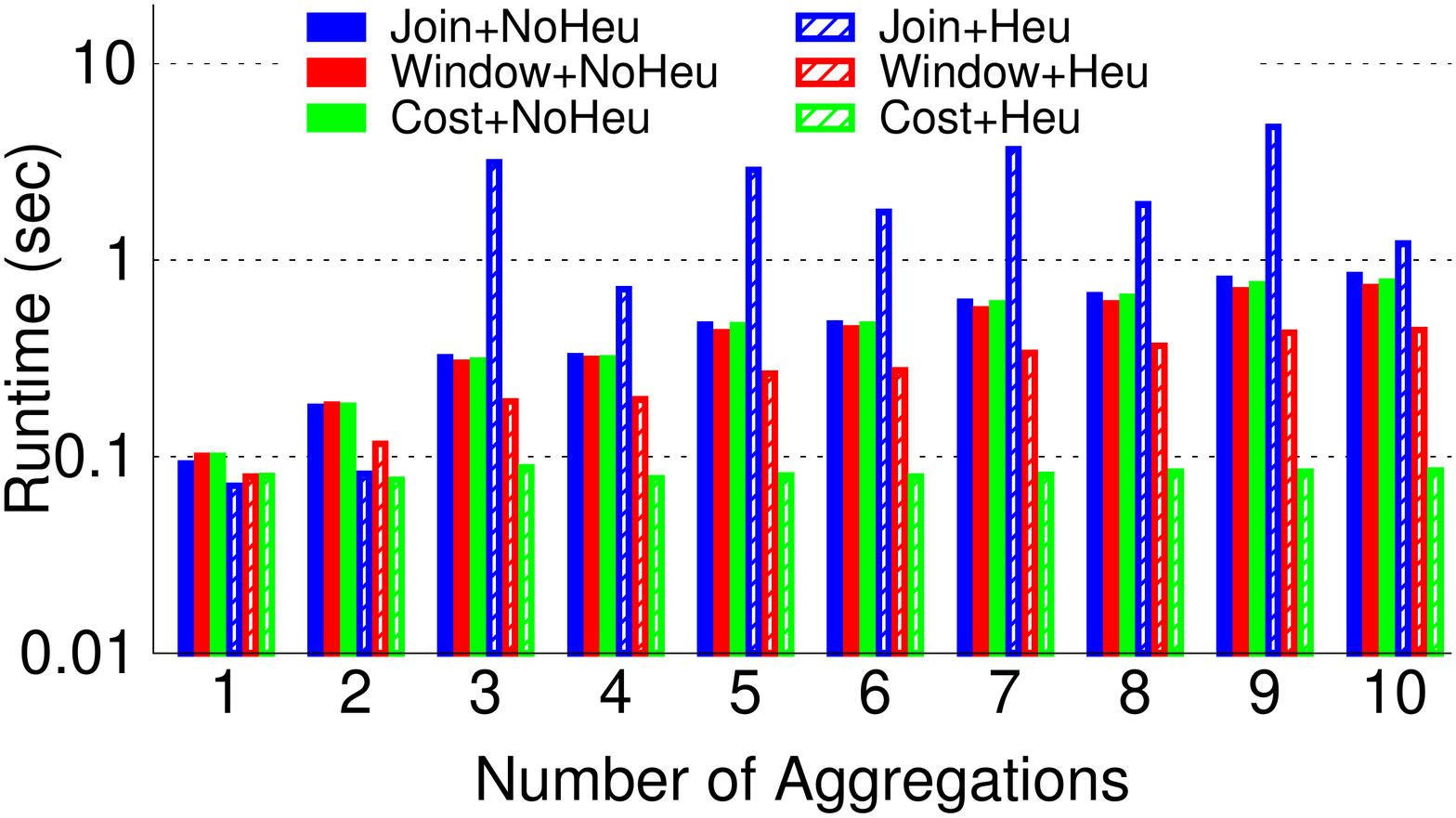}\\[-7mm]
  \caption{1GB \textit{SimpleAgg} runtime}
  \label{fig:simpleAgg-comb-1GB}  
  \end{minipage}
  \begin{minipage}[b]{0.5\linewidth}
  \includegraphics[width=0.97\linewidth,trim=0 50pt 0 90pt, clip]{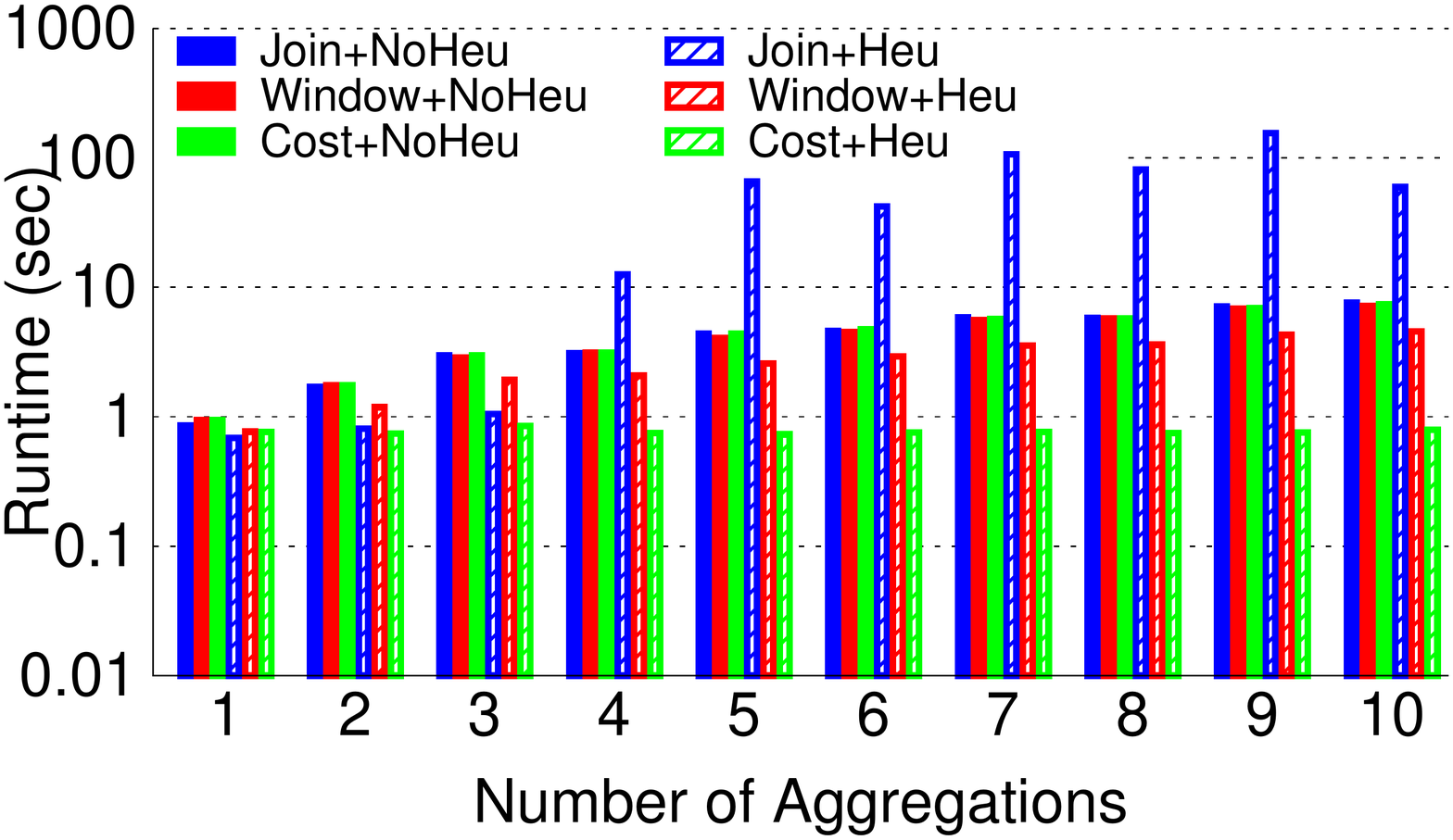}\\[-7mm] 
  \caption{10GB \textit{SimpleAgg} runtime}
  \label{fig:simpleAgg-comb-10GB}
  \end{minipage}

  \begin{minipage}[b]{0.5\linewidth}
  \includegraphics[width=1\linewidth,trim=0 50pt 0 100pt, clip]{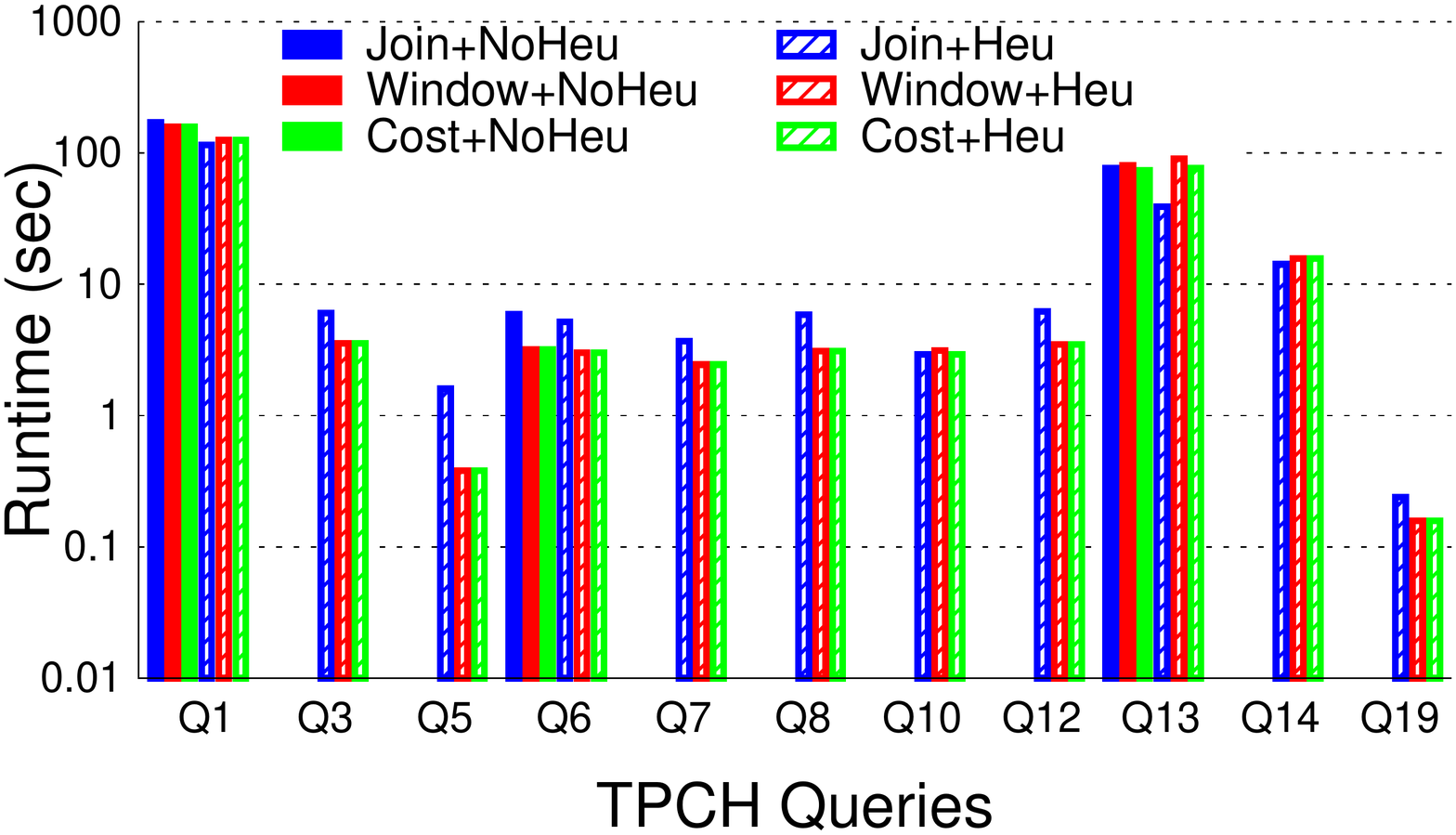}\\[-7mm]
  \caption{Runtime \textit{TPC-H} - 1GB}
  \label{fig:tpch-comb-1GB}  
  \end{minipage}
  \begin{minipage}[b]{0.5\linewidth}
  \includegraphics[width=1\linewidth,trim=0 50pt 0 100pt, clip]{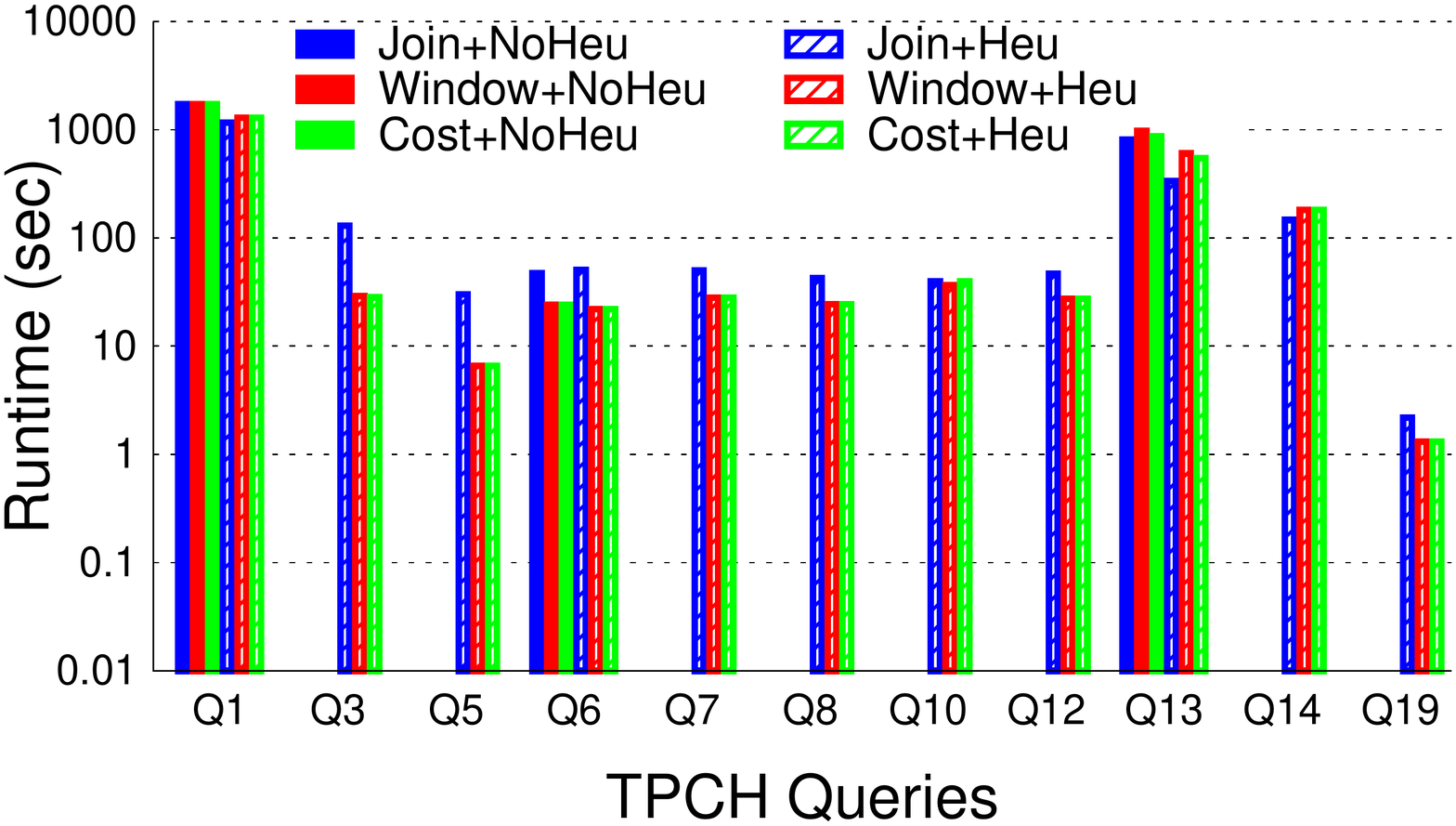}\\[-7mm] 
  \caption{Runtime \textit{TPC-H} - 10GB}
  \label{fig:tpch-comb-10GB}
\end{minipage}

  \begin{minipage}[b]{0.5\linewidth}
  \includegraphics[width=1\linewidth,trim=0 60pt 0 100pt, clip]{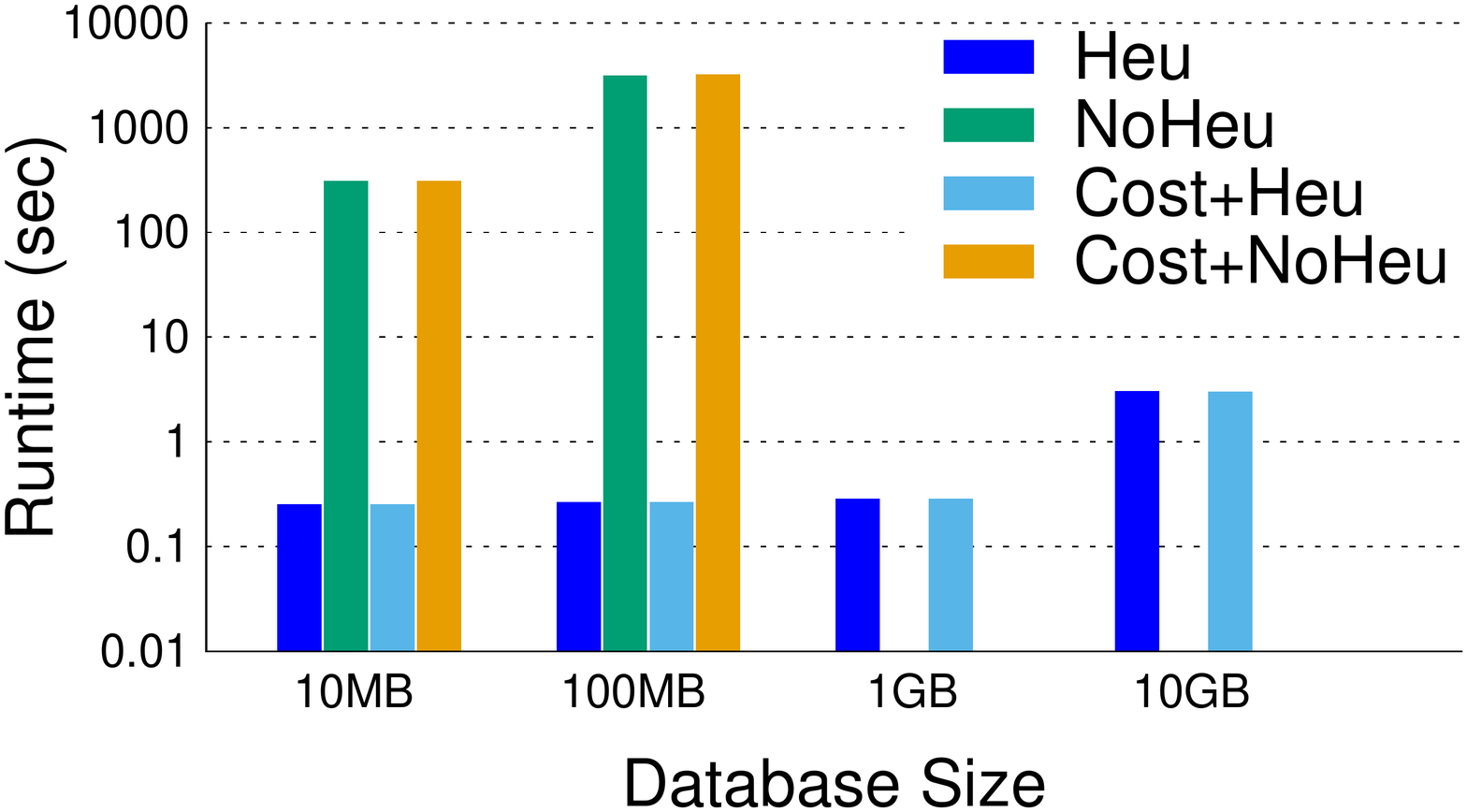}\\[-7mm]
  \caption{Provenance Export}
  \label{fig:export}  
  \end{minipage}
  \begin{minipage}[b]{0.5\linewidth}
  \includegraphics[width=1\linewidth,trim=0 60pt 0 100pt, clip]{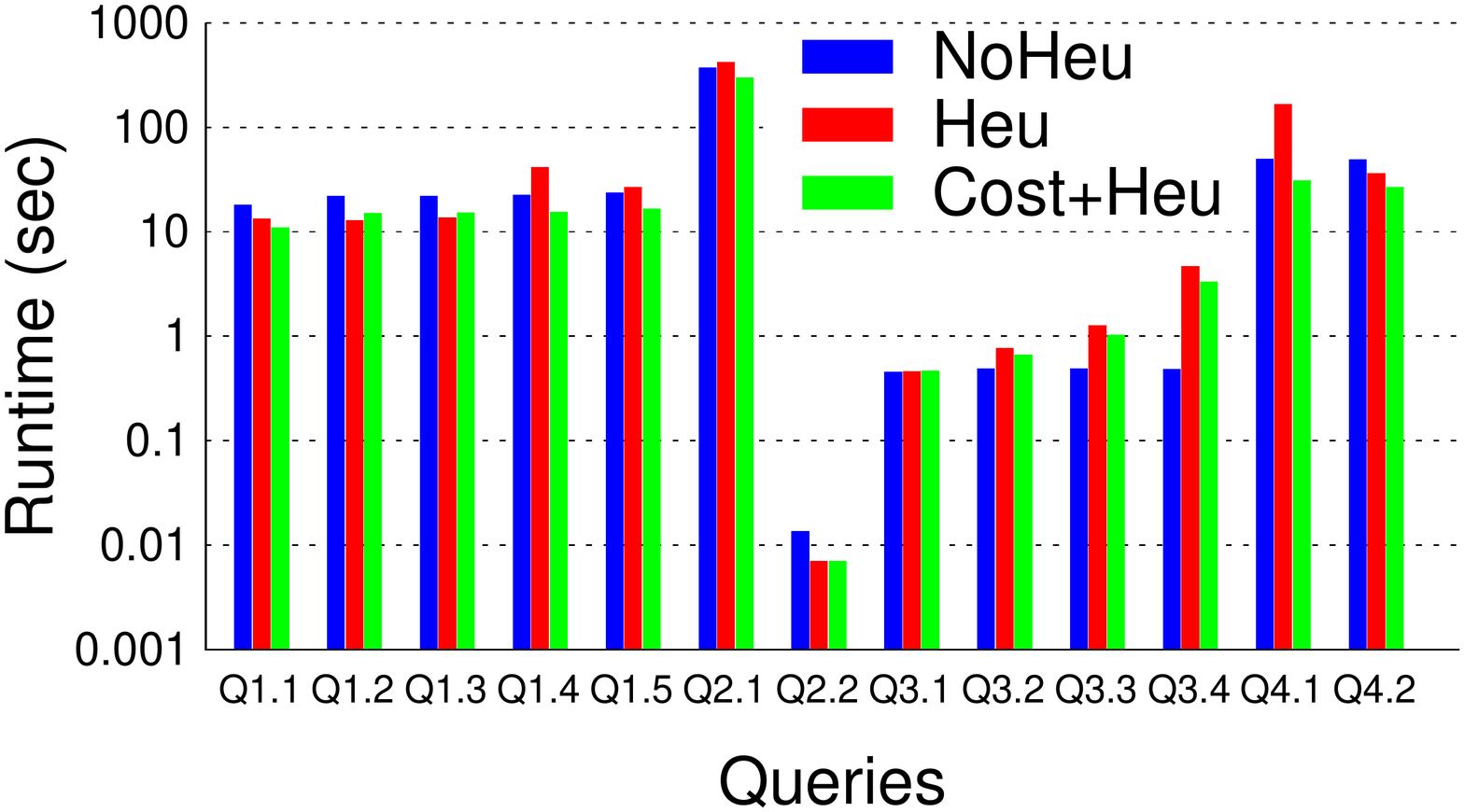}\\[-7mm] 
  \caption{Datalog Provenance}
  \label{fig:provenance-game}
  \end{minipage}

\begin{minipage}[b]{1\linewidth}  
\begin{minipage}[b]{0.48\linewidth}
\includegraphics[width=1\linewidth,trim=0 -10pt 0 30pt, clip]{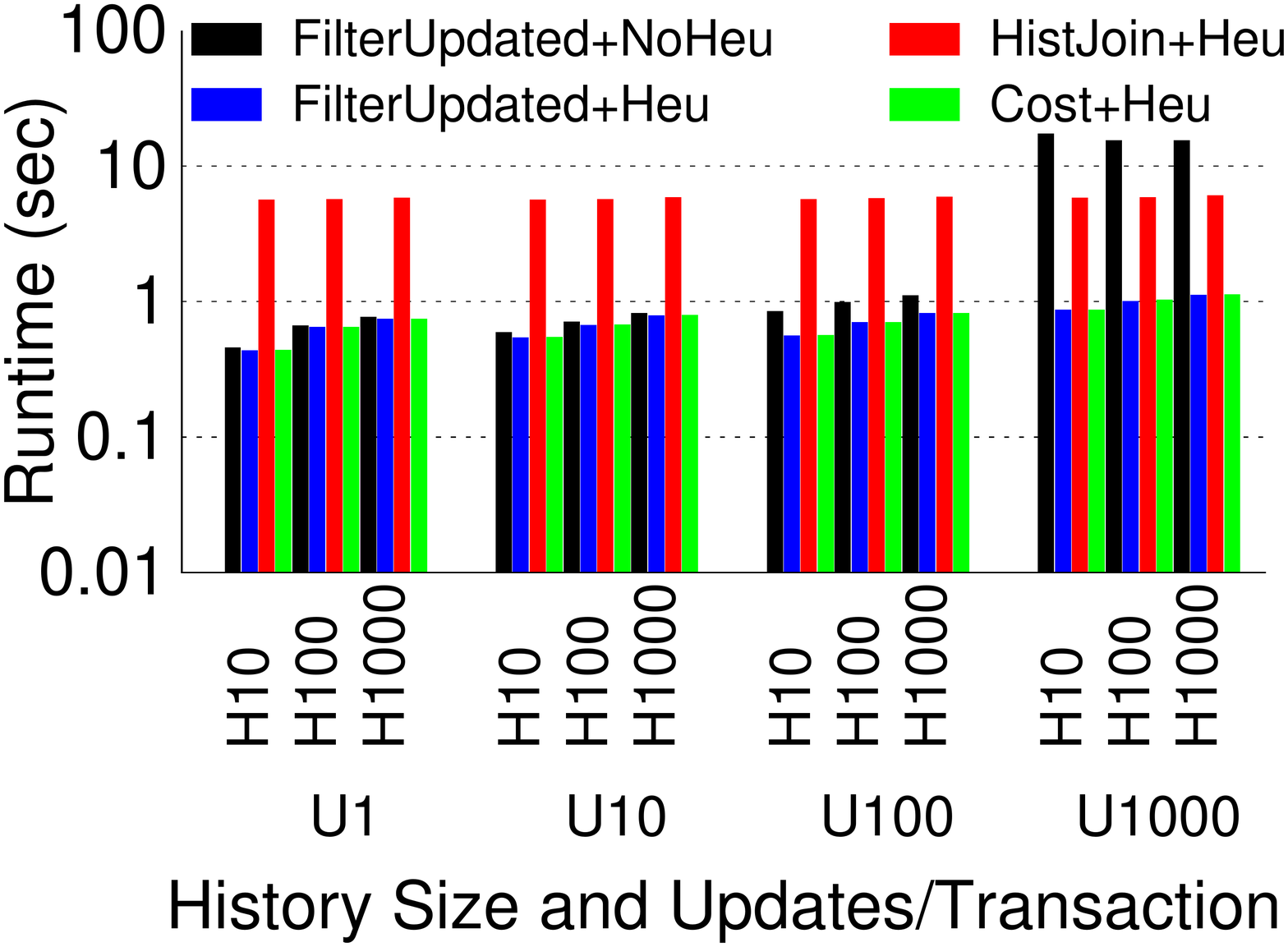}\\
\vspace{-7mm}
\end{minipage}
  \begin{minipage}[b]{0.48\linewidth}
\includegraphics[width=1\linewidth,trim=0 0pt 0 30pt, clip]{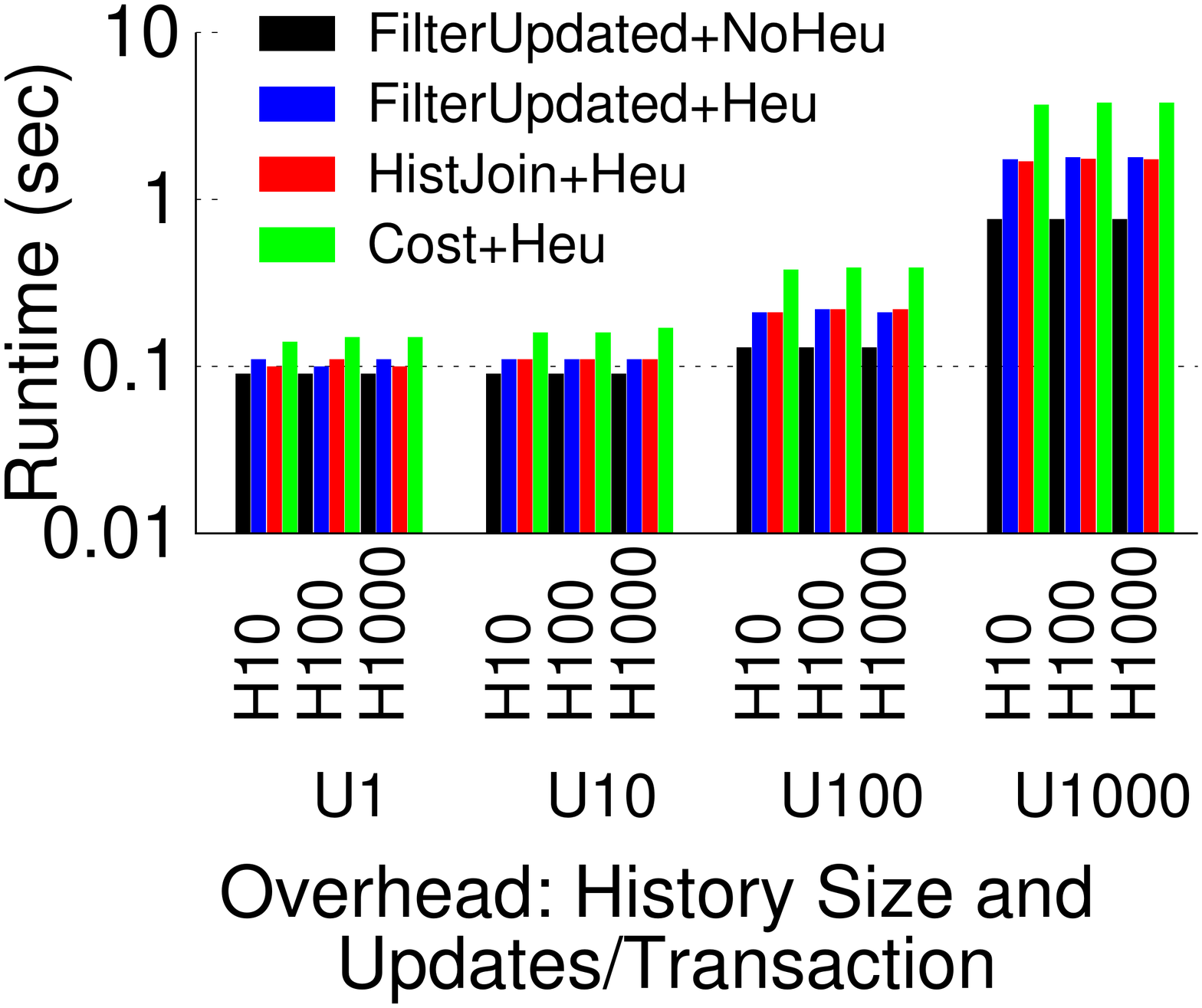}\\
  \vspace{-7mm}
\end{minipage}
  \caption{Transaction provenance - runtime and overhead}
  \label{fig:Transaction-provenance-runtime}
\end{minipage}

  \begin{minipage}{1\linewidth}
  \begin{minipage}[b]{0.49\linewidth}
  \includegraphics[width=1\linewidth,trim=0 50pt 0 100pt, clip]{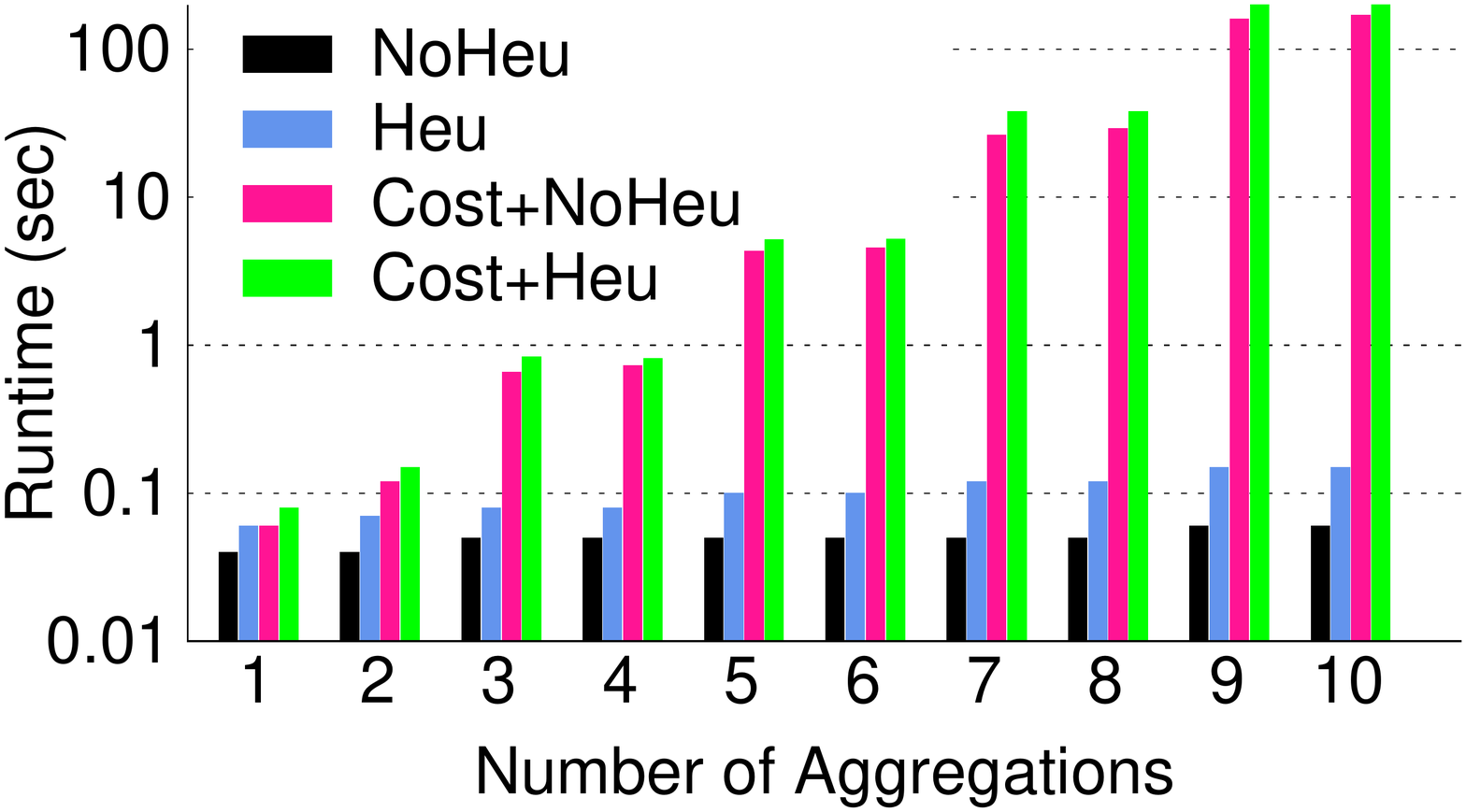}\\[-7mm]
  \label{fig:simple-agg-overhead}  
  \end{minipage}
  \begin{minipage}[b]{0.49\linewidth}
  \includegraphics[width=1\linewidth,trim=0 50pt 0 100pt, clip]{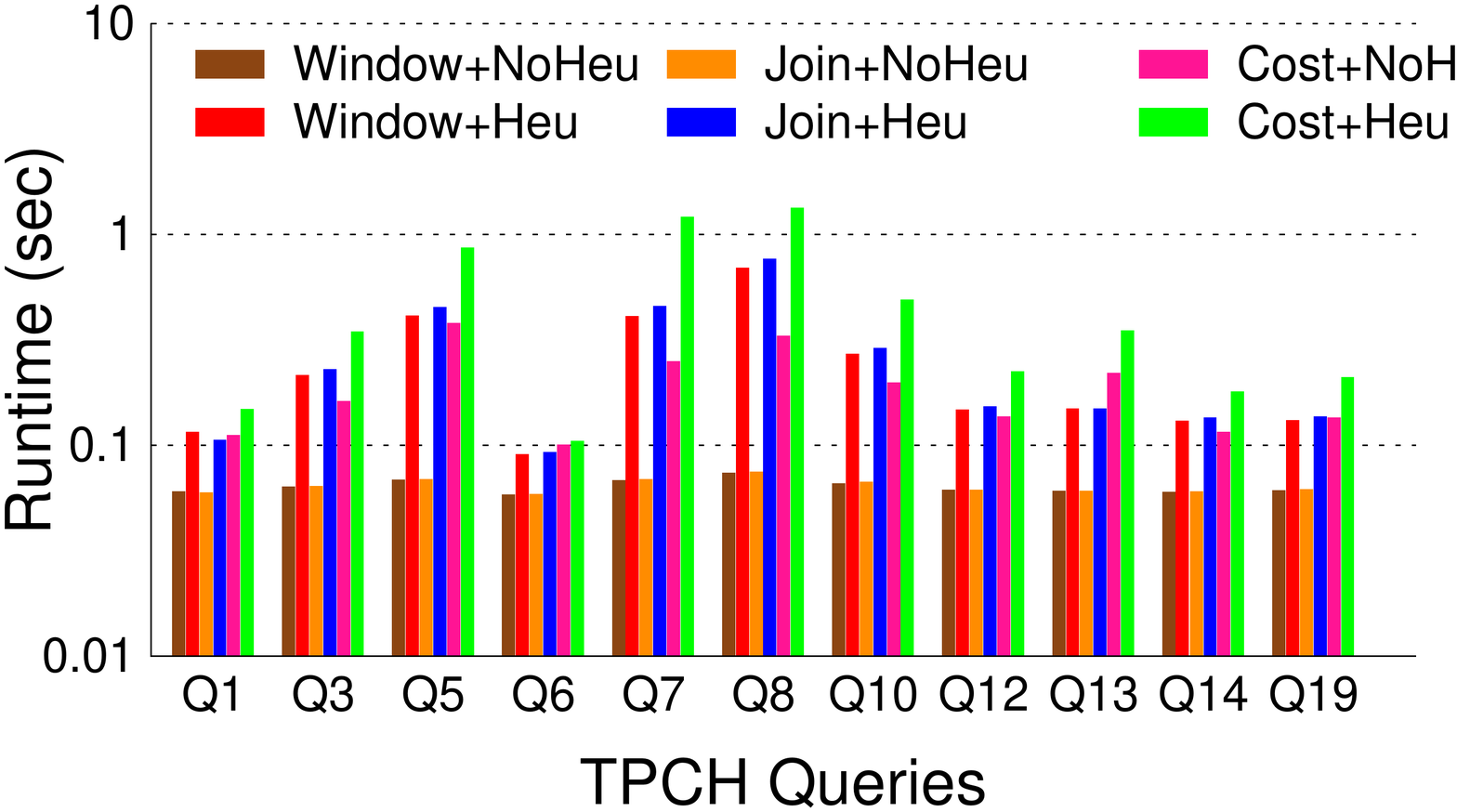}\\[-7mm] 
  \label{fig:tpch-overhead}
  \end{minipage}
  \caption{\textit{SimpleAgg} (\textbf{Left}) and \textit{TPC-H} (\textbf{Right}) Overhead}
  \label{fig:simagg-tpch-overhead}
\end{minipage}

\end{minipage}
\begin{minipage}{0.33\linewidth}
  \begin{minipage}{1\linewidth}
  \centering
\resizebox{1\linewidth}{!}{  
  \begin{minipage}{1.6\linewidth}
    \centering
  \begin{tabular}{|c|r|r|r|r|r|} \hline 
    \rowcolor[gray]{.9}   & \multicolumn{1}{|c|}{\textbf{NoHeu}}  & \multicolumn{1}{|c|}{\textbf{NoHeu}} & \multicolumn{1}{|c|}{\textbf{Heu}}  & \multicolumn{1}{|c|}{\textbf{Heu}}  & \multicolumn{1}{|c|}{\textbf{Cost+Heu}}\\
    \rowcolor[gray]{.9}   & \multicolumn{1}{|c|}{\textbf{(Worst)}} & \multicolumn{1}{|c|}{\textbf{(Best)}} & \multicolumn{1}{|c|}{\textbf{(Worst)}} & \multicolumn{1}{|c|}{\textbf{(Best)}} & \\
    \hline 
Min & 1.33 & 1.33 & \textbf{1.00} & \textbf{1.00} &  \textbf{1.00}\\ \hline 
Avg & 1,878.76 & 1,877.95 & 14.16 & 2.82 & \textbf{1.04} \\ \hline 
Max & $+$12,173.35& $+$12,173.35 & 68.63 & 7.80 & \textbf{1.18} \\ \hline 
  \end{tabular}
\end{minipage}
}\\
\caption{Min, max, and avg runtime relative  to the best method per workload aggregated over all workloads.}
\label{tab:overview-all}
\end{minipage}

\begin{minipage}{1\linewidth}
\resizebox{0.7\linewidth}{!}{  
  \begin{minipage}{1.05\linewidth}
    \centering
    $\,$\\[2mm]
  \begin{tabular}{|c|r|r|r|r|} \hline 
  \rowcolor[gray]{.9}  \textbf{Queries} & \textbf{FilterUpdated} & \textbf{HistJoin}& \textbf{FilterUpdated} & \textbf{Cost}\\
  \rowcolor[gray]{.9}  \textbf{Queries} & \textbf{+NoHeu} & \textbf{+Heu} & \textbf{+Heu} & \textbf{+Heu}\\ \hline 
HSU/T & 55.11 & 69.50 &  \textbf{8.91} & \textbf{8.96} \\ \hline 
TAPU & 30.13 & 26.08 & \textbf{12.94} & \textbf{12.89} \\ \hline 
  \end{tabular}
\end{minipage}
}
\caption{Total workload runtime for transaction provenance}
\label{tab:overview-sum-transaction}
\end{minipage}

\begin{minipage}{1\linewidth}
 
\resizebox{0.7\linewidth}{!}{
  \begin{minipage}{0.95\linewidth}
  \begin{tabular}{|c|r|r|r|} \hline 
\rowcolor[gray]{.9}  \textbf{Queries} & \textbf{NoHeu} & \textbf{Heu} & \textbf{Cost+Heu} \\ \hline 
Export 10M & 310.49 &  \textbf{0.25}&  \textbf{0.25} \\ \hline 
Export 100M & 3,136.94 & \textbf{0.27} &  \textbf{0.26}  \\ \hline 
Export 1G & +21,600 & \textbf{0.28} &  \textbf{0.28}   \\ \hline 
Export 10G & +21,600 & \textbf{3.03} & \textbf{3.01}  \\ \hline \hline
Datalog Provenance & 583.96 & 736.50 & \textbf{437.75} \\ \hline 
  \end{tabular}
\end{minipage}
}\\
\caption{Total runtime for export and Datalog workloads}
\label{tab:overview-sum-export-gp}
\end{minipage}

  \begin{minipage}[b]{1\linewidth}
  \includegraphics[width=1\linewidth,trim=0 0pt 0 20pt, clip]{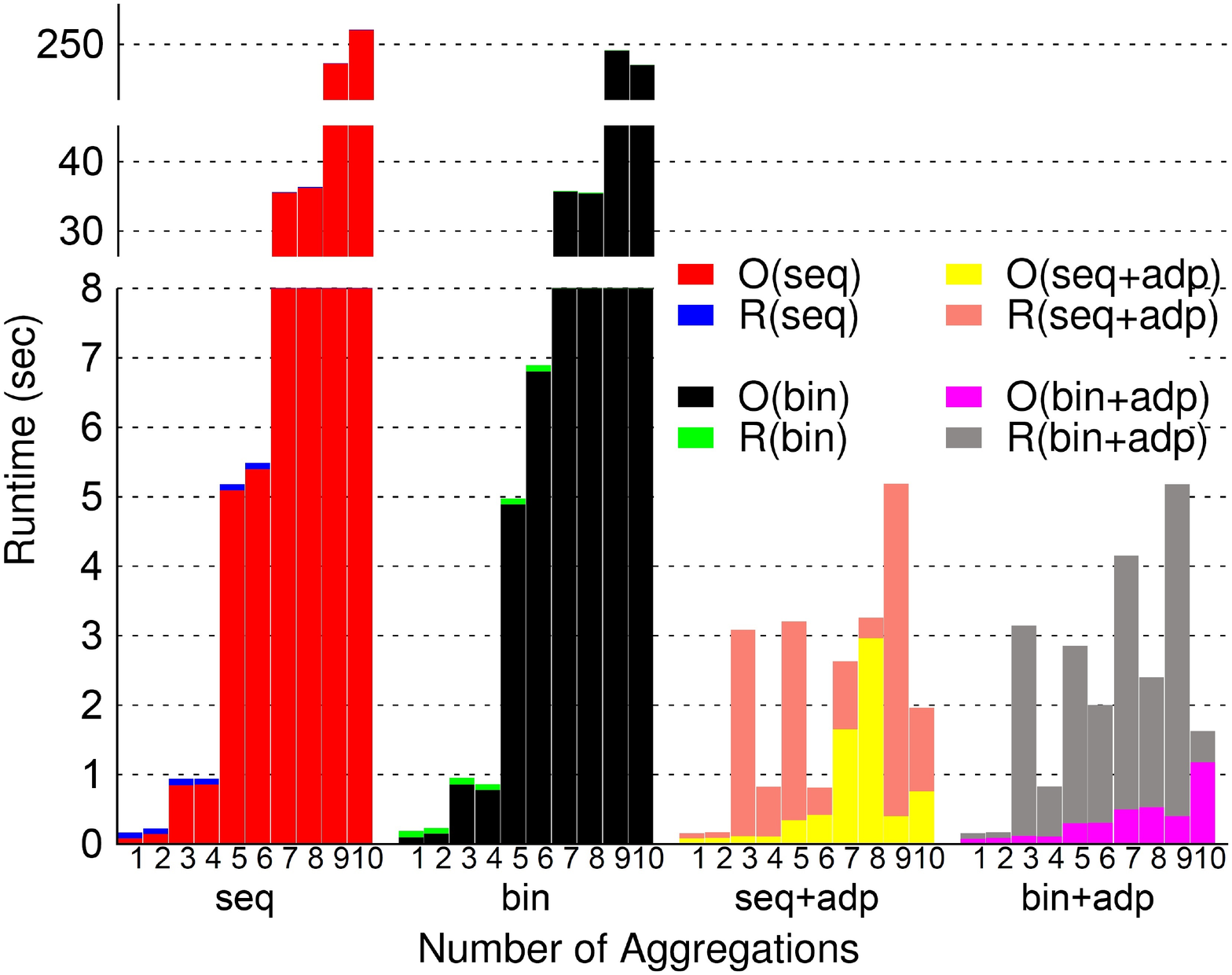}\\[-5mm]
  \caption{Optimization + runtime for Simple Agg. - 1GB}
  \label{fig:simAggs-all-stack}  
  \end{minipage}
  \begin{minipage}[b]{1\linewidth}
  \includegraphics[width=1\linewidth,trim=0 30pt 0 40pt, clip]{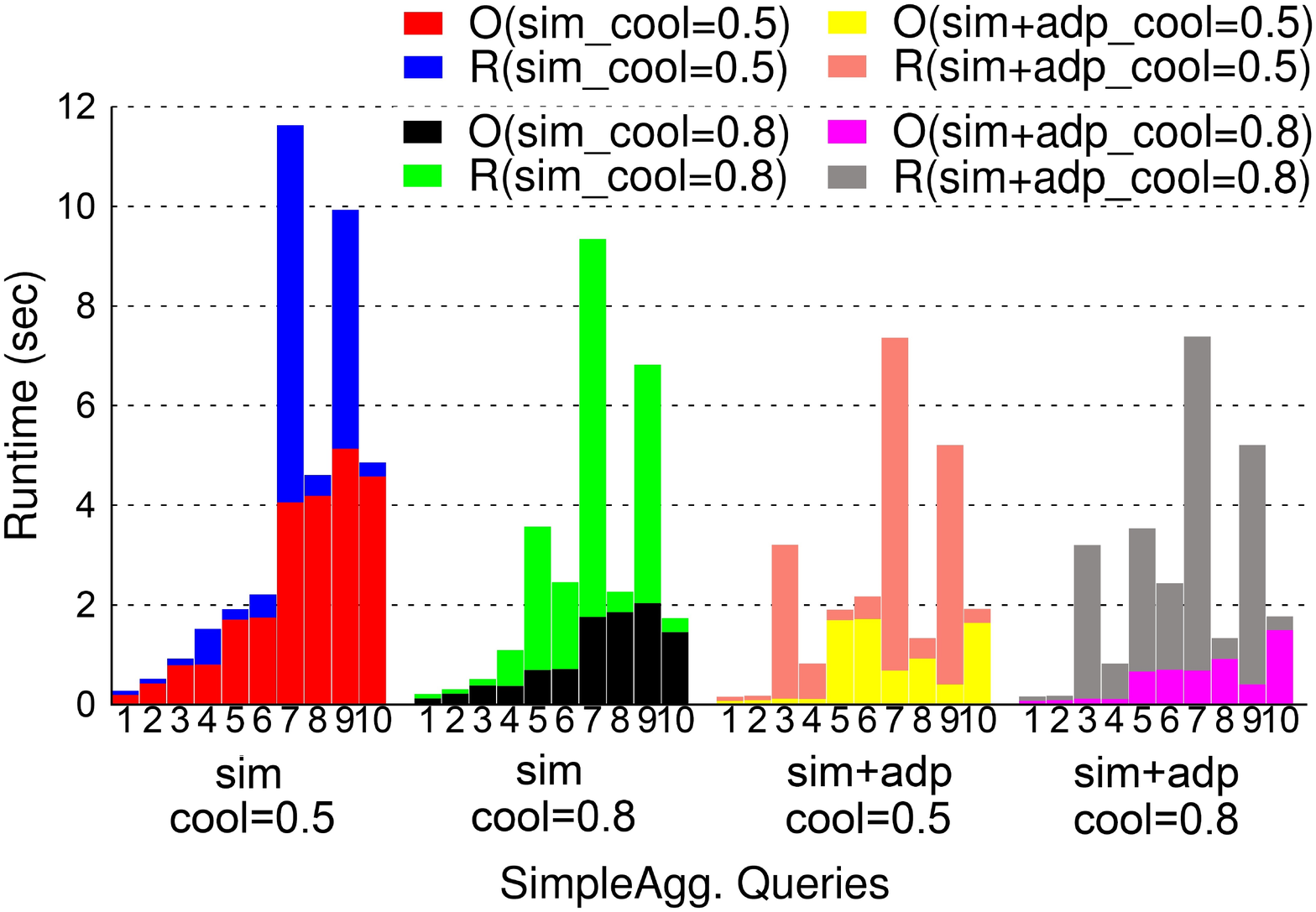}\\[-5mm] 
  \caption{Optimization + runtime for Simple Aggregation workload using Simulated Annealing - 1GB dataset}
  \label{fig:simple-agg-sim-annl}
  \end{minipage}

\end{minipage}
\end{figure*}

\subsection{Measuring Query Runtime}

\parttitle{Overview} 
Fig.~\ref{tab:overview-all} shows an overview of our results. We show the average  runtime of each method relative to the best method per workload, e.g., if \textit{Cost} performs best for a workload then its runtime is normalized to 1. We  use relative overhead instead of total runtime over all workloads, because some workloads are significantly more expensive than other. 
For the \textit{NoHeu} and \textit{Heu} methods we report the performance of the best and the worst option for each choice point. For instance, for the \textit{SimpleAgg} workload the performance is impacted by the choice of whether the \textit{Join} or \textit{Window} method is used to instrument aggregation operators with \textit{Window} performing better (\textit{Best}). Numbers prefixed by a $'+'$  indicate that for this method some queries of the workload did not finish within the maximum time we have allocated for each query. Hence, the runtime reported for these cases should be interpreted as a lower bound on the actual runtime.  Compared with other methods, \textit{Cost+Heu} is on average only 4\% worse than the best method for the workload and has 18\% overhead in the worst case. Note that we confirmed that in all cases where an inferior plan was chosen by our CBO that was because of inaccurate cost estimations by the backend database. If we heuristically choose the best option for each choice point, then this results in a  178\% overhead over CBO on average. However, achieving this performance requires that the best option for each choice point is known upfront. 
Using a suboptimal heuristic on average increases runtime by a factor of $\sim$ 14 compared to CBO. These results also confirm the critical importance of our PATs since deactivating these transformations 
increases runtime by a factor of $\sim$ 1,800 on average.

\parttitle{Simple Aggregation Queries} 
We measure the runtime of computing provenance for the \textit{SimpleAgg} workload over the 1GB and 10GB TPC-H datasets varying the number of aggregations per query. 
The total workload runtime is shown in Fig.~\ref{tab:overview-sum-avg-sagg-tpch} (the best method is shown in bold). We also show the average runtime per query relative to the runtime of \textit{Join+NoHeu}. 
 CBO significantly outperforms the other methods. The \textit{Window} method is more effective than the \textit{Join} method if a query contains multiple levels of aggregation. Our heuristic optimization improves the runtime of this method by about 50\%. The unexpected high runtimes of \textit{Join+Heu} are explained below. 
Fig.~\ref{fig:simpleAgg-comb-1GB} and \ref{fig:simpleAgg-comb-10GB} show the results for individual queries. Note that the y-axis is log-scale.
Activating \textit{Heu} 
improves performance in most cases, but the dominating factor for this workload  is choosing the right method for instrumenting aggregations. 
The exception is the \textit{Join} method, where runtime increases when \textit{Heu} 
is activated.  
We inspected the plans used by the backend DBMS for this case. A suboptimal join order was chosen for \textit{Join+Heu} based on inaccurate estimations of intermediate result sizes. For \textit{Join} the DBMS did not remove intermediate operators that blocked join reordering and, thus, executed the joins in the order provided in the input query which turned out to be more efficient in this particular case.
Consistently, 
CBO did either select \textit{Window} as the superior method (confirmed by inspecting the generated execution plan) or did outperform both \textit{Window} and \textit{Join} by instrumenting some aggregations using the \textit{Window} and others with the \textit{Join} method.\\[2mm]

\parttitle{TPC-H Queries} 
We compute the provenance of TPC-H queries to determine whether the results for simple aggregation queries
translate to more complex queries. 
The total workload execution time is shown in Fig.~\ref{tab:overview-sum-avg-sagg-tpch}. 
 We also show the average runtime per query relative to the runtime of \textit{Join+NoHeu}.
Fig.~\ref{fig:tpch-comb-1GB} and \ref{fig:tpch-comb-10GB} show the running time for each query for the 1GB and 10GB  datasets. Our CBO significantly outperforms the other methods with the only exception of \textit{Join+Heu}. Note that the runtime of  \textit{Join+Heu} for Q13 and Q14 is lower than \textit{Cost+Heu} which causes this effect.  
Depending on the dataset size and query, there are cases where the \textit{Join} method is superior and others where the \textit{Window} method is superior. The runtime difference between these methods is less pronounced than for \textit{SimpleAgg}, presenting a challenge for our CBO. 
Except for Q13 which contains 2 aggregations, all other queries only contain one aggregation. The CBO was able to determine the best method to use in almost all cases. Inferior choices are again caused by inaccurate cost estimates. 
We also show the results for \textit{NoHeu}. However, only three queries finished within the allocated time slot of 6 hours (Q1, Q6 and Q13). These results demonstrate the need for PATs and the robustness of our CBO method.

\parttitle{Transactions}
We next compute the provenance of transactions executed over the synthetic dataset using the techniques introduced in~\cite{AG17}. We vary the number of updates per
transaction ($U1$ up to $U1000$) and the size of the database's history
($H10$, $H100$, and $H1000$). 
The total workload runtime is shown in Fig.~\ref{tab:overview-sum-transaction}.  
The left graph in Fig.~\ref{fig:Transaction-provenance-runtime} shows detailed results. 
We compare the runtime of \textit{FilterUpdated} and \textit{HistJoin} (\textit{Heu} and \textit{NoHeu}) with \textit{Cost+Heu}. 
Our CBO choses \textit{FilterUpdated}, the superior option.

\parttitle{Provenance Export}
Fig.~\ref{fig:export} shows results for the provenance export workload for dataset sizes from 10MB up to 10GB (total workload runtime is shown in Fig.~\ref{tab:overview-sum-export-gp}). 
\textit{Cost+\-Heu} and \textit{Heu} both outperform \textit{NoHeu} demonstrating the key role of PATs for this workload. 
Our provenance instrumentations use window operators for enumerating intermediate result tuples which prevents the database from pushing selections and reordering joins.  \textit{Heu} outperforms \textit{NoHeu}, because it removes some of these window operators (PAT rule~\eqref{eq:window-function}).
 CBO does not further improve performance, because the export query does not apply aggregation or duplicate elimination, i.e., none of the choice points were hit.

\parttitle{Why Questions for Datalog}
The approach~\cite{LS16} we use for generating provenance for 
Datalog queries with negation may produce queries which contain a large amount of duplicate elimination operators and shared subqueries. The heuristic application of PATs would remove all but the top-most duplicate elimination operator (rules~\eqref{eq:duplicate-remove} and~\eqref{eq:duplicate-remove-set}  in Fig.~\ref{fig:algebraic-rules}). However, this is not always the best option, because a duplicate elimination, while adding overhead, can reduce the size of inputs for downstream operators. Thus, as mentioned before we consider the application of Rule 2 as an optimization choice in our CBO.
The total workload runtime and results for individual queries are shown in Fig.~\ref{tab:overview-sum-export-gp} and Fig.~\ref{fig:provenance-game}, respectively. 
Removing all redundant duplicate elimination operators (\textit{Heu}) is not always better than removing none (\textit{NoHeu}). Our CBO (\textit{Cost+Heu}) has the best performance in almost all cases by choosing a subset of duplicate elimination operators to remove. 

\begin{figure}
\begin{minipage}{1\columnwidth}
\captionsetup{font=small}

\centering
  \resizebox{0.9\linewidth}{!}{
  \begin{minipage}{1.1\linewidth}

\label{tab:one-level}
  \begin{tabular}{|c| c || r  r | r  r |} \cline{3-6} 
\multicolumn{2}{c|}{} &  \multicolumn{2}{c|}{\cellcolor{lgrey}\textbf{Normal (NoHeu)}} &  \multicolumn{2}{c|}{\cellcolor{lgrey}\textbf{Fac (Heu)}} \\ \hline
\rowcolor[gray]{.9} \textbf{Query} &    \textbf{Size} &  \textbf{Prov} & \textbf{Xml} & \textbf{Prov} & \textbf{Xml} \\ \hline 
\multirow{4}{*}{\textbf{Q1}} &    10MB & 0.0991 & 0.0538  & 0.1062  & \textbf{0.0410}   \\  
&    100MB &  0.8481  & 0.37629  & 0.9302  & \textbf{0.2582}     \\   
&    1GB &  8.4233  & 6.9261  & 8.9676   & \textbf{4.8069}  \\   
&   10GB & 122.5400  & 95.1900  & 150.2800 & \textbf{77.2800}   
    \\ \hline
\multirow{4}{*}{\textbf{Q2}}&  10MB & 0.1876 & 0.0584 & 0.1386  & \textbf{0.0352} \\ 
&  100MB & 17.5624  &  \cellcolor{red!25}error & 12.4271  & \textbf{0.2627} \\ 
&  1GB & +3600.0000  & \cellcolor{red!25}error  & 1329.0000  & \textbf{5.3133}  \\
&  10GB & +3600.0000  & \cellcolor{red!25}error  & +3600.0000  & \textbf{86.9500} \\ \hline
\multirow{2}{*}{\textbf{Q3}}& 10MB & 0.4312 & 0.1357 
                                                                                      & 0.4414  & \textbf{0.0918}
  \\ 
  
&  100MB & 42.8268  & 35.9454  & 
                                                                       47.4869  & \textbf{5.2949} 
  
  \\ \hline
  
  \end{tabular}

\end{minipage}
}
\end{minipage}\\[-2mm]
\caption{Runtime for Factorization Queries (Q1 to Q3)}
\label{tab:factor-runtimes}
\end{figure}

\parttitle{Factorizing Provenance}
We compare the runtime of Pipeline L1 (\textbf{Prov}) 
against P5 (\textbf{XML}) which produces a nested representation of provenance. We test the effect of the heuristic application of aggregation push-down (Rules 5 and 8 from Fig.~\ref{fig:algebraic-rules}) to factorize provenance.  Fig.~\ref{tab:factor-runtimes} shows the runtimes for the factorization workload (queries Q1 to Q3). In general, \textit{XML} outperforms \textit{Prov} since it reduces the number of query results (rows) and total size of the results in bytes. 
\textit{Prov} does not benefit much from aggregation pushdown, because this does not affect the size of the returned provenance. This optimization improves performance for \textit{XML}, specifically for larger database instances. In summary, \textit{XML+Heu} is the fasted method in all cases, outperforming \textit{Prov+Heu} by a factor of up to 250. Note that DBMS X does not support large XML values in certain query contexts that require sorting. A query that encounters such a situation will fail with an error message (marked in red in Fig.~\ref{tab:factor-runtimes}).

\parttitle{Set vs. Bag Coalescing}
We also run sequenced temporal queries comparing \textit{Heu} (use set-coalesce) and \textit{NoHeu}. 
The result set of Query Q1 is small. Thus, using set-coalesce (\textit{Heu}) only improves performance  by $\sim$10\%.  The runtimes are 4.85s \textit{(Heu)} and 5.27s \textit{(NoHeu)}.  
+Choosing the right coalescing operator is more important for Query Q2 which returns 2.8M tuples (35.38s for \textit{Heu} and 64s for \textit{NoHeu}). 

\subsection{Optimization Time and CBO Strategies}

\parttitle{Simple Aggregation}
We show the optimization time of several methods in Fig.~\ref{fig:simagg-tpch-overhead} (left). 
Heuristic optimization (\textit{Heu}) results in an overhead of $\sim$50ms compared to the time of compiling a provenance request without optimization (\textit{NoHeu}). This overhead is only slightly affected by the number of aggregations. 
The overhead is higher for \textit{Cost} because we have 2 choices for each aggregation, i.e., the plan space size is $2^{i}$ for $i$ aggregations. 
We have measured where time is spend during CBO 
and have determined that the majority of time is spend in costing SQL queries using the backend DBMS. Note that even though we did use the exhaustive search space traversal method for our CBO,  
the sum of optimization time and runtime for \textit{Cost} is still less than this sum for the \textit{Join} method for some queries.

\parttitle{TPC-H Queries}
In Fig.~\ref{fig:simagg-tpch-overhead} (right), we show the optimization time for TPC-H queries. Activating PATs results in $\sim$50ms overhead in most cases with a maximum overhead of $\sim$0.5s.  This is more than offset by the gain in query performance (recall that with \textit{NoHeu} only 3 queries finish within 6 hours for the 1GB dataset). CBO takes up to 3s in the worst case.

\parttitle{CBO Strategies}
We now compare query runtime and optimization time for the CBO search space traversal strategies introduced in Sec.~\ref{sec:cbo}.   
Recall that the \textit{sequential-leaf-traversal \textbf{(seq)}} and
\textit{binary-search-traversal \textbf{(bin)}} strategies are both exhaustive strategies. 
 \textit{Simulated Annealing \textbf{(sim)}} is the metaheuristic as introduced in Sec.~\ref{sec:traversal-strategies}.
We also combine these strategies with our \textit{adaptative \textbf{(adp)}} heuristic that limits time spend on optimization based on the expected runtime of the best plan found so far. 
 Fig.~\ref{fig:simAggs-all-stack} shows the total time (runtime (\textbf{R}) + optimization time (\textbf{O})) for the simple aggregation workload. We use this workload because it contains some queries with a large plan search space.  
 Not surprisingly,  the runtime of queries produced by \textit{seq} and \textit{bin} is better than \textit{seq+adp} and \textit{bin+adp} as \textit{seq} and \textit{bin} traverse the whole search space. However, their total time is much higher than \textit{seq+adp} and \textit{bin+adp} for larger numbers of aggregations. 
Fig.~\ref{fig:simple-agg-sim-annl} shows the total time of \textit{sim} with and without the \textit{adp} strategy for the same workload. We used cooling rates (\textit{cr}) of 0.5 and 0.8 because they resulted in the best performance. 
The \textit{adp} strategy improves the runtime in all cases except for the query with 3 aggregations.  
We also evaluated the effect of the \textit{cr} and \textit{c} parameters for simulated annealing \textit{\textbf{(sim)}} and its  adaptive version \textit{\textbf{(sim+adp)}}  by varying the \textit{cr}  (0.1 $\sim$ 0.9) and \textit{c} value (1, 100 and 10000) for Simple Aggregation query Q10 over the 1GB dataset. The choice of parameter \textit{c} had negledible impact. Thus, we focus on \textit{cr}. Tab.~\ref{tab:q10-sim} shows the optimization time for \textit{c}=10000 for these two methods. The query execution time was 0.27s for \textit{cr} (0.1 $\sim$ 0.8) and 1.2s for \textit{cr} 0.9. The total cost is minimized (2.0+0.27=2.27s) when for \textit{cr} 0.8. \textit{sim+adp} further reduces the optimization time to roughly 1.64s independent of the \textit{cr}.

\begin{table}
\begin{minipage}{0.9\linewidth}
\captionsetup{font=small}

  \resizebox{0.83\linewidth}{!}{
  \begin{minipage}{0.9\linewidth}
\centering

  \begin{tabular}{|c|c|c|c|c|c|c|c|c|c|} \cline{2-10}
\multicolumn{1}{c}{}   &  \multicolumn{9}{|c|}{\textbf{cooling rate (cr)}} \\ \hline
  \rowcolor[gray]{.9} \textbf{Method}  &  \textbf{0.1} & \textbf{0.2} & \textbf{0.3} & \textbf{0.4} & \textbf{0.5} & \textbf{0.6} & \textbf{0.7} & \textbf{0.8} & \textbf{0.9} \\ \hline 
    Sim & 28.9 & 14.6 & 8.6  & 6.5  & 4.6  & 3.6 & 2.6 & 2.0 &1.5  \\ \hline
    Sim+Adp & 1.64 & 1.64 & 1.64  & 1.64  & 1.64  & 1.64 & 1.63 & 1.62 &1.5  \\ \hline

  \end{tabular}  
\end{minipage}
}
\end{minipage}\\[-2mm]
\caption{Parameter Sensitivity for Simulated Annealing}
\label{tab:q10-sim}
\end{table}

 \section{Conclusions and Future Work}\label{sec:conclusion}

We present the first cost-based optimization framework for provenance instrumentation and its implementation in GProM. 
Our approach supports both heuristic and cost-based choices and is applicable to a wide range of instrumentation pipelines.
We study provenance-specific algebraic transformations (PATs) and instrumentation choices (ICs), i.e., alternative ways of realizing provenance capture.
We demonstrate experimentally 
that our optimizations improve performance by several orders of magnitude for diverse provenance tasks.                                       
An interesting avenue for future work is to incorporate CBO with provenance compression. 

{\footnotesize \parttitle{Acknowledgements} This work was supported by NSF Award \#1640864. Opinions, findings and conclusions expressed in this material are those of the authors and do not necessarily reflect the views of the National Science Foundation.}

\appendices

\section{Pipelines}\label{sec:pipelines}

Fig.~\ref{fig:all-pipelines} shows all pipelines described in this paper. We present additional details about pipeline L5 and L6 (Fig.~\ref{fig:xml-rewrite-approach} and~\ref{fig:temporal-rewrite-approach}) in the following.

\begin{figure}[t]
  \centering
\subfloat[
Provenance is captured using an annotated version of relational algebra which is first translated into relational algebra over a relational encoding of annotated relations and then into SQL code. 
]{  \label{fig:general-rewrite-approach}
  \begin{minipage}{0.96\linewidth}
  \centering
\includegraphics[width=1\columnwidth]{figs/gen_rewr_approach.pdf}
\end{minipage}
}\\[-0.5mm]
\subfloat[
In addition to the steps of \textbf{(a)}, this pipeline 
uses \emph{reenactment}~\cite{AG17} to compile annotated updates into annotated queries. 
]{\label{fig:trans-rewrite-approach}
  \begin{minipage}{0.98\linewidth}
  \centering
  \includegraphics[width=0.9\columnwidth]{figs/gen_reenact_approach.pdf}
\end{minipage}
}\\[-0.5mm]
\subfloat[Computing provenance graphs for Datalog queries~\cite{LS16} based on a rewriting called \emph{firing rules}. The instrumented Datalog program is first compiled into relational algebra and then into SQL.]{  \label{fig:DL-rewrite-approach}
\begin{minipage}{0.98\linewidth}
  \centering
\includegraphics[width=1\columnwidth]{figs/gen_dl_approach.pdf}
\end{minipage}  
}\\[-2.2mm]
\subfloat[Translating the relational provenance encoding produced by L1 into PROV-JSON.]{  \label{fig:expor-rewrite-approach}
\begin{minipage}{0.98\linewidth}
  \centering
\includegraphics[width=1\columnwidth]{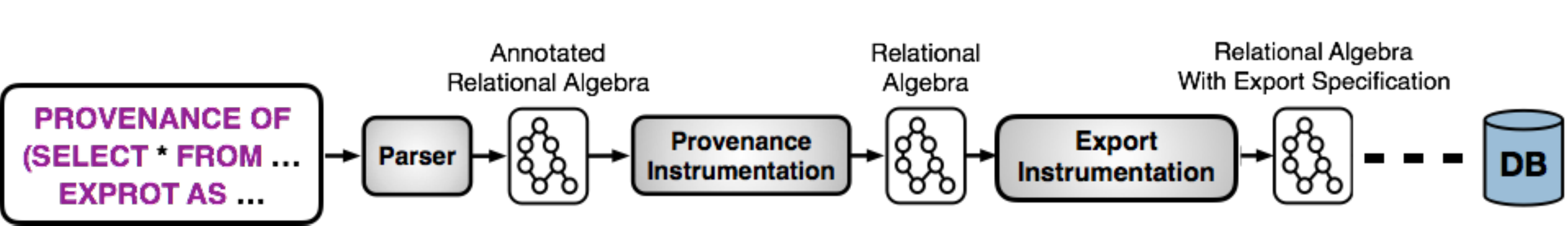}
\end{minipage}  
}\\[-0.5mm]
\subfloat[Capturing provenance for queries using L1 and encoding it in a nested representation (XML).]{  \label{fig:xml-rewrite-approach}
\begin{minipage}{0.98\linewidth}
  \centering
\includegraphics[width=1\columnwidth]{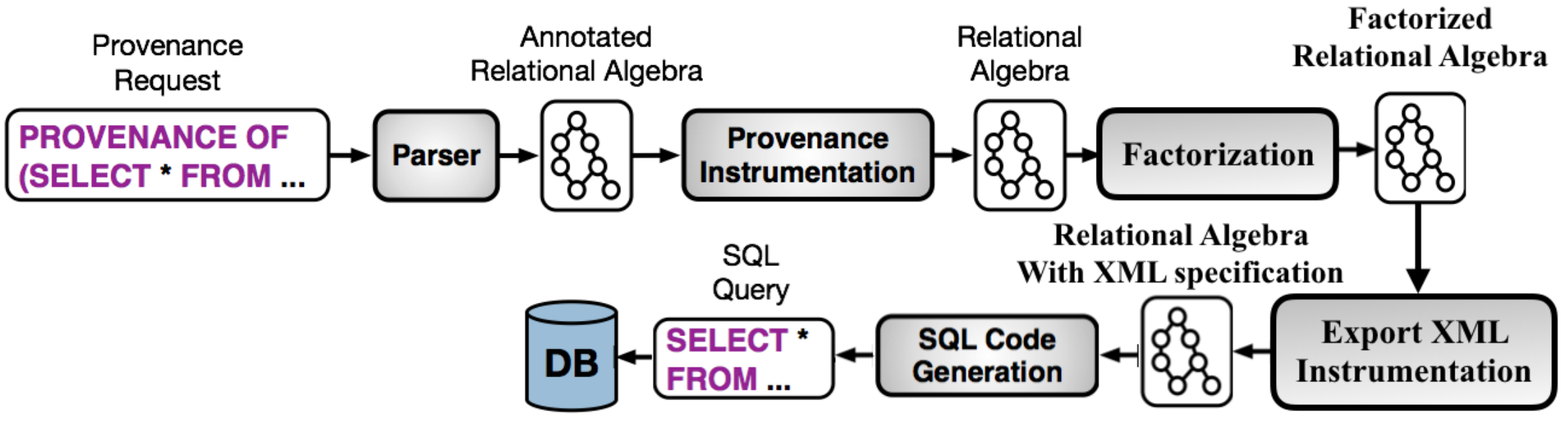}
\end{minipage}  
}\\[-0.5mm]
\subfloat[Translating temporal queries with sequenced semantics~\cite{DBLP:reference/db/BohlenJ09} into SQL queries over an interval encoding of temporal data.]{  \label{fig:temporal-rewrite-approach}
\begin{minipage}{0.98\linewidth}
  \centering
\includegraphics[width=1\columnwidth]{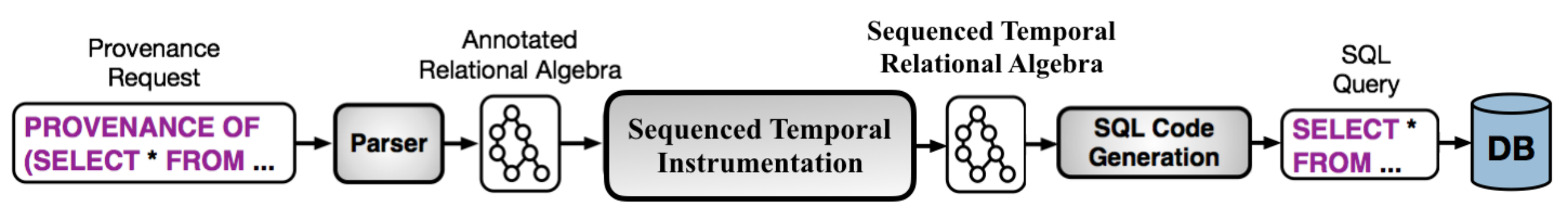}
\end{minipage}  
}\\

\caption{Instrumentation: \textbf{(a)} SQL, \textbf{(b)} transactions, \textbf{(c)} Datalog, \textbf{(d)} provenance export, \textbf{(e)} Factorized Provenance and \textbf{(f)} Sequenced Temporal Queries.}
\label{fig:all-pipelines}

\end{figure}

\begin{figure*}[t]
\centering

\begin{minipage}{0.2\linewidth}
  \centering
  \subfloat[Emp]{\label{fig:l5-ex-emp}
    \centering
            \begin{minipage}{1.0\linewidth}
              \centering
              \begin{tabular}{|c|c|c}
                \cline{1-2}
		\chead name & \chead cid & \underline{$\mathbb{N}[X]$} \\ \cline{1-2}
		 Peter & 1 & $x_1$\\  \cline{1-2}
		 Alice & 1 & $x_2$ \\   \cline{1-2}
 		 Bob & 1 & $x_3$ \\ \cline{1-2}  
    \end{tabular}
  \end{minipage}
  } 
\end{minipage}
\begin{minipage}{0.2\linewidth}
  \centering
  \subfloat[Company]{\label{fig:l5-ex-company}
    \centering
            \begin{minipage}{1.0\linewidth}
              \centering
              \begin{tabular}{|c|c|c} \cline{1-2}
		\chead cid & \chead cname & \underline{$\mathbb{N}[X]$} \\ \cline{1-2}
		 1 & IBM & $y$\\  \cline{1-2}
    \end{tabular}
  \end{minipage}
  }
\end{minipage}
\vspace{5mm}
\begin{minipage}{0.3\linewidth}
  \centering
  \subfloat[Query result (with provenance)]{\label{fig:ex-L5-q-result}
    \centering
            \begin{minipage}{1.0\linewidth}
              \centering
              \begin{tabular}{|c|l}
                \cline{1-1}
                 \chead cname &  \underline{$\mathbb{N}[X]$} \\ \cline{1-1}
		 IBM & $x_1 \cdot y + x_2 \cdot y + x_3 \cdot y$\\  \cline{1-1}
    \end{tabular}
  \end{minipage}
  }
\end{minipage}
\newsavebox{\xmlquery}
\begin{lrbox}{\xmlquery}
\begin{lstlisting}[style=psql]
SELECT cname
FROM Emp NATURAL JOIN Company
GROUP BY cname;
\end{lstlisting}
\end{lrbox}
\begin{minipage}{0.45\linewidth}
\subfloat[Query $Q_{comp}$ returning companies with at least one employee.]{\label{fig:ex-L5-q}
\hspace{1.5cm}\usebox{\xmlquery}\hspace{1.5cm}
}
\end{minipage}
\vspace{4mm}

\begin{minipage}{0.68\linewidth}
  \centering
  \subfloat[Provenance of query $Q_{comp}$ generated by Pipeline L1]{\label{fig:ex-L5-pipeline-L1-result}
    \centering
            \begin{minipage}{1.0\linewidth}
              \centering
\begin{tabular}{ |c||c|c||c|c|}

\multicolumn{1}{|c||}{\bf result} & \multicolumn{2}{c||}{\bf prov. Emp} & \multicolumn{2}{c||}{\bf prov. Company} \\
 \chead cname & \chead P(Emp,name) & \chead P(Emp,cid) & \chead P(Company,cid) & \chead P(Company,cname) \\ \hline
 IBM & Peter & 1 & 1 & IBM\\   \hline
 IBM & Alice & 1 & 1 & IBM\\   \hline
 IBM  & Bob & 1 & 1 & IBM\\   \hline
    \end{tabular}
  \end{minipage}
  }
 
\end{minipage}
\hspace{2mm}
\begin{minipage}{.26\linewidth}
  \centering
  \subfloat[Provenance of query $Q_{comp}$ as generated by Pipeline L5. The full XML document is shown in Fig.~\ref{fig:ex-L5-pipeline-L1-result-xml-no-fact}]{\label{fig:ex-L5-pipeline-L5-result}
    \centering
\begin{minipage}{1.0\linewidth}
              \centering
\begin{center}
\vspace{10mm}
\begin{tabular}{ |c|c| }
\hline
  \chead cname &  \chead prov  \\  \hline
 IBM & \lstinline[style=xmlstyle]!<add> ...! \\  \hline
\end{tabular}
\end{center}
  \end{minipage}
  }
\end{minipage}

\begin{minipage}{1.0\linewidth}
\begin{minipage}[t]{.45\linewidth}

\newsavebox{\xmlresultunfac}
\begin{lrbox}{\xmlresultunfac}
\begin{lstlisting}[style=xmlstyle]
<add>
  	<mult> 
  		<Emp> Peter, 1 </Emp>
  		<Company> 1, IBM </Company>     
  	</mult>
  	<mult> 
  		<Emp> Alice, 1 </Emp>
  		<Company> 1, IBM </Company>     
  	</mult>
  	<mult> 
  		<Emp> Bob, 1 </Emp>
  		<Company> 1, IBM </Company>     
  	</mult>
</add>
\end{lstlisting}
\end{lrbox}

\subfloat[XML representation of ``flat'' provenance polynomial $x_1 \cdot y + x_2 \cdot y + x_3 \cdot y$ generated by running Pipeline L5.]{\label{fig:ex-L5-pipeline-L1-result-xml-no-fact}
\hspace{1cm}\usebox{\xmlresultunfac}\hspace{1cm}
}
\end{minipage}
\hspace{7mm}
\begin{minipage}[t]{.5\linewidth}
\centering
\newsavebox{\xmlresultfac}
\begin{lrbox}{\xmlresultfac}
\begin{lstlisting}[style=xmlstyle]
<add>
	<mult>
		<add> 
			<Emp> Peter, 1 </Emp>
			<Emp> Alice, 1 </Emp>   
			<Emp> Bob, 1 </Emp> 
		</add>
		<Company> 1, IBM </Company>     
	</mult>
</add>
\end{lstlisting}
\end{lrbox}

\subfloat[XML representation of the factorized polynomial $(x_1 + x_2 + x_3) \cdot y$ generated by  pushing the group-by to the input employee relation based on PAT rule~(\ref{eq:group-by-push-down-preserving-join}).]{\label{fig:ex-L5-pipeline-L1-result-xml-fact}
\hspace{1cm}\usebox{\xmlresultfac}\hspace{1cm}
}
\end{minipage}
\end{minipage}

\caption{Provenance for a query $Q_{comp}$ generated by Pipeline L1 (flat encoding of provenance polynomials) and Pipeline L5 (flat and factorized encoding of provenance polynomials represented as an XML document).}
\label{fig:p5-exp}

\end{figure*}

\subsection{Pipeline L5 - Factorized Provenance}
\label{sec:pipel-l5-fact}

Pipeline L5 represents the provenance of a tuple in the result of a query as a single XML document. The XML document storing the provenance of a tuple is a nested encoding of a provenance polynomial. The result of pipeline L5 is an SQL query which computes the result of the input query and uses an additional column to store the XML document representing the provenance of a tuple. We use SQL/XML features supported by many database systems to construct such XML documents.

Consider a database with relations \textsf{Emp(name,cid)} and \textsf{Company(cid, cname)} as shown in Fig.~\ref{fig:l5-ex-emp} and~\ref{fig:l5-ex-company}. The following query (SQL code shown in Fig.~\ref{fig:ex-L5-q}) returns companies with at least one employee:

$$Q_{comp} = \Aggregation{cname}{}(Emp \join Company)$$

This query returns a single result tuple over the example database. We show the result and the provenance polynomial (annotation of the tuple in semiring $\mathbb{N}[X]$)\footnote{Note that we have taken some liberty with provenance polynomials here. The correct way of representing the result of a group-by query that preserves the generality of semiring $\mathbb{N}[X]$ requires the introduction of a $\duplicate$  operator for semiring expressions~\cite{AD11d}.} computed for $Q_{comp}$ in Fig~\ref{fig:ex-L5-q-result}. The provenance polynomial encodes that there are three alternative ways of deriving the query result (addition): joining \textsf{(Peter,1)} with \textsf{(1,IBM)} ($x_1 \cdot y$), joining \textsf{(Alice,1)} with \textsf{(1,IBM)} ($x_2 \cdot y$), and joining \textsf{(Bob,1)} with \textsf{(1,IBM)} ($x_3 \cdot y$).
Using Pipeline L1 to compute the provenance of this query, we get the result shown in Fig.~\ref{fig:ex-L5-pipeline-L1-result}. The provenance polynomial is encoded as three tuples each representing one monomial (multiplication). If Pipeline L5 is used, then a single result tuple is produced storing the full provenance of the query result tuple \textsf{(IBM)}. The result relation produced by L5 is shown in Fig.~\ref{fig:ex-L5-pipeline-L5-result}. The full XML document is shown in  Fig.~\ref{fig:ex-L5-pipeline-L1-result-xml-no-fact}.
Pipeline L5 factorizes the provenance based on the structure of the input query by constructing the polynomial annotating a query result tuple  one step at a time by instrumenting operators to combine the provenance for their inputs.
In the case of $Q_{comp}$, the query first applies a join (multiplication) and then aggregates the result of the join (addition). Thus, the XML document generated by Pipeline L5 represents a ``flat'' provenance polynomial which is a sum of products. Since factorization is determined by the query structure we can generate a different factorization by applying equivalence preserving algebraic transformations to restructure the query. For instance, assuming that \textsf{cid} is a key for relation \textsf{comp}  we can rewrite $Q_{comp}$ by pushing the aggregation into the left input of the join grouping on \textsf{cid} instead of \textsf{cname}:
 
$${Q_{fac}} = \projection_{cname}(\Aggregation{cid}{count(*)}(Emp) \join Company) $$ 

Using this restructured query, Pipeline L5 would produce the more concise XML document shown in Fig.~\ref{fig:ex-L5-pipeline-L1-result-xml-fact} which corresponds to the factorized provenance polynomial $(x_1 + x_2 + x_3) \cdot y$.

Note that the transformation we have applied here is one of our provenance-specific algebraic transformations (PAT) (Rule~\eqref{eq:group-by-push-down-preserving-join} shown in Fig.~\ref{fig:algebraic-rules-supp}). In fact, the purpose of this rule, and also of the similar PAT rules~\eqref{eq:add-duplicate-removal} and~\eqref{eq:group-by-push-down}, is to factorize the provenance generated by Pipeline L5.
An interesting avenue for future work is to see how to combine our rules with techniques that have been developed for factorized databases~\cite{OS16}.

\begin{figure*}[t]
\centering
\begin{minipage}{1\linewidth}
\begin{minipage}{0.3\linewidth}
  \centering
  \subfloat[Temporal relation Emp]{\label{fig:temp-ex-emp}
    \centering
            \begin{minipage}{1.0\linewidth}
              \centering
\begin{tabular}{ |c|c||c|}
\hline
 \chead name & \chead sal & \chead period  \\ \hline
 Peter & 30k & [3, 10) \\   \hline
Bob & 30k & [3, 15)  \\   \hline
Alice & 30k & [2, 4)  \\   \hline
Alice & 50k & [4, 9)  \\   \hline
\end{tabular}
  \end{minipage}
  } 
\end{minipage}
\begin{minipage}{0.3\linewidth}
  \centering
\subfloat[Alternative encoding of Emp]{\label{fig:temp-ex-emp-alt}
    \centering
            \begin{minipage}{1.0\linewidth}
              \centering
\begin{tabular}{ |c|c||c|}
\hline
 \chead name & \chead sal & \chead period  \\ \hline
 Peter & 30k & [3, 10) \\   \hline
Bob & 30k & [3, 10)  \\   \hline
Bob & 30k & [10, 15)  \\   \hline
Alice & 30k & [2, 4)  \\   \hline
Alice & 50k & [4, 9)  \\   \hline
\end{tabular}
  \end{minipage}
}
  \end{minipage}
\begin{minipage}{0.3\linewidth}
  \centering
  \subfloat[Normalized temporal query result]{\label{fig:temp-ex-q-result}
    \centering
            \begin{minipage}{1.0\linewidth}
              \centering
\begin{tabular}{ |c||c|}
\hline
 \chead numEmp & \chead period  \\ \hline
1 & [2, 3) \\   \hline
3 & [3, 9)  \\   \hline
2 & [9, 10)  \\   \hline
1 & [10, 15)  \\   \hline
\end{tabular}
  \end{minipage}
  } 
\end{minipage}

\caption{Example interval-timestamped temporal database and normalized sequenced query result produced by Pipeline L6.}
\label{fig:bag}
\end{minipage}

\end{figure*}

\subsection{Pipeline L6 -  Sequenced Temporal Queries}
\label{sec:l6.-sequ-temp}

We now briefly describe Pipeline L6 and discuss the instrumentation choice between set-coalesce and bag-coalesce for this pipeline (Sec.~\ref{sec:set-vs-bag}).

An important type of temporal queries are queries with so-called sequenced semantics~\cite{DBLP:reference/db/BohlenJ09}. 
Given a temporal database, a query Q under sequenced semantics returns a temporal relation that assigns to each point in time the result of evaluating Q over the snapshot of the database at this point in time.
Pipeline L6 instruments a non-temporal input query to evaluate it under sequenced semantics over an interval-timestamped encoding of temporal data.
By interval-timestamped we are referring to a common way of representing temporal data by associating each tuple with the time interval during which it is valid and storing this interval inline with the tuple.  Fig.~\ref{fig:temp-ex-emp} shows an example of such an encoding. For instance, at time $3$, Alice did earn 30k and at time $7$ she did earn 50k. There are many ways of how to represent a temporal database using interval-timestamped relations which are all equivalent in terms of the snapshots they encode. For instance, Fig.~\ref{fig:temp-ex-emp-alt} shows an alternative encoding of the \textsf{Emp} relation where Bob's salary is recorded as two tuples instead of one tuple.
To avoid having to deal with this potentially confusing ambiguity, Pipeline L6 represents query results using a unique normal form.  

For example, consider the following non-temporal query that counts the number of employees.

\begin{lstlisting}[style=psql]
SELECT count(*) AS numEmp
FROM Emp
\end{lstlisting}

Interpreted under sequenced semantics, this query will show how the number of employees changes over time. 
The result produced by Pipeline L6 for this query is shown in Fig.~\ref{fig:temp-ex-q-result}. For instance, 
from time $9$ to $10$ (exclusive) there were two employees (Peter and Bob). 

A standard method for normalizing interval-temporal data is called \textit{coalescing}~\cite{BS96}. Coalescing merges duplicates of a tuple with overlapping or adjacent time-intervals. For instance, applying coalescing to the relation from Fig.~\ref{fig:temp-ex-emp-alt} we get the relation from Fig.~\ref{fig:temp-ex-emp} because the adjacent interval for tuple \textsf{(Bob,30k)} would be merged into a time interval \textsf{[3,15)}.
Coalescing is only applicable to set semantics since it merges overlapping intervals which is not correct for bag semantics. The normalization we apply in Pipeline L6 is a generalization of coalescing for bag semantics. The details of this normalization and how we implement it though instrumentation are beyond the scope of this paper. However, for understanding the instrumentation choice discussed in Sec.~\ref{sec:set-vs-bag}  it is only important to know that 1) our implementation of set-coalescing as instrumentation is more efficient than our implementation of bag-coalescing and 2) if a query result does not contain duplicates then set-coalescing produces the same result
as bag coalescing. Based on these observations we can use set-coalescing instead of bag-coalescing whenever a query result is guaranteed to not contain any duplicates to improve query performance.
If property $\keyProp$  for the root operator of a query contains a key that does not contain any temporal attributes, then the query result is guaranteed to not contain any duplicates and we can apply set-coalescing. Note that the following less strict condition is also sufficient: If $k$ is a key for every snapshot then we can apply set-coalescing. This condition holds if the non-temporal input query result has a key. This follows from the definition of sequenced semantics, but the details are beyond the scope of this paper. For instance, our example query is an aggregation, i.e., for the non-temporal version we have $\keyProp(\rootOp) = \{\{numEmp\}\}$. Thus, set-coalescing can be used here.

\section{Background}\label{sec:supp-background}

\subsection{Definition of $\exprEval$}\label{sec:supp-back-expr-eval}
Recall that the semantics of projection expressions is defined using a function $\exprEval(t,e)$ which returns the result of evaluating $e$ over $t$. The evaluation function $\exprEval(t,e)$ is defined recursively in Fig.~\ref{fig:project-expr-eval-def}. Here $t$ is a tuple, $c$ a constant, $a$ an attribute, $e$, $e_1$, and $e_2$ are expressions, and $\compExpr$ is a comparison operator as defined in the grammar for projection expressions. For example, consider the evaluation of  projection expression $\eIf{a=3}{b \cdot 2}{a+c}$ over a tuple $t=$\textsf{(2,4,3)} with schema \textsf{(a,b,c)}.
\begin{align*}
  &\exprEval(t, \eIf{a=3}{b \cdot 2}{a+c})\\ =
  &\begin{cases}
    \exprEval(t, (b \cdot 2)) & \mathtext{if} \exprEval(t, (a=3)) \\
    \exprEval(t, (a + c)) & \mathtext{otherwise}
  \end{cases}
\end{align*}
Now to determine which of the two cases applies we have to evaluate $\exprEval(t, (a=3))$.
\begin{align*}
  \exprEval(t, (a=3)) &= (\exprEval(t, a) = \exprEval(t, 3))\\
  &= (2 = 3) = false
\end{align*}
Based on this result we can proceed with the evaluation of $\exprEval(t, \eIf{a=3}{b \cdot 2}{a+c})$.

\begin{align*}
  &\exprEval(t, \eIf{a=3}{b \cdot 2}{a+c}) \\
  =   &\exprEval(t, {a+c})\\
 = &\exprEval(t,a) + \exprEval(t,c)\\
  = &2 + 3 = 5
\end{align*}

\begin{figure}[t]
  \centering
  \begin{align*}
  \exprEval(t, c) &= c\\
  \exprEval(t, a) &= t.a\\
  \exprEval(t, e_1 + e_2) &= \exprEval(t, e_1) + \exprEval(t, e_2)\\
  \exprEval(t, e_1 \cdot e_2) &= \exprEval(t, e_1) \cdot \exprEval(t, e_2)\\
  \exprEval(t, e_1 \wedge e_2) &= \exprEval(t, e_1) \wedge \exprEval(t, e_2)\\
  \exprEval(t, e_1 \vee e_2) &= \exprEval(t, e_1) \vee \exprEval(t, e_2)\\
    \exprEval(t, \neg e) &= \neg \exprEval(t, e)\\
    \exprEval(t, e_1\, \compExpr\, e_2) &= \exprEval(t,e_1)\, \compExpr\, \exprEval(t,e_2)\\
  \exprEval(t, \eIf{e}{e_1}{e_2}) &=
  \begin{cases}
    \exprEval(t, e_1) & \mathtext{if} \exprEval(t, e) \\
    \exprEval(t, e_2) & \mathtext{otherwise}
  \end{cases}
\end{align*}

\caption{Evaluation rules for projection expressions}
\label{fig:project-expr-eval-def}
\end{figure}

\subsection{Window Operator Example}\label{sec:supp-win-op-ex}
Consider the relation \textsf{Emp(name, salary, month)} shown in Fig~\ref{fig:window-example} (a) storing the salary an employee has received for a certain month. Query

$$\Win{sum(salary)}{x}{name}{month}(Emp)$$

computes for each employee and month the total salary the employee has received up to and including this month. The result of this query  is shown in Fig~\ref{fig:window-example} (b). For instance, to compute the result for Bob in the 2nd month, we determine the partition (group) to which tuple $t$ = \textsf{(Bob, 4700,2)} belongs to. This partition contains all tuples with \textsf{name = Bob}. Within this partition tuples are sorted on their \textsf{month} value. The window for $t$ contains all tuples from the partition that have a month value that is less than or equal to $t.month = 2$. Thus, the window contains $t$ itself and the tuple \textsf{(Bob, 6000, 1)}. Computing $sum(salary)$ over this window we get $6000 + 4700 = 10700$.

\begin{figure}[t]
\centering
\begin{minipage}{0.4\linewidth}
  \centering

  \subfloat[Emp]{
    \centering
          \resizebox{0.8\linewidth}{!}{
            \begin{minipage}{1.0\linewidth}
              \centering
\begin{tabular}{ |c|c|c|}
\hline
 \chead name & \chead salary & \chead month  \\ \hline
 Bob & 5300 & 3 \\   \hline
 Alice & 6500 & 2 \\   \hline
 Alice & 5600 & 4 \\   \hline
 Bob & 6000 & 1 \\   \hline
 Bob & 4700 &  2 \\   \hline
 Alice & 6800 & 3 \\   \hline
 Alice & 5800 & 1\\   \hline
\end{tabular}
  \end{minipage}
  }
  } 
\end{minipage}
\begin{minipage}{0.5\linewidth}
  \centering
  \subfloat[Query Result]{
    \centering
          \resizebox{0.8\linewidth}{!}{
            \begin{minipage}{1.0\linewidth}
              \centering
\begin{tabular}{ |c|c|c|c|}
\hline
 \chead name & \chead salary & \chead month &  \chead x \\ \hline
 Bob & 6000 & 1 & 6000 \\   \hline
 Bob & 4700 & 2  & 10700 \\   \hline
 Bob & 5300 & 3  & 16000 \\   \hline
 Alice & 5800 & 1 & 5800 \\   \hline
 Alice & 6500 & 2 & 12300 \\   \hline
 Alice & 6800 & 3 & 19100 \\   \hline
 Alice & 5600 & 4 & 24700 \\   \hline
\end{tabular}
  \end{minipage}
  }
  } 
\end{minipage}

\caption{Example application of the window operator}
\label{fig:window-example}

\end{figure}

\section{Undecidability of Property Inference}\label{sec:supp-properties-inference-rule}

In this section, we provide the remaining proofs for the undecidability claims made in Sec.~\ref{sec:properties-inference}. In particular we claimed that the following problems are undecidable for the bag algebra we are considering in this work: computing candidate keys, determining all equivalences that hold, computing a minimal set of sufficient attributes, 
and determining whether a query is duplicate-insensitive. Note that the proof of undecidability of determining candidate keys was already shown in Section~\ref{sec:properties-inference}. 

\newcommand{\theokeys}{1}
\setcounter{Theorem}{1}
\begin{Theorem}\label{theo:set-is-undecidable}
Let $\qSub$ be a subquery of a query $\query$. The problem of deciding whether $\qSub$ is duplicate-insensitive is undecidable.
\end{Theorem}
\begin{proof}
  We prove the result using a similar reduction as used in the proof of Theorem~\ref{theo:keys-undecidable}. Given a polynomial $f$ we construct a query, such that a subquery of this query is duplicate-insensitive iff $f$ is injective. We construct $R$ as in the proof of Theorem~\ref{theo:keys-undecidable} and define $\query_f = \Aggregation{}{count(*)}(\projection_{f(x_1, \ldots, x_n) \to b}(R))$ and $\qSub = \projection_{f(x_1, \ldots, x_n) \to b}(R)$. Eliminating duplicates from the result of $\qSub$ will only affect the count iff there are any duplicates which is the case if $f$ is not injective. To see why this is the case, assume that $f$ is not injective, then there have to exist two inputs $I$ and $J$ such that $f(I) = y =  f(J)$. Consider the instance $R = \{I,J\}$ for which $\qSub$ returns $\{y^2\}$. Then $\Aggregation{}{count(*)}(\qSub)$ returns $(2)$, but $\Aggregation{}{count(*)}(\duplicate(\qSub))$ returns $(1)$. $\qSub$ can only be duplicate-insensitive if it does not contain any duplicates, because otherwise the count would decrease when these duplicates are removed.
  Thus, $\qSub$ is duplicate-insensitive if $f$ is injective and it follows that determining whether $\qSub$ is duplicate-insensitive is undecidable.
\end{proof}

\begin{Lemma}
$\aEquiv$ is an equivalence relation.  
\end{Lemma}
\begin{proof}
To prove that $\aEquiv$ is an equivalence relation we have to prove that it is reflexive, symmetric, and transitive. WLOG consider a subquery $\qSub$ of a query $\query$ and let $a,b,c \in \schema{\qSub}$. 

\myproofpar{reflexivity}
We have to show that $a \aEquiv a$. Consider the equivalence
$\selection_{a=a}(\qSub) \equiv \qSub$. This equivalence holds because $a=a$ is a tautology. Thus, trivially $\query \equiv \query[\qSub \gets \selection_{a=a}(\qSub)]$ has to hold too.

\myproofpar{symmetry}
We have to show that if $a \aEquiv b$, then also $b \aEquiv a$. The  equivalence
$\selection_{a=b}(\qSub) \equiv \selection_{b=a}(\qSub)$ follows from the symmetry of equality. Thus, $\query \equiv \query[\qSub \gets \selection_{a=b}(\qSub)]$ implies $\query \equiv \query[\qSub \gets \selection_{b=a}(\qSub)]$.

\myproofpar{transitivity}
From $\selection_{a=b}(Q) \equiv Q$ and $\selection_{b=c}(Q) \equiv Q$ follows that $\selection_{a=b}(\selection_{b=c}(Q)) \equiv Q$. Using the standard equivalence $\selection_{\theta_1}(\selection_{\theta_2}(Q)) \equiv \selection_{\theta_1 \wedge \theta_2}(Q)$ we get $\selection_{a=b \wedge b=c}(Q) \equiv Q$. Using the fact that equality is transitive we deduce that $\selection_{a=b \wedge b=c \wedge a=c}(Q) \equiv Q$. Then applying the above equivalence this implies $\selection_{a=c}(\selection_{a=b \wedge b=c}(Q)) \equiv Q$. Substituting $\selection_{a=b \wedge b=c}(Q)$ with $Q$ based on $\selection_{a=b \wedge b=c}(Q) \equiv Q$ we get $\selection_{a=c}(Q) \equiv Q$. 
\end{proof}

\begin{Theorem}\label{theo:eq-undecidable}
Let $\qSub$ be a subquery of a query $\query$ and $a,b \in \schema{\qSub}$. Determining whether $a \aEquiv b$ is undecidable.
\end{Theorem}
\begin{proof}
  We prove the claim through a reduction from query equivalence which is known to be undecidable for full relational algebra (for both sets and bags). The undecidability for bags is a corollary  of the undecidability of containment of union of conjunctive queries (UCQs) over bags~\cite{IR95}, because $\query \sqsubseteq \query'$ iff $\query - \query' \equiv \query_\emptyset$ where $\query_\emptyset$ is a query that returns the emptyset on all inputs (e.g., $R - R$ for some relation $R$). 
An alternative derivation of this result is based on the undecidability of equivalence of relational calculus queries for sets~\cite{D69}.

Consider two queries $Q_1$ and $Q_2$ and let $Q_{cntDiff} = \Aggregation{}{count(*) \to a}((Q_1 - Q_2) \union (Q_2 - Q_1))$. Note that $Q_{cntDiff}$ computes the number of tuples in the symmetric difference of $Q_1$ and $Q_2$. We claim that $\forall I: Q_{cntDiff}(I) = \{(0)\}$ iff $Q_1 \equiv Q_2$. This trivially holds, because the symmetric difference of two queries can only be empty on all inputs iff the two queries are equivalent. 
Now consider the following query $\query_{test} = \projection_{a, 0 \to b} (\query_{cntDiff})$. Based on the definition of $\aEquiv$,  $a \aEquiv b$ holds for query $\query_{test}$ iff $\query_{cntDiff}$ returns $(0)$ on all inputs. Thus, $a \aEquiv b$ iff $\query_1 \equiv \query_2$. However, since query equivalence is undecidable it follows that determining whether $a \aEquiv b$  is undecidable too.
\end{proof}

\begin{Theorem}\label{theo:icols-undecidable}
Let $\qSub$ be a subquery of a query $\query$ and let $E \subset \schema{\qSub}$. The problem of determining whether $E$ is sufficient is undecidable.   
\end{Theorem}
\begin{proof}
  We prove the claim by reduction from query equivalence. Let $\query_1$, $\query_2$, and $\query_{cntDiff}$ be as in the proof for Theorem~\ref{theo:eq-undecidable}. Furthermore, define $Q_{test} = \Aggregation{}{count(*)}(\duplicate(\qSub))$ and $\qSub = \projection_{a, 1 \to b}( Q_{cntDiff}) \union \{(0,1)\}$. We have $\schema{\qSub} = \{a,b\}$. Consider the problem of deciding whether $E = \{b\}$ is a sufficient set of attributes for  $\qSub$ within the context of $\query_{test}$. We claim that $\{b\}$ is sufficient iff $\query_1 \equiv \query_2$. If this claim holds then deciding whether $\{b\}$ is sufficient is undecidable since query equivalence is undecidable. Thus, it remains to prove that claim.

\myproofpar{$\Leftarrow$}
Assume that $\query_1 \equiv \query_2$.  We have to show that $\{b\}$ is sufficient. As discussed in the proof of Theorem~\ref{theo:eq-undecidable}, $\query_{cntDiff}$ returns $(0)$ on all inputs iff $\query_1 \equiv \query_2$. Consider $\qSub$. This query returns $\{(c, 1)^1, (0,1)^1\}$ for some $c \neq 0$ on any instance $I$ such that $\query_1(I) \neq \query_2(I)$ and $\{(0,1)^2\}$ on all instances $I$ where $\query_1(I) = \query_2(I)$. It follows that $\query_{test}(I) = \{(1)^1\}$ for instances
where $\query_1(I) = \query_2(I)$ and $\{(2)^1\}$ otherwise. Since $\query_1 \equiv \query_2$, $\qSub(I) = \{(0,1)^2\}$ and $\query_{test}(I) = \{(1)^1\}$ for all instances $I$. Let $\query_{rewr} = \query_{test}[\qSub \gets \projection_{b}(\qSub)]$. We have to show that $\query_{rewr} \equiv \query_{test}$. Since for any instance $I$, $\projection_{b}(\qSub)(I) = \{(1)^2\}$ we have $\query_{rewr}(I) = \query_{test}(I)$.

\myproofpar{$\Rightarrow$}
We prove that $\query_1 \not\equiv \query_2$ implies that $\{b\}$ is not sufficient. If $\query_1 \not \equiv \query_2$, then there has to exist an instance $I$ such that $\query_{cntDiff}(I) = \{(c)^1\}$ for some $c \neq 0$. Then $\qSub(I) = \{(c,1)^1, (0,1)^1\}$ and $\query_{test}(I) = \{(2)^1\}$. However, $\projection_{b}(\qSub)(I) = \{(1)^2\}$ and, thus, for $\query_{rewr}$ as above we have $\query_{rewr}(I) = \{(1)^1\}$. Now since, $\query_{test}(I) \neq \query_{rewr}(I)$ we have $\query_{test} \not\equiv \query_{rewr}$ and, thus, $\{b\}$ is not sufficient. 
\end{proof}

\section{Property Inference}\label{sec:supp-inference-rule}

In this section we introduce and discuss the inference rules for properties $\icolsProp$, $\keyProp$, and $\ecProp$.

\begin{table*}[p]
\renewcommand{\arraystretch}{1.4}
  \centering
  \begin{tabular}{|c|c|l|} \hline 
\rowcolor[gray]{.9} \thead{Rule} & \thead{Operator $\curOp$} & \thead{Property \textit{\icolsProp{}} inferred for the input(s) of $\curOp$}\\ \hline 
 1,2 & $ \rootOp$ or $\duplicate(R)$ & $\icolsProp{}(\rootOp) = \schema{\rootOp}$, $ \icolsProp{}(R) =  \schema{R}$
  
  \\ \hline
  
  3 & $\selection_ {\theta}(R)$ & $  \icolsProp{}(R)  = \icolsProp{}(R) \cup \icolsProp{}(\curOp) \union cols(\theta) $   
  
  \\ \hline
4 & $\projection_ {e_{1} \rightarrow b_{1},...,e_{n} \rightarrow b_{n}(R)}$ & $   \icolsProp{}(R)  = \icolsProp{}(R) \cup (\bigcup_{i \in \{1,...,n\}} cols(e_i))$ \\ \hline
  5 & $R \crossprod S$ & $   \icolsProp{}(R)  = \icolsProp{}(R) \cup (\icolsProp{}(\curOp) \intersection \schema{R}) $   \\ 
                 &  & $  \icolsProp{}(S)  = \icolsProp{}(S) \cup (\icolsProp{}(\curOp) \intersection \schema{S}) $ \\ \hline
 6 & $_{G}\aggregation _{F(a)}(R)$ & $ \icolsProp{}(R) = \icolsProp{}(R) \cup G \cup \{a \}$
\\   \hline
7, 8,9  & $R \union S$ or $R \intersection S$     or $R \difference S$ & $  \icolsProp{}(R)  = \schema{R}  $, $  \icolsProp{}(S)  = \schema{S}  $  \\ \hline
10 & $\omega_{f(a) \to x, G\|O}(R)$ &  $\icolsProp{}(R) = \icolsProp{}(R) \cup \icolsProp{}(\curOp) - \{x\} \cup \{a\} \cup G \cup O$\\ \hline
  \end{tabular}
\caption{Top-down inference of property \textit{\icolsProp{}}} 
\label{tab:top-down-icols}

\end{table*}

\begin{table*}[p]
\centering
\renewcommand{\arraystretch}{1.4}

  \begin{tabular}{|c|c|l|} \hline 
\rowcolor[gray]{.9}  \thead{Rule} & \thead{Operator $\curOp$} & \thead{Property \textit{key} inferred for operator $\curOp$}\\ \hline 
    1 & $\selection_{\theta}(R)$ & $\keyProp(\curOp) = \keyProp(R) \cup \{ k \cup \{b\} - \{a\} \mid a,b \in \schema{R} \wedge k \in \keyProp(R) \wedge a \in k \wedge \theta \Rightarrow (a=b) \}$\\ \hline
  2 & 
      $R \difference S$& $ key(\curOp) = key(R) $ 
  
  \\ \hline
     3 & $\projection_ {a_{1} \rightarrow b_{1},\ldots, a_{n} \rightarrow b_{n}}(R)$ 
                     & $ key(\curOp) = \{ \aKey[B/A] | \aKey \in key(R) \wedge \aKey \subseteq \{a_1,..., a_n\} \}$ for $A = \{a_1,\ldots, a_n\}$ and $B = \{b_1, \ldots, b_n\}$ \\                      
                     \hline                     
  4 & $R \crossprod S$ & $ key(\curOp) = \{\aKey_1 \union \aKey_2 | \aKey_1 \in key(R) \wedge \aKey_2 \in key(S)\} $   \\ \hline

    5 & $\Aggregation{G}{f(a)}(R)$  & $\keyProp(\curOp) = \minKey(\{G\} \cup \{\aKey \mid \aKey \in key(R) \wedge \aKey \subseteq G \})$ \\

   \hline
    6 & $\aggregation_{f(a)}(R)$ & $key(\curOp) = \{\{f(a)\}\}$ \\ \hline
  7 & $\duplicate(R)$ & $ key(\curOp) =  \minKey(key(R) \cup \{\schema{R}\})$ \\ \hline 
  8 & $R \union S$ & $ key(\curOp) = \emptyset $ \\ \hline 
                                                    
  9 & $R \intersection S$ & $ key(\curOp) = \minKey(key(R) \union key(S)[\schema{R}/\schema{S}])$  \\ \hline

  10 & $\omega_{f(a) \to x, G\|O}(R)$ & $ key(\curOp) = key(R)  $ \\ \hline
  \end{tabular}

\caption{Bottom-up inference of property \textit{key}} 
 \label{tab:bottom-up-key}

\end{table*}

\begin{table*}[p]
\centering
\renewcommand{\arraystretch}{1.4}
  \begin{tabular}{|c|c|l|} \hline 
\rowcolor[gray]{.9}  \thead{Rule} & \thead{Operator} $\bf \curOp$ & \thead{Property \textit{$\ecb$}  inferred for operator $\curOp$}\\ \hline 
  1 & R & $\ecb(\curOp) = \{\{a\}\mid a \in \schema{R} \}$ \\ \hline
  2 & $\selection_ {\theta}(R)$ & $\ecb(\curOp) = \ecClosure (\ecb(R) \cup 
\{\{a,b\}\mid  \theta \Rightarrow (a=b) \} 
)$

    \\ \hline                                                          
    3 & $\projection_ {a_{1} \rightarrow b_{1},...,a_{n} \rightarrow b_{n}}(R)$ & $\ecb(\curOp) = \ecClosure ( \{\{ b_i, b_j \} \mid \exists \anEC \in \ecb(R) \wedge a_i \in \anEC \wedge a_j \in \anEC \}
 \union \{\{ b_i, c \} \mid \exists \anEC \in \ecb(R) \wedge a_i \in \anEC \wedge c \in \aDom \}$\\ &&\hspace{1.7cm} $\union \{\{b_i\} | i \in \{1,\ldots,n\}\})$ \\
  \hline
 5 & $R \crossprod S$ & $\ecb(\curOp) = \ecb(R) \cup \ecb(S) $\\ \hline
 6 & $ _{G} \aggregation _{F(a)}(R) $ & $\ecb(\curOp) = \{  \anEC \cap (G \cup \aDom) \mid \anEC \in \ecb(R) \}  \cup \{\{F(a)\}\} $ \\ \hline
 7,8 & $\duplicate(R)$ or $R \difference S$ & $\ecb(\curOp) = \ecb(R) $ \\ \hline
  9 & $R \union S$ & $\ecb(\curOp) = \ecClosure (\{ \anEC \cap \anEC' \mid \anEC \in \ecb(R) \wedge \anEC' \in \ecb(S)[\schema{S}/\schema{R}] \}$ \\
     \hline
 10 & $R \intersection S$ & $\ecb(\curOp) = \ecClosure( \ecb(R) \union \ecb(S)[\schema{S}/\schema{R}])$ \\ \hline
11 & $\omega_{f(a) \to x, G\|O}(R)$ & $\ecb(\curOp) = \ecb(R) \cup \{\{x\}\}$ \\ \hline
  \end{tabular}
 
\caption{Bottom-up inference of property \textit{ec}} 
\label{tab:bottom-up}
\end{table*}

\begin{table*}[p]
\centering
\renewcommand{\arraystretch}{1.4}

\begin{tabular}{|c|c|l|} \hline 
\rowcolor[gray]{.9}  \thead{Rule} & \thead{Operator $\curOp$} & \thead{Property \textit{$\ect$} inferred for the input(s) of $\curOp$}\\ \hline 
  1,2 & $\selection_ {\theta}(R)$ or $\duplicate(R)$ & $ \ect(R, \curOp) = \ecProp(\curOp)$ 
  
  \\ \hline
 
  3 & $\projection_ {a_{1} \rightarrow b_{1},...,a_{n} \rightarrow b_{n}}(R)$ & $\ect(R, \curOp) = \ecClosure ( \{\{ a_i, a_j \} \mid \exists \anEC \in \ecProp(\curOp) \wedge b_i \in \anEC \wedge b_j \in \anEC \})$ \\ \hline  
4  & $R \crossprod S$ & $\ect(R, \curOp) =  \{\anEC-\schema{S}|\anEC \in \ecProp(\curOp)\}$  \\ 
    &    & $\ect(S, \curOp) = \{\anEC-\schema{R}|\anEC \in \ecProp(\curOp) \}$ \\ \hline                    
   
   5 &  $_G \aggregation _{F(a)}(R)$ & $\ect(R, \curOp) = \{ \anEC \cap (G \cup \aDom) |\anEC \in \ecProp(\curOp)\}  $ \\ 
   \hline
  6,7 & $R \union S$ & $ \ect(R, \curOp) = \ecProp(\curOp)$ \\ 
                &   or $R \intersection S$                                & $ \ect(S, \curOp) = \ecProp(\curOp)[\schema{R}/\schema{S}] $ \\ \hline
 8 & $R \difference S$ & $ \ect(R, \curOp) = \ecProp(\curOp)$, $\ect(S, \curOp) = \emptyset$ \\ \hline

    9 & $\omega_{f(a) \to x, G\|O}(R)$ & $ \ect(R, \curOp) = \{ \anEC \cap (G \cup \aDom) |\anEC \in \ecProp(\curOp)\}$  \\ \hline
  \end{tabular}

\caption{Top-down inference of property \textit{\ecProp}} 
\label{tab:top-down}

\end{table*}

\parttitle{Inferring the icols Property}
We compute \icolsProp{} in a top-down traversal using the inference rules shown in Tab.~\ref{tab:top-down-icols}. 
Given the undecidability of determining a minimal set of sufficient attributes, we developed rules that compute a sufficient  set of attributes which may or may not be minimal. Having forsaken minimality, we take the liberty to ignore opportunties for reducing the size of $\icolsProp$ if this unnecessarily complicates the computation (e.g., requires more than one traversal of the algebra graph). For instance, we do not consider interactions of the $\icolsProp$ with the $\keyProp$ property. 
As a general rule, a set of attributes $E$ is sufficient for an operator's input if 1) it contains all attributes that are needed to generate output attributes that are sufficient for the parents of the operator (recall that we are dealing with algebra graphs), 2) it contains all attributes needed to evaluate the operator itself (e.g., attributes used in the condition of a selection operator), and 3) projecting the output of the operator's input on this set of attributes does not affect the number of duplicates produced by the operator (which in turn could affect the result of downstream operators).

We initialize $\icolsProp$ for all operators to the empty set. Then $\icolsProp$ for the root operator $\rootOp$ of query $Q$ is set to $\schema{Q}$ since all these attributes are part of the query result (Rule 1). All input attributes are needed to evaluate a duplicate elimination operator, because removing an attribute may change the number of tuples in the result (Rule 2). All attributes from a selection's condition $\theta$ (denoted as $cols(\theta)$) are needed to evaluate the selection. Thus, all attributes needed to evaluate the ancestors of the selection plus $cols(\theta)$ are required to evaluate the selection (Rule 3). For a projection we need all attributes that are used to compute the projection expressions determining the values of attributes that are part of $\icolsProp$ for the projection (Rule 4). For crossproduct we restrict the columns needed to evaluate ancestors of the cross product to its inputs. This is correct, since the number of duplicates produced by the crossproduct are not affected by projections of its inputs (Rule 5). 
For an aggregation we need all group-by attributes to guarantee that the same number of tuples are returned even if some group-by attribute values are not accessed by ancestors of the aggregation (Rule 6). Additionally, we need the attribute over which the aggregation function is computed. For instance, for query $\Aggregation{b,c}{sum(a)}(R)$ we would set $\icolsProp(R) = \{a,b,c\}$. 
For union, intersection, and difference we need all input attributes to not affect the result, because applying a projection to only one of the inputs would cause the schema of the inputs to no longer be the same (Rules 7 to 9).\footnote{The number of duplicates produced by a bag union is not affected by additional projections. Thus, only attributes needed to evaluate ancestors of the union have to be retained. We could extend the definition of sufficient sets of attributes to deal with such cases where a projection has to be applied to both inputs. However, it is not necessary to add this additional level of complexity since our PAT rule~\eqref{eq:union-icols} presented in Appendix~\ref{sec:supp-inference-rule} pushes projections though union operations which has ultimately the same effect.} To evaluate a window operator we need all attributes that are used to compute the aggregation function, order-by ($O$) and partition-by parameters ($G$) (Rule 10).

\begin{Example}\label{ex:icols-op-inference}
Consider the following query $\query = \projection_{a + b \to x}(\Win{sum(c)}{d}{b}{a}(\selection_{a < 5}(R)))$ where $\schema{R} =  (a,b,c)$. Since the projection operator is the root of the query, we have $\icolsProp(\projection_{a+b \to x}) = \{x\}$. Proceeding with the top-down traversal, we set $\icolsProp$ for the window operator to $\{a,b\}$ since $a$ and $b$ are needed to compute the projection expressions $a+b$. For the selection we set $\icolsProp$ to $\{a,b,c\}$, because these columns are needed to compute the window operator. Finally, $\icolsProp(R) = \{a,b,c\}$, the attributes needed to evaluate the selection condition ($cols(a<5) = \{a\}$) and the ancestors of the selection. Note that attribute $d$ which stores the result of $sum(c)$ is not needed to evaluate the query result and, thus, the window operator can be removed. In fact, one of the PAT rules we introduce in Sec.~\ref{sec:heuristic} removes window operators for which the result of the aggregation function is not part of $\icolsProp$ for the window operator.
\end{Example}

\parttitle{Inferring the key Property}
We compute property \keyProp{} in a bottom-up traversal (see Tab.~\ref{tab:bottom-up-key}). 
Sometimes we may infer a super key $k$ which is a superset of another super key $k'$.
Note that any superset of a super key is also a super key. 
Thus, it would be redundant to store both $k$ and $k'$ since from $k'$ we can infer $k$.
We use a function $\minKey(K)$ where $K$ is a set of keys to remove such redundant keys. 
Function $\minKey$ is defined as:
$$\minKey(K) = \{\aKey| \aKey\in K \wedge \nexists \aKey' \in K : \aKey' \subset \aKey\}$$
For instance, $\minKey(\{a,b,c\}, \{a,b\}\}) = \{\{a,b\}\}$ because $\{a,b,c\}$ contains $\{a,b\}$.
Property $\keyProp{}$ for a relation $R$ is determined based on primary key and uniqueness constraints that hold on $R$. For most database systems  this information is available through the system catalog. For instance, if $\schema{R} = (a,b,c,d,e)$, $\{a,b\}$ is the primary key of $R$, and uniqueness constraints are defined for  $\{c,d\}$ and $\{e\}$, then $\keyProp(R) = \{ \{a,b\}, \{c,d\}, \{e\} \}$.
Any key that holds for the input of a selection is naturally also a key of the selection's output since a selection returns a subset of its input relation  (Rule 1). Furthermore, if the condition of a selection implies an equality $a=b$, then functional dependencies $a \to b$ and $b \to a$ hold. Which means that we can replace attribute $a$ with $b$ in any key $k$. For any two keys $\aKey$ and $\aKey'$ generated in this fashion, it may be the case that $\aKey \subseteq \aKey'$. We apply $\minKey$ to remove keys that contain other keys. We use a sufficient condition for checking the implication $\theta \Rightarrow (a=b)$ by transforming the condition into conjunctive normal form and then checking whether a conjunct $a=b$ exists.  A projection returns one result tuple for every input tuple. 
Thus, any key $k$ that holds over the input of a projection will hold (modulo renaming) in the projection's output unless some of the key attributes are projected out (Rule 3). To simplify the exposition we have stated Rule 3 for a projection where each projection expression is a reference to an attribute. The extension to generalized projection uses the same condition. For instance, for a projection $\projection_{a + b \to d, c \to e}(R)$ where $keys(R) = \{\{a\}, \{c\}\}$ we would infer one key $\{e\}$.
A cross product returns all combinations $(t,s)$ of tuples $t$ from the left and $s$ from the right  input. It is possible to uniquely identify $t$ and $s$ using a pair of keys from the left and right input. Thus, for any key $k_1$ for $R$ and any key $k_2$ for $S$, $k_1 \cup k_2$  is a key for the output of the crossproduct.
For aggregation operators we consider two cases: 1) aggregation with group-by and 2) without group-by. For an aggregation with group-by, the values for group-by attributes are unique in the output and, thus, are a superkey for the relation. Furthermore, all keys that are subsets of the group-by attributes are still keys in the output. Hence, if none of the keys are contained in the group-by attributes we can use the group-by attributes as a key and otherwise use all keys contained in the group-by attributes (Rule 5).  Aggregation without group-by returns a single tuple. For this type of aggregation,  the aggregation function result is a trivial key (Rule 6). The bag union of two input relations does not have a key even if both inputs have keys because we do not know whether the values for these keys overlap (Rule 8). 
The result relation computed by an intersection $R \intersection S$ is a subset of both $R$ and $S$. Thus, any key from either input is guaranteed to hold over the output (Rule 9). Of course, attributes from keys of $S$ have to be renamed. Set difference returns a subset of the left input relation. Thus, any key that holds over the left input is guaranteed to hold over the output (Rules 2). The window operator adds a new attribute value to every tuple from its input. Thus, every key that holds over the input also holds over the window operator's output (Rule 10).

\begin{Example}\label{ex:set-op-inference}
Consider the algebra graph shown below. We show property $\keyProp$ for each operator as red annotations. Assume that the primary key of relation $R$  is $\{a,b\}$ and that a unique constraint is defined for attribute $d$. Thus, $\keyProp(R) = \{\{a,b\}, \{d\}\}$. The selection enforces a condition $b=c$. Thus, in addition to the keys that hold over relation $R$ we can infer an additional key $\{a,c\}$. None of these keys is contained in each other. That is, $\minKey$ returns  $\{\{a,b\}, \{a,c\}, \{d\}\}$. The projection $\projection_{a,c}$ only retains keys that are subsets of the set of projection attributes $\{a, c\}$. It follows that $\keyProp(\projection_{a,c}) = \{\{a,c\}\}$.
  \begin{center}
  \begin{tikzpicture}
[op/.style={anchor=south},
pro/.style={red,font=\footnotesize},
conn/.style={->,line width=1pt}]

\node[op] (d2) at (0,1) {$\projection_{a,c}$};
\node[pro,right] (d2s) at (d2.east) {$\{\{a,c\}\}$};

\node[op] (s) at (0,0.3) {$\selection_{b=c}$};
\node[pro,right] (ss) at (s.east) {$\{\{a,b\}, \{a,c\}, \{d\}\}$};

\node[op] (r) at (0,-0.5) {$R$};
\node[pro,right] (rs) at (r.east) {$\{\{a,b\}, \{d\}\}$};

\draw[conn] (s) to (d2);

\draw[conn] (r) to (s);
\end{tikzpicture}
\end{center}
\end{Example}

\parttitle{Inferring Property \ecProp{}}
We compute \ecProp{} in a bottom-up traversal (Tab.~\ref{tab:bottom-up}) followed by a top-down traversal (Tab.~\ref{tab:top-down}).
In the inference rules we use an operator $\ecClosure$ that takes a set $S$ of ECs as input and merges ECs if they overlap. This corresponds to repeated application of transitivity: $a = b \wedge b = c \Rightarrow a=c$.
Operator $\ecClosure$ is defined as the least fixed-point of operator $\cal E$ shown below:\\[-12mm]
\begin{center}
\resizebox{1\linewidth}{!}{
  \begin{minipage}{1.1\linewidth}
    \begin{align*}
      {\cal E}(S) = &\{ \anEC \union \anEC' \mid \anEC \in S \wedge \anEC' \in S \wedge \anEC \cap \anEC' \neq \emptyset \wedge \anEC \neq \anEC' \} \\
                     &\cup \{ \anEC \mid \anEC \in S \wedge  \not\exists
                       \anEC' \in S: \anEC \neq \anEC' \wedge \anEC \cap \anEC' \neq \emptyset \}
    \end{align*}
  \end{minipage}
}
\end{center}

The bottom-up traversal computes an initial set of equivalences $\ecb(op)$ for each operator $op$ in the query.
The top-down inference rules propagate equivalences from parents to children, restricting them to attributes that exist in the children where necessary.
Since we are dealing with algebra graphs, an operator may have multiple parents. It is only safe to propagate an equivalence from a parent to a child if this equivalence holds for all parents of the child. The top-down rules compute $\ect(op,p)$ which stores a set of equivalences that could be propagated from parent $p$ to operator $op$ if this parent would be the only parent of $op$. To compute the set of equivalences which are propagated, we intersect the sets of the equivalences for all parents of an operator. The final set of equivalence $\ecProp(op)$ for an operator $op$ is then computed as the union of the set of equivalence determined during bottom-up inference ($\ecb(op)$) and the result of the pair-wise intersection of equivalence classes $\ect(op,p)$ for all parents:
\begin{center}
    \centering
\resizebox{1\linewidth}{!}{
  \begin{minipage}{1\linewidth}
    \begin{align*}
\ecProp(op) = \ecClosure(\ecb(op) \cup \{ \{a,b\} \mid \forall p \in \opparents(op):  \\
			\exists \anEC \in \ect(op,p): a \in \anEC \wedge b \in \anEC \}) 
    \end{align*}
  \end{minipage}
}
\end{center}
Here $\opparents(op)$ denotes the parent operators of $op$. 
In the following we first discuss the bottom-up inference rules, then the top-down inference rules, and finally present two examples of how to apply these rules.

\begin{figure*}[t]\centering
  
 \begin{center}
\begin{tikzpicture}
[op/.style={anchor=south},
conn/.style={->,line width=1pt}]

\node[op] (j) at (0,2.5) {$\intersection$};

\node[op] (p1) at (-0.5,1.8) {$\projection_{a}$};
\node[op] (p2) at (0.5,1.8) {$\projection_{c}$};

\node[op] (s1) at (-0.5,1) {$\selection_{a=3}$};
\node[op] (r1) at (-0.5,0.3) {$R$};

\node[op] (s2) at (0.5,1) {$\selection_{c=d}$};
\node[op] (r2) at (0.5,0.3) {$S$};

\draw[conn,bend left] (p1) to (j);
\draw[conn,bend right] (p2) to (j);

\draw[conn] (s1) to (p1);
\draw[conn] (s2) to (p2);

\draw[conn] (r1) to (s1);
\draw[conn] (r2) to (s2);
\end{tikzpicture}   
\begin{tikzpicture}
[op/.style={anchor=south},
pro/.style={red,font=\footnotesize},
conn/.style={->,line width=1pt}]

\node[op] (j) at (0,2.5) {$\intersection$};
\node[pro,right] (js) at (j.north east) {$\{\{a,3\}\}$};

\node[op] (p1) at (-0.5,1.8) {$\projection_{a}$};
\node[pro,left] (p1s) at (p1.west) {$\{\{a,3\}\}$};
\node[op] (p2) at (0.5,1.8) {$\projection_{c}$};
\node[pro,right] (p2s) at (p2.east) {$\{\{c\}\}$};

\node[op] (s1) at (-0.5,1) {$\selection_{a=3}$};
\node[pro,left] (s1s) at (s1.west) {$\{\{a,3\},\{b\}\}$};
\node[op] (r1) at (-0.5,0.3) {$R$};
\node[pro,left] (r1s) at (r1.west) {$\{\{a\},\{b\}\}$};

\node[op] (s2) at (0.5,1) {$\selection_{c=d}$};
\node[pro,right] (s2s) at (s2.east) {$\{\{c,d\}\}$};
\node[op] (r2) at (0.5,0.3) {$S$};
\node[pro,right] (r2s) at (r2.east) {$\{\{c\},\{d\}\}$};

\draw[conn,bend left] (p1) to (j);
\draw[conn,bend right] (p2) to (j);

\draw[conn] (s1) to (p1);
\draw[conn] (s2) to (p2);

\draw[conn] (r1) to (s1);
\draw[conn] (r2) to (s2);
\end{tikzpicture}
\begin{tikzpicture}
[op/.style={anchor=south},
pro/.style={red,font=\footnotesize},
conn/.style={->,line width=1pt}]

\node[op] (j) at (0,2.5) {$\intersection$};
\node[pro,right] (js) at (j.north east) {$\{\{a,3\}\}$};

\node[op] (p1) at (-0.5,1.8) {$\projection_{a}$};
\node[pro,left] (p1s) at (p1.west) {$\{\{a,3\}\}$};
\node[op] (p2) at (0.5,1.8) {$\projection_{c}$};
\node[pro,right] (p2s) at (p2.east) {$\{\{c,3\}\}$};

\node[op] (s1) at (-0.5,1) {$\selection_{a=3}$};
\node[pro,left] (s1s) at (s1.west) {$\{\{a,3\},\{b\}\}$};
\node[op] (r1) at (-0.5,0.3) {$R$};
\node[pro,left] (r1s) at (r1.west) {$\{\{a,3\},\{b\}\}$};

\node[op] (s2) at (0.5,1) {$\selection_{c=d}$};
\node[pro,right] (s2s) at (s2.east) {$\{\{c,d,3\}\}$};
\node[op] (r2) at (0.5,0.3) {$S$};
\node[pro,right] (r2s) at (r2.east) {$\{\{c,d,3\}\}$};

\draw[conn,bend left] (p1) to (j);
\draw[conn,bend right] (p2) to (j);

\draw[conn] (s1) to (p1);
\draw[conn] (s2) to (p2);

\draw[conn] (r1) to (s1);
\draw[conn] (r2) to (s2);
\end{tikzpicture}
\end{center}

\caption{Example application of the inference rules for property $\ecProp$}
\label{fig:ex-ec-inference}
\end{figure*}
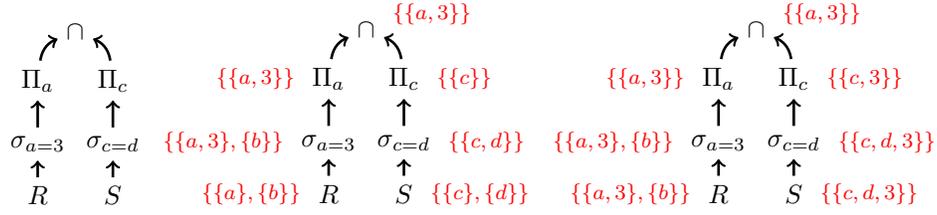

\parttitle{Bottom-up Inference}
For a relation $R$ we place each attribute in its own equivalence class (Rule 1).
For selections (Rule 2), we transform the selection condition into  conjunctive normal form, for each conjunct $a = b$ add a new EC $\{a,b\}$,  and then apply the $\ecClosure$ operator to merge the equivalence classes that contain $a$ and $b$. For instance, if $\ecb(R) = \{\{a,b\}, \{c,d\}\}$ for the input of a selection $\selection_{a=c \wedge c = 3}(R)$, then $\ecb(\selection_{a=c \wedge c = 3}) = \ecClosure(\{a,b\}, \{c,d\}, \{a,c\}, \{c, 3\}) = \{a,b,c,d,3\}$. 
Any equivalence $a \aEquiv b$ that holds over the input of a projection, also holds over its output as long as attributes $a$ and $b$ are present (potentially under different names) in the output (Rule 3). Analog $a \aEquiv c$ where $c$ is a constant holds if $a$ is present. Recall that we denote the domain of constants as $\aDom$.   Since only a subset of the attributes from an equivalence class of the input may be present, we reconstruct equivalence classes based on which attributes are present using the $\ecClosure$ operator. Note that to keep the presentation simple, we stated the rule for a projection where the all projection expressions are references to attributes. This inference rule can be applied to generalized projections by ignoring projection expressions that are not just references to attributes. For instance, for a projection $\projection_{a+b \to x, c \to y, d \to z}(R)$ where $\ecb(R) = \{\{a,b\}, \{c,d\}\}$ we would infer $\{\{y,z\}\}$ as the ECs for the projection.
All equivalences from both inputs also hold over the result of a cross product, because each output tuple is a concatenation of one tuple from the left and one tuple from the right input (Rule 4).
An equivalence $a \aEquiv b$ holds over the result of an aggregation operator if it holds in the input and the attributes $a$ and $b$ are part of the result schema which is the case if $a,b \in G$. Similar if $a \equiv c$ for $c \in \aDom$ then the equivalence holds if $a \in G$. We intersect ECs from the input of the aggregation with $G \cup \aDom$ to find all such equivalences and then apply $\ecClosure$ to compute their closure. The attribute storing the aggregation function result is placed in a new EC by itself since we cannot assume it to be equal to any of the group-by attributes (Rule 6).
Any equivalence that holds over the input of a duplicate elimination operator also holds over its output since the operator does not modify tuples (Rule 7). 
Since a difference operator returns a subset of its left input, any equivalence that holds over the left input also holds over the output (Rule 8).
The rule for union (Rule 9) renames the attributes of $S$ in $\ecb(S)$ to the attributes of $R$ which we write as $\ecb(S)[\schema{S}/\schema{R}]$. An equivalence holds over the result of a union if it holds (modulo renaming) over both inputs, because if an equivalence holds only over one of the inputs the other input may contain a tuple which does not fulfill the equivalence. Rule 9 intersects equivalence classes from both inputs (after renaming) to find equivalences that hold in both inputs and then applies the $\ecClosure$ operator to merge any overlapping equivalence classes in the result. 
For example, consider a relation $R$ with schema $\schema{R} = (a,b,c)$ where $\ecb(R)=\{\{a,b,c\}\}$ and a relation $S$ with schema $\schema{S} = (d,e,f)$ where $\ecb(S)=\{\{d,e\},\{f\}\}$. Then for a query $R \union S$, we have $\ecb(R \union S)=\{\{a,b\},\{c\}\}$.
Since any tuple in the result of an intersection has to be present in both inputs, any equivalence $a \aEquiv b$ that holds over one of the inputs also holds over the output. Rule 10 unions the set of equivalence classes for the left input and right input (after appropriate renaming) and then applies $\ecClosure$ to merge overlapping ECs.
For example, consider a relation $R$ with schema $\schema{R} = (a,b,c)$ where $\ecb(R)=\{\{a,b\}, \{c,3\}\}$ and relation $S$ with schema $\schema{S} = (d,e,f)$ where $\ecb(S)=\{\{d\},\{e,f\}\}$. For the query $R \intersection S$ we have $\ecb(R \intersection S)=\{\{a,b,c,3\}\}$.
A window operator extends each tuple from its input with a new attribute storing the result of the aggregation function. Thus, any equivalence that holds over the input also holds over the output. We cannot assume that the result of the aggregation function is equal to any of the other attributes for all inputs. Hence, we place the aggregation result attribute $x$ in its own equivalence class (Rule 11).

\parttitle{Top-down Inference}
The top-down inference rules are based on algebraic equivalences that push selections redundantly down through operators. 
The rules for selection and duplicate removal (Rules 1 and 2) propagate all equivalences to the child of the operator. Equivalences that hold over the result of a projection can be pushed to its input (Rule 3) if the attributes occurring in the equivalence exist in the input (modulo renaming). For a crossproduct we push equivalences to a child by restricting them to the schema of the child (Rule 4). We can propagate equivalences to the child of an aggregation if the equivalences are over group-by attributes and constants (Rule 5). All equivalences that hold for an intersection or union can be propagated to both children if renaming is applied to adapt the attribute names for the right input (Rules 6 and 7). For difference operators we can only  propagate equivalences to the left input (Rule 8). Finally, for a window operator we can propagate equivalences that involve partition-by attributes ($G$) and constants ($\aDom$).

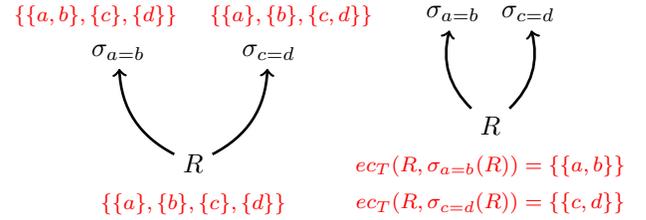
\begin{figure}[t]
\begin{center}
  \begin{minipage}{0.5\linewidth}
\begin{tikzpicture}
[op/.style={anchor=south},
pro/.style={red,font=\footnotesize},
conn/.style={->,line width=1pt}]

\node[op] (p1) at (-1,1.8) {$\selection_{a=b}$};
\node[pro] (p1s) at (-1.3,2.5)  {$\{\{a,b\},\{c\},\{d\}\}$};
\node[op] (p2) at (1,1.8) {$\selection_{c=d}$};
\node[pro] (p2s) at (1.3,2.5)  {$\{\{a\},\{b\},\{c,d\}\}$};
\node[op] (r) at (0,0.3) {$R$};
\node[pro] (rs) at (0, 0)  {$\{\{a\},\{b\},\{c\},\{d\}\}$};

\draw[conn,bend left] (r) to (p1);
\draw[conn,bend right] (r) to (p2);

\end{tikzpicture}
\end{minipage}
\begin{minipage}{0.4\linewidth}
\begin{tikzpicture}
[op/.style={anchor=south},
pro/.style={red,font=\footnotesize},
conn/.style={->,line width=1pt}]

\node[op] (p1) at (-0.5,1.8) {$\selection_{a=b}$};
\node[op] (p2) at (0.5,1.8) {$\selection_{c=d}$};
\node[op] (r) at (0,0.3) {$R$};
\node[pro] (rs1) at (0, 0)  {$\ect(R, \selection_{a=b}(R)) = \{\{a,b\}\}$};
\node[pro] (rs2) at (0, -0.5)  {$\ect(R, \selection_{c=d}(R)) = \{\{c,d\}\}$};

\draw[conn,bend left] (r) to (p1);
\draw[conn,bend right] (r) to (p2);

\end{tikzpicture}
\end{minipage}
\end{center}
\caption{Example for computing $\ecb$, the result of bottom-up traversal (left), $\ect$, the auxiliary result of top-down traversal (right), and $\ecProp$ (equal to $\ecb$ shown on the left for this example)}
\label{fig:ex-ec-bottom-up-top-down}
\end{figure}

 \begin{Example}
Consider the algebra tree shown in Fig.~\ref{fig:ex-ec-inference} (left). The result of bottom-up inference and final result after top-down inference  of the $\ecProp$ property is shown in the middle and right of this figure, respectively. During bottom-up inference the equalties enforced by the selections are incorporated into the sets of equivalence classes for the input relations $R(a,b)$ and $S(c,d)$. The projections preserve equivalences $x \aEquiv y$ for which  $x$ is projected on and either  $y$ is also an attribute in the projection result or $y$ is a constant. Equivalences from both inputs (modulo renaming) hold for the intersection. The top-down rules propagate equivalences from parents to children. Note that in this graph every operator has only one parent. Thus, the final set of equivalence classes for an operator $op$ is $\ecClosure(\ecb(op) \union \ect(op,p))$ where $p$ is the only parent of the operator. Based on the final result produced by top down inference, we know that only tuples from $R$ where $a = 3$ are of interest and for relation $S$ we only are interested in the tuple $(3,3)$.   
 \end{Example}

 \begin{Example}\label{eg:ec-bottom-up-top-down-example}
   Consider the algebra tree shown in Fig.~\ref{fig:ex-ec-bottom-up-top-down}, the $\ecb$ (bottom-up inference) and $\ect$ (top-down inference) are shown on the left and right of this figure, respectively. Note that for this particular example, the final result $\ecProp$ is same as $\ecb$ shown on the left. During bottom-up inference, the equalities enforced by the selections are incorporated into the sets of equivalence classes for the input relation $R$. The top-down rules then determine which equivalences can be propagated from a parent to its child.
   As the result of top-down inference we get $\ect(op,p)$ for an operator $op$ and one of its parents $p$. Since both selections do not have parents, nothing is inferred for these operators. For operator $R$ we get $\ect(R, \selection_{a=b}(R)) = \{\{a,b\}, \{c\}, \{d\}\}$ and $\ect(R, \selection_{c=d}(R)) = \{\{a\},\{b\},\{c,d\}\}$. As explained in the beginning of this section, it is only safe to propagate equivalences that hold for all parents. Here there are no equivalences that hold for both parents of $R$. Thus, we have $\ecProp(R) = \ecb(R)$. To see why it is unsafe to propagate equivalence that only hold for some parents consider the following instance $R = \{(1,1,2,3), (1,3,2,2)\}$. Then, $\selection_{a=b}(R) = \{(1,1,2,3)\}$ and $\selection_{c=d}(R)=\{(1,3,2,2)\}$. If we would have propagated equivalences unconditionally from parents to children, then $\ecProp(R) = \{\{a,b\}, \{c,d\}\}$. However, replacing $R$ with $\selection_{a=b}(R)$ ($\selection_{c=d}$) would affected the result of $\selection_{c=d}$ ($\selection_{a=b}$) and, thus, neither $a \aEquiv b$ nor $c \aEquiv d$ hold for $R$. 
 \end{Example}

\section{Property Inference Correctness Proofs}\label{sec:supp-correct-proof}

We now prove the correctness of our property inference rules. Recall that, as mentioned in the first paragraph of Sec.~\ref{sec:prop-inference}, the correctness criteria we are applying are based on Defs.~\ref{def:def_keys},  \ref{def:def_set}, \ref{def:def_ec}, and \ref{def:def_icols}.
In the following, we make use of the height $\qH(Q)$ of a query $Q$ which is defined as:
\begin{itemize}
\item If $Q = R$ where $R$ is a relation, then $\qH(Q) = 1$
\item If $Q = op(Q_1)$, then $\qH(Q)  = \qH(Q_1) + 1$
\item If $Q = op(Q_1,Q_2)$, then $\qH(Q) = \max(\qH(Q_1),\qH(Q_2)) + 1$.
\end{itemize}

We also define the depth $\qD(op)$ of an operator $op$ in a query $\query$ assuming that operators are uniquely identified within the context of a query. Consider an operator $op$ within a query $\query$ and let $\qSub$ be the subquery rooted at $op$. We define:

$$\qD(op) = \qH(\query) - \qH(\qSub)$$.

For example, if $Q = \projection_A(\selection_{\theta}(R))$, then  $\qH(Q)=3$. The depth of operator $\projection_A$ is $\qD(\projection_A) = \qH(\query) - \qH(\query) = 0$, $\qD(\selection_{\theta}) = \qH(\query) - \qH(\selection_\theta(R)) = 3- 2 = 1$ and $\qD(R) = 2$.

\begin{Theorem}\label{thm:icols}
  Let $Q$ be a query, $op$ an operator in $Q$, and $\qSub$ be the subquery of $Q$ rooted at $op$. 
The set $icols(op)$ 
is a sufficient set of attributes for $\qSub$. 
\end{Theorem}
\begin{proof}
Recall that according to Def.~\ref{def:def_icols},  a set of attributes $E$ is called sufficient for $\qSub$ if $Q \equiv Q[\qSub \leftarrow \projection_{E}(\qSub)]$. Thus, we have to show that for any query $\query$, subquery $\qSub$ where $op$ is the root of $\qSub$ the following equivalence holds: $Q[\qSub \leftarrow \projection_{\icolsProp(op)}(\qSub)]$. Note that $\qSub \equiv \qSub'$ implies $Q[\qSub \gets \qSub']$ for any subquery $\qSub$ of a query $\query$ and query $\qSub'$. Thus, where convenient we will prove $\qSub \equiv \projection_{\icolsProp(op)}(\qSub)$ instead of directly proving $Q[\qSub \leftarrow \projection_{\icolsProp(op)}(\qSub)]$.
We prove the claim by induction over the depth of an operator $op$. 
Note that we will prove the claim for one parent of an operator at a time.
This is correct since 1) all inference rules are monotone in the sense that they may add additional attributes to $\icolsProp(R)$, but never remove any attribute from $\icolsProp(R)$; and 2) that any superset of a sufficient set of attributes is also sufficient.
That is, if we prove that $\icolsProp(R)$ inferred based on one parent of $R$ is sufficient wrt. to this parent, then since $\icolsProp(R)$ is the union of all the sets of sufficient attributes inferred for all parents it follows that $\icolsProp(R)$ is sufficient wrt. all parents. 
The transformation shown below illustrates this argument. We start from an operator $op$ with parents $p_1$ to $p_n$. Let $\icolsProp_j$ denote the set of attributes inferred based on the rule for parent $p_j$. Note that the rules enforce that $\icolsProp(op) = \bigcup_{i \in \{1,\ldots,n\}} \icolsProp_i$. Based on the proof of each individual rule shown in the following, for each parent $p_j$ we can introduce a projection $\projection_{\icolsProp_j}$ on the path between $op$ and $p_j$ without affecting the query result.  We can then merge these individual projections into one projection on $\icolsProp(op) = \bigcup_{i \in \{1,\ldots,n\}} \icolsProp_i$.

  \begin{center}
    \begin{minipage}[c]{1.8cm}
    \begin{tikzpicture}
[op/.style={anchor=south},
conn/.style={->,line width=1pt}]

\node[op] (p1) at (-0.8,-1) {$p_1$};
\node[op] (pi) at (0,-1) {$p_j$};
\node[op] (pn) at (0.8,-1) {$p_n$};

\node[op] (op) at (0,-2) {$op$};

\draw[conn,dotted,-] (p1) to (pi);
\draw[conn,dotted,-] (pi) to (pn);

\draw[conn] (op) to (p1);
\draw[conn] (op) to (pi);
\draw[conn] (op) to (pn);
\end{tikzpicture}
\end{minipage}
\begin{minipage}[c]{0.5cm}
{\huge ${\equiv}$}
\end{minipage}
\begin{minipage}[c]{3.5cm}
    \begin{tikzpicture}
[op/.style={anchor=south},
conn/.style={->,line width=1pt}]

\node[op] (p1) at (-1.2,-1) {$p_1$};
\node[op] (pi) at (0,-1) {$p_j$};
\node[op] (pn) at (1.2,-1) {$p_n$};

\node[op] (s1) at (-1.2,-2) {$\projection_{\icolsProp_1}$};
\node[op] (si) at (0,-2) {$\projection_{\icolsProp_j}$};
\node[op] (sn) at (1.2,-2) {$\projection_{\icolsProp_n}$};

\node[op] (op) at (0,-3) {$op$};

\draw[conn,dotted,-] (p1) to (pi);
\draw[conn,dotted,-] (pi) to (pn);

\draw[conn] (s1) to (p1);
\draw[conn] (si) to (pi);
\draw[conn] (sn) to (pn);

\draw[conn] (op) to (s1);
\draw[conn] (op) to (si);
\draw[conn] (op) to (sn);
\end{tikzpicture} 
\end{minipage}
\begin{minipage}[c]{0.5cm}
{\huge ${\equiv}$}
\end{minipage}
\begin{minipage}[c]{2cm}
    \begin{tikzpicture}
[op/.style={anchor=south},
conn/.style={->,line width=1pt}]

\node[op] (p1) at (-0.8,-1) {$p_1$};
\node[op] (pi) at (0,-1) {$p_j$};
\node[op] (pn) at (0.8,-1) {$p_n$};

\node[op] (s) at (0,-2) {$\projection_{\icolsProp(op)}$};

\draw[conn,dotted,-] (p1) to (pi);
\draw[conn,dotted,-] (pi) to (pn);

\node[op] (op) at (0,-3) {$op$};

\draw[conn] (s) to (p1);
\draw[conn] (s) to (pi);
\draw[conn] (s) to (pn);

\draw[conn] (op) to (s);
\end{tikzpicture} 
\end{minipage}
\end{center}

\myproofpar{Base case}
Let $op$ be an operator of depth 0, i.e., $Q = op(Q_1)$ for some query $Q_1$ ($Q = op(Q_1,Q_2)$ if $op$ is a binary operator). We prove that $\icolsProp(op)$ is a sufficient set of attributes. 
Applying the rules from Table~\ref{tab:top-down-icols}, we get $\icolsProp(op) = \schema{Q}$. Substituting this into the correctness conditions we get: 
$Q[Q \leftarrow \projection_{\schema{Q}}(Q)] = \projection_{\schema{Q}}(Q) \equiv Q$ which trivially holds because a projection on all attributes returns its input unmodified.

\myproofpar{Inductive step} Assume we have proven that the condition of Def.~\ref{def:def_icols} holds for any operator of a query $Q$ with depth less than or equal to $n$. We have to prove that the same holds for any operator $op_{n+1}$ of depth $n+1$. Let $op_n$ denote a parent of such an operator (of depth $n$). Let $Q_{n+1}$ ($Q_{n}$) denote the subquery of $Q$ with root $op_{n+1}$ ($op_{n}$).
The set $\icolsProp(op_{n+1})$ is computed based on $\icolsProp(op_n)$ and the type of operator $op_n$.
Based on the induction hypothesis, we know that the condition of Def.~\ref{def:def_icols} holds for $op_n$.  For each operator type, we have to prove that $Q[Q_{n+1} \gets \projection_{icols(op_{n+1})}(Q_{n+1}) \equiv Q$ given that $Q[Q_n \gets \projection_{icols(op_{n})}(Q_{n}) \equiv Q]$ holds and that $op_n$ is of this type.  

\myproofpar{$op_n = \duplicate$} If $op_n = \duplicate$, then we get $icols(op_{n+1}) = \schema{Q_{n+1}}$.
Obviously this holds, because $Q_{n+1} \equiv \projection_{\schema{Q_{n+1}}}(Q_{n+1})$.

\myproofpar{$op_n = \selection$} If $op_n = \selection_\theta$, then 
we get $icols(op_{n+1}) = icols(op_n) \cup cols(\theta)$ where $cols(\theta)$ denotes the columns referenced in the selection condition $\theta$. We have to show that $$Q[ Q_n \gets \projection_{icols(op_n)}(\selection_\theta(\projection_{icols(op_n) \cup cols(\theta)}(Q_{n+1})))] \equiv Q$$ This holds, because a projection can be pushed through a selection as long as $cols(\theta)$ is retained (the condition $\theta$ is only well-defined if all attribute from $cols(\theta)$ are available).

\myproofpar{$op_n = \projection$} Consider $op_n = \projection_ {A}$ for  $ A = e_{1} \rightarrow b_{1},...,e_{n} \rightarrow b_{n}$.  
We have $\icolsProp{}(op_{n+1}) = cols(e_1) \cup ... \cup cols(e_n)$ where $cols(e_i)$ denotes the columns referenced in the expression $e_i$.
We have to show that $$Q[ Q_n \gets \projection_{icols(op_n)}(\projection_ {A}(\projection_{\icolsProp(op_{n+1})}(Q_{n+1})))] \equiv Q$$ This holds, because the result of a projection is not affected by removing attributes that are not referenced in any of its projection expressions.

\myproofpar{$op_n = \crossprod$}  Consider $op_n = Q_{left} \crossprod Q_{right}$ and let $op_{left}$ and $op_{right}$ denote the root operators of $Q_{left}$ and  $Q_{right}$, respectively. We have $   \icolsProp{}(op_{left})  = \icolsProp{}(op_n) \intersection \schema{Q_{left}} $  and $  \icolsProp{}(op_{right})  = \icolsProp{}(op_n) \intersection \schema{Q_{right}}$.
We have to prove that:
\begin{align*}
  Q[ Q_n \gets \projection_{icols(op_n)}(\projection_{ icols(op_{left})}(Q_{left}) \crossprod Q_{right})] \equiv Q\\
  Q[ Q_n \gets \projection_{icols(op_n)}(Q_{left} \crossprod \projection_{\icolsProp{}(op_{right})} (Q_{right}))] \equiv Q
\end{align*}
Since cross product is commutative, it is sufficient to prove one of these two equivalences. 
Based on the induction hypothesis we know that 
$$  Q[ Q_n \gets \projection_{icols(op_n)}(Q_{left} \crossprod Q_{right})] \equiv Q$$
The above equivalence follows from the standard algebraic equivalence shown below: $$\projection_{A}(R \crossprod S) = \projection_A(\projection_{A \cap \schema{R}}(R) \crossprod S)$$.

\myproofpar{$op_n = \aggregation$} For $op_n = \Aggregation{G}{f(a)}$ where $G = \{b_1, \ldots, b_n\}$ we have $icols(op_n) = G \cup \{a\}$. The aggregation's output  is computed based on the group-by attributes $G$ and the input $a$ alone. Thus, $\Aggregation{G}{f(a)}(\projection_{b_1, \ldots, b_n, a}(Q_{n+1})) \equiv \Aggregation{G}{f(a)}(Q_{n+1})$.

\myproofpar{$op_n = \difference$ or $op_n = \intersection$ or $op_n = \union$} Let $op_n = \difference$, $\union$, or $\intersection$, $op_{n+1}$ be either the left or the right input of the set operation, and $\qSub$ be the subquery rooted at $op_{n+1}$. We have $\icolsProp(op_{n+1}) = \schema{\qSub}$. As established for other operators above, $\schema{\qSub}$ is a sufficient set of attributes for $\qSub$.

\myproofpar{$op_n = \win$} Let  $op_n = \Win{f(a)}{x}{G}{O}(op_{n+1})$. Similar to aggregation,  to compute the output of a window operator we need the partition attributes $G$, order attributes $O$, and the input attribute $a$ for the aggregation function. Since $\icolsProp{}(op_{n+1}) = \icolsProp{}(op_n) - \{x\} \cup \{a\} \cup G \cup O$, all attributes that are need to compute $f(a)$ are present. The result then follows from the fact that a projection can be redundantly pushed through a window operator. That is, let $A \supseteq (\{a\} \cup G \cup O)$ and $x \not\in A$, then $\projection_A(\Win{f(a)}{x}{G}{O}(R)) \equiv \projection_A(\Win{f(a)}{x}{G}{O}(\projection_A(R)))$. To see why this is the case consider the definition of $\win$. Each tuple is extended with an additional attribute $x$ that stores the result of $f(a)$ over $P_t$ which is defined as $\{ (t_1.a)^n | {t_1}^n \in R \wedge t_1.G = t.G \wedge t_1 \leq_O t \} $. Note that none of the expressions used in the comprehension are affected by a projection that retains $a$, $G$, and $O$. Thus, the equivalence $\projection_A(\Win{f(a)}{x}{G}{O}(R)) \equiv \projection_A(\Win{f(a)}{x}{G}{O}(\projection_A(R)))$ holds.
\end{proof}

\begin{table*}
\centering
\renewcommand{\arraystretch}{1.4}

\begin{tabular}{|c|c|l|} \hline 
\rowcolor[gray]{.9} Rule & Operator $\curOp$ & Inferred property \textit{\setProp{}} for the input(s) of $\curOp$\\ \hline 
  1&  $\circledast$  & $\setProp{}(\circledast) = false$
  \\ \hline
  2&  $_{G}\aggregation _{F(a)}(R)$ & $ \setProp{}(R) = false$
  \\ \hline
 3 & $\selection_{\theta}(R)$ & $ \setProp{}(R) = \setProp{}(R) \wedge \setProp{}(\curOp) $ 
  
  \\ \hline
 4 & $\projection_{A}(R)$ & $  \setProp{}(R) = \setProp{}(R) \wedge \setProp{}(\curOp)  $ \\ \hline
                     
 5 &  $\duplicate(R)$ & $  \setProp{}(R) = \setProp{}(R) \wedge true $ \\ \hline
 6-9 &  $R \join_ {a=b} S$ or $R \crossprod S$ or  & $ \setProp{}(R) = \setProp{}(R) \wedge \setProp{}(\curOp) $ \\ 
  & $R \union S$ or $R \intersection S$ & $ \setProp{}(S) = \setProp{}(S) \wedge \setProp{}(\curOp)  $ \\ \hline
10 & $R \difference S$ & $ \setProp{}(R) = false $ \\ 
     &                 & $\setProp{}(S) = false  $ \\ \hline 
11 & $\omega_{f(a) \to x, G\|O}(R)$ & $\setProp{}(R) = false$ \\ \hline
  \end{tabular}  

  \caption{Top-down inference of Boolean property \textit{\setProp{}}}
\label{tab:top-down-set-supp}
\end{table*}

\begin{Theorem}\label{thm:bottom-up-ec}
Let $op$ be an operator in a query $\query$ and $\qSub$ denote the subquery rooted at $op$, then  $\ecb(op$)  is a set of equivalence classes for $\qSub$. \end{Theorem}

\begin{proof}
  To prove that a set of attributes and constants $\anEC$ is an equivalence class we have to show that $\forall a,b \in E: Q \equiv Q[\qSub \leftarrow \selection_{a=b}(\qSub)]$ holds. In the following, let $op$ denote the operator rooted at $\qSub$. Several rules make use of operator $\ecClosure$ which merges overlapping equivalence classes. Note that since $a \aEquiv b$ was shown to be transitive, the application of this operator is guaranteed to return a set of equivalence classes if its input is a set of equivalence classes. That is, it is sufficient to show that the claim holds for an input to $\ecClosure$ to prove that it holds for the output of this operator. Furthermore, observe that if we can prove that $a=b$ for all tuples in the result of a subquery $\qSub$ then this implies that $Q \equiv Q[\qSub \gets \selection_{a=b}(\qSub)]$, i.e., $a \aEquiv b$. Based on this observations we will sometimes prove that $a=b$ holds for all tuples instead of proving directly that $a \aEquiv b$. 
  We show the claim by induction over the height of a subquery $\qSub$.

\myproofpar{Base case}
Consider a query $Q$ and a subquery $\qSub$ of height $1$, i.e., $\qSub = R$ for some relation $R$. According Rule 1 from Table~\ref{tab:bottom-up},  $\ecb(op)$ contains one singleton set $\{a\}$ for each attribute $a \in \schema{R}$. WLOG consider a particular class $\{a\}$ which represents a single equivalence $a \aEquiv a$. We know that $a \aEquiv a$ since $\aEquiv$ is an equivalence relation and, thus, is reflexive.

\myproofpar{Inductive step}
Assume that the claim holds for subqueries with height up to and including $n$. We have to show that the claim holds for a subquery of height $n+1$.  
We prove this individually for each type of operator. 

\myproofpar{$op = \selection$} 
Consider a selection $\selection_{\theta}(R)$. We have $\ecb(\selection_{\theta}(R)) = \ecClosure (\ecb(R) \cup \{\{a,b\}\mid  \theta \Rightarrow (a=b) \})$. 
Consider an equivalence $a \aEquiv b$ from $\ecb(R)$. We know that $\query \equiv \query[R \gets \selection_{a=b}(R)]$. Thus, it will be sufficient to focus on tuples $t$ from $R$ for which $t.a = t.b$ holds. 
If a tuple $t$ is in the result of the selection, then the tuple also exists in the selection's input $R$. Thus, if $a \aEquiv b$ holds in $R$ then $t.a=t.b$ for any such $t$ which implies that $a \aEquiv b$ holds for $\selection_{\theta}$. Based on the definition of selection, a tuple is in the result of $\selection_{\theta}(R)$ if it exists in the input  and $t \models \theta$. Consider $\{a',b'\} \in \{\{a,b\}\mid  \theta \Rightarrow (a=b) \}$. Since, $\theta \Rightarrow (a'=b')$ from $t \models \theta$ we can deduce that $t.a' = t.b'$ and, thus, $a' \aEquiv b'$.

\myproofpar{$op = \projection$} 
Let $op = \projection_{A}(R)$ for $A = a_1 \to b_1, \ldots, a_n \to b_n$.
Then 
$\ecb(op) =  \ecClosure ( \{\{ b_i, b_j \} \mid \exists \anEC \in \ecb(R) \wedge a_i \in \anEC \wedge a_j \in \anEC \}
\union \{\{ b_i, c \} \mid \exists \anEC \in \ecb(R) \wedge a_i \in \anEC \wedge c \in \aDom \}
\union \{\{b_i\} | i \in \{1,\ldots,n\}\})$.
Note that any singleton attribute set is an equivalence class because of reflexivity. Thus, we only have to prove that the two element sets in  $\{\{ b_i, b_j \} \mid \exists \anEC \in \ecb(R) \wedge a_i \in \anEC \wedge a_j \in \anEC \}$ and in $\{\{ b_i, c \} \mid \exists \anEC \in \ecb(R) \wedge a_i \in \anEC \wedge c \in \aDom \}$ are equivalence classes. First consider $\{b_1,b_2\} \in \{\{ b_i, b_j \} \mid \exists \anEC \in \ecb(R) \wedge a_i \in \anEC \wedge a_j \in \anEC \}$ for some attributes $b_1$ and $b_2$. 
 Based on the induction hypothesis we know that $Q \equiv Q[R \gets \selection_{a_1 = a_2}(R)]$. Consider the subquery $\qSub' = \projection_A(\selection_{a_1 = a_2}(R))$ and a tuple $t$ from $\selection_{a_1 = a_2}(R)$. We know that $t.a_1 = t.a_2$. Now consider, $t' = t.A$. We have $t'.b_1 = t.a_1$ and $t'.b_2 = t.a_2$ and, thus, $t'.b_1 = t'.b_2$. This implies that $\projection_A(\selection_{a_1 = a_2}(R)) \equiv \selection_{b_1 = b_2}(\projection_A(R))$ from which follows $Q \equiv Q[\projection(R) \gets \selection_{b_1 = b_2}(\projection_A(R))]$, i.e., $b_1 \aEquiv b_2$ holds. Now consider $\{b_1,c\} \in \{\{ b_i, c \} \mid \exists \anEC \in \ecb(R) \wedge a_i \in \anEC \wedge c \in \aDom \}$ where $c$ is a constant. Then we can use a modified version of the argument used for the case with two attributes to prove that $a_1 \aEquiv c$ over $R$ implies $b_1 \aEquiv c$ for the projection.

\myproofpar{$op = \crossprod$} 
We have $ec(R \crossprod S) = \ecb(R) \cup \ecb(S)$. WLOG consider an equivalence $a \aEquiv b$ that holds over the left input. Then $Q \equiv Q[R \gets \selection_{a=b}(R)]$. Consider the modified subquery $\qSub' = \selection_{a=b}(R) \crossprod S$. Applying the standard equivalence $\selection_{\theta}(R) \crossprod S \equiv \selection_{\theta}(R \crossprod S)$ we get $\selection_{a=b}(R) \crossprod S \equiv \selection_{a=b}(R \crossprod S)$ and in turn $Q \equiv Q[(R \crossprod S) \gets \selection_{a=b}(R \crossprod S)]$ which implies $a \aEquiv b$. The proof for an equivalence that holds over the right input is symmetric.

\myproofpar{$op = \aggregation$} 
We have $ec(\Aggregation{G}{f(a)}(R)) = \{  G \cap \anEC \mid \anEC \in \ecb(R) \}  \cup \{\{f(a)\}\}$.
The singleton $\{\{f(a)\}\}$ trivially is an equivalence class.
Consider an equivalence $a \aEquiv b$ from $\ecb(R)$ where $a, b \in G$. We know that  $Q \equiv Q[R \gets \selection_{a=b}(R)]$. Since selections over group-by attributes can be pulled up through aggregations this implies that $a \aEquiv b$ holds for the aggregation. Using an analog argument we can show that an equivalence $a \aEquiv c$ where $c \in \aDom$ holds for the result of the aggregation if it holds over the input.

\myproofpar{$op = \duplicate$} 
We have $\ecb(\duplicate(R)) = \ecb(R)$. Since selections can be pulled up through duplicate elimination operators, if $a \aEquiv b$ holds for $R$ then $a \aEquiv b$ holds for $\duplicate(R)$. 

\myproofpar{$op = \union$}
We have $\ecb(R \union S) = \ecClosure (\{ \anEC \cap \anEC' \mid \anEC \in \ecb(R) \wedge \anEC' \in \ecb(S)[\schema{S}/\schema{R}] \})$. WLOG let $\schema{R} = (a_1, \ldots, a_n)$ and $\schema{S} = (b_1, \ldots, b_n)$. We can restate the above comprehension as $a_i \aEquiv a_j$ holds for $R \union S$ if $a_i \aEquiv a_j$ holds for $R$ and $b_i \aEquiv b_j$ holds for $S$. We can prove that $a_i \aEquiv a_j$ holds for $R \union S$ by applying the following standard equivalence $\selection_{a_i = a_j}(R \union S) \equiv \selection_{a_i=a_j}(R) \union \selection_{b_i=b_j}(S)$. 

\myproofpar{$op = \intersection$} 
The proof is analog to the proof for union using the equivalence $\selection_{a_i = a_j}(R \intersection S) \equiv \selection_{a_i = a_j}(R) \intersection S \equiv R \intersection \selection_{b_i = b_j}(S)$ to prove that any equivalence $a_i \aEquiv a_j$ ($b_i \aEquiv b_j$) that holds for $R$ ($S$) also holds for $R \intersection S$. 

\myproofpar{$op = \difference$} 
Similar to the proofs for union and intersection. The equivalence we are using here is $\selection_{a_i = a_j} (R - S) \equiv \selection_{a_i = a_j}(R) - S$. 
Since the result of a difference is  a subset of the result of its left child, any equivalence that holds for its left child also holds for the difference operator. 

\myproofpar{$op = \win$} 
We have $\ecb(\Win{f(a)}{x}{G}{O}(R)) = \ecb(R) \cup \{\{x\}\}$. The singleton class $\{x\}$ holds because of reflexivity. Any class $a \aEquiv b$ that holds for $R$ also holds for $op$ based on the following equivalence $\Win{f(a)}{x}{G}{O}(\selection_{a=b}(R)) \equiv \selection_{a=b}(\Win{f(a)}{x}{G}{O}(\selection_{a=b}(R)))$.  
\end{proof}

\begin{Theorem}\label{thm:top-down-ec}
Let $Q$ be a query, $\qSub$ a subquery of $Q$, and $op$ be the root operator of $\qSub$. Every  $\anEC \in \ecProp(op)$ is an equivalence class for $\qSub$.  
\end{Theorem}
\begin{proof}
  Recall that $\ecProp(op)$ is computed as
$\ecProp(op) = \ecClosure(\ecb(op) \cup \{ \{a,b\} \mid \forall p \in \opparents(op): \exists \anEC \in \ect(op,p): a \in \anEC \wedge b \in \anEC \})$ where $\ecb(op)$ is the result of bottom-up inference which we have already proven to be correct (Theorem~\ref{thm:bottom-up-ec}).  
  In the top-down inference we have to take into account that one operator may have multiple parents. Only equivalences that hold for all parents can be pushed to a child. This is encoded in the definition of $\ecProp$ by intersecting the equivalence classes for all parents of an operator. Here, $\ect(op,p)$ stores equivalences that can be pushed down to $op$ from a parent $p$. Proving the correctness of $\ecProp(op)$, thus amounts to proving that for any equivalence $a \aEquiv b$ from $\ect(op,p)$, we can push down a selection $\selection_{a=b}$ over $p$ to $op$. To see why this is sufficient to prove that claim WLOG consider an operator $op$ with parents $p_1$, \ldots, $p_n$ for $n \in \mathbb{N}$. The algebra graph fragment corresponding to this operator and its parents are shown below. If an equivalence $a \aEquiv b$ holds for a parent $p_j$ then we can add a selection $\selection_{a=b}$ on top of the parent without affecting the query result. Then based on the proof of top down inference rules  for $\ect(op,p)$ presented in the following, we can push such a selection redundantly, still preserving equivalence. Finally, if such a selection can be pushed for every parent, we can replace the individual selections with a single selections that lies on all paths between $op$ and its parents. Based on the definition of $\aEquiv$ this implies that $a \aEquiv b$ holds for $op$. 

  \begin{center}
    \begin{minipage}[c]{4cm}
    \begin{tikzpicture}
[op/.style={anchor=south},
conn/.style={->,line width=1pt}]

\node[op] (ps1) at (-1.2,0) {$\selection_{a=b}$};
\node[op] (psi) at (0,0) {$\selection_{a=b}$};
\node[op] (psn) at (1.2,0) {$\selection_{a=b}$};

\node[op] (p1) at (-1.2,-1) {$p_1$};
\node[op] (pi) at (0,-1) {$p_i$};
\node[op] (pn) at (1.2,-1) {$p_n$};

\node[op] (op) at (0,-2) {$op$};

\draw[conn,dotted,-] (ps1) to (psi);
\draw[conn,dotted,-] (psi) to (psn);

\draw[conn] (p1) to (ps1);
\draw[conn] (pi) to (psi);
\draw[conn] (pn) to (psn);

\draw[conn,dotted,-] (p1) to (pi);
\draw[conn,dotted,-] (pi) to (pn);

\draw[conn] (op) to (p1);
\draw[conn] (op) to (pi);
\draw[conn] (op) to (pn);
\end{tikzpicture}
\end{minipage}
\begin{minipage}[c]{1cm}
{\huge ${\equiv}$}
\end{minipage}
\begin{minipage}[c]{3cm}
    \begin{tikzpicture}
[op/.style={anchor=south},
conn/.style={->,line width=1pt}]

\node[op] (ps1) at (-1.2,0) {$\selection_{a=b}$};
\node[op] (psi) at (0,0) {$\selection_{a=b}$};
\node[op] (psn) at (1.2,0) {$\selection_{a=b}$};

\node[op] (p1) at (-1.2,-1) {$p_1$};
\node[op] (pi) at (0,-1) {$p_i$};
\node[op] (pn) at (1.2,-1) {$p_n$};

\node[op] (s1) at (-1.2,-2) {$\selection_{a=b}$};
\node[op] (si) at (0,-2) {$\selection_{a=b}$};
\node[op] (sn) at (1.2,-2) {$\selection_{a=b}$};

\node[op] (op) at (0,-3) {$op$};

\draw[conn,dotted,-] (ps1) to (psi);
\draw[conn,dotted,-] (psi) to (psn);

\draw[conn] (p1) to (ps1);
\draw[conn] (pi) to (psi);
\draw[conn] (pn) to (psn);

\draw[conn,dotted,-] (p1) to (pi);
\draw[conn,dotted,-] (pi) to (pn);

\draw[conn] (s1) to (p1);
\draw[conn] (si) to (pi);
\draw[conn] (sn) to (pn);

\draw[conn] (op) to (s1);
\draw[conn] (op) to (si);
\draw[conn] (op) to (sn);
\end{tikzpicture} 
\end{minipage}

\begin{minipage}[c]{1cm}
{\huge ${\equiv}$}
\end{minipage}
\begin{minipage}[c]{3cm}
    \begin{tikzpicture}
[op/.style={anchor=south},
conn/.style={->,line width=1pt}]

\node[op] (ps1) at (-1.2,0) {$\selection_{a=b}$};
\node[op] (psi) at (0,0) {$\selection_{a=b}$};
\node[op] (psn) at (1.2,0) {$\selection_{a=b}$};

\node[op] (p1) at (-1.2,-1) {$p_1$};
\node[op] (pi) at (0,-1) {$p_i$};
\node[op] (pn) at (1.2,-1) {$p_n$};

\node[op] (s) at (0,-2) {$\selection_{a=b}$};

\draw[conn,dotted,-] (ps1) to (psi);
\draw[conn,dotted,-] (psi) to (psn);

\draw[conn] (p1) to (ps1);
\draw[conn] (pi) to (psi);
\draw[conn] (pn) to (psn);

\draw[conn,dotted,-] (p1) to (pi);
\draw[conn,dotted,-] (pi) to (pn);

\node[op] (op) at (0,-3) {$op$};

\draw[conn] (s) to (p1);
\draw[conn] (s) to (pi);
\draw[conn] (s) to (pn);

\draw[conn] (op) to (s);
\end{tikzpicture} 
\end{minipage}
\end{center}

\myproofpar{Base case} 
Let $op$ be the root operator, then since the root operator of a query has no parents $\ecProp(op) = \ecb(op)$ which is a set of equivalence classes as we have already proven (Theorem~\ref{thm:bottom-up-ec}).

\myproofpar{Inductive step}
Assume operator $op$ is a parent of operator $op_{1}$ (and of $op_2$ if $op$ is binary) and consider $\ect(op_{i},op)$ which is computed based on $\ecProp{}(op)$. Furthermore, let $\qSub$ be the subquery rooted at $op$. Consider an equivalence $a \aEquiv b$ from $\ecProp(op)$. We know that $Q[\qSub \gets \selection_{a=b}(\qSub)] \equiv Q$. If we can prove that we can redundantly push this selection down to $op_i$, then $a \aEquiv b$ also holds for $op_i$. The claim then follows from standard equivalences that allow selections to be pushed through projections, selections, cross products, aggregations (if the selection is over  group-by attributes), union, and intersection. For set difference we only can push selections to the left input. For window operators it is safe to push selections on partition-by attributes.
\end{proof}

\begin{Theorem}
  Consider an operator $op$ in a query $\query$ and let $\qSub$ denote the subquery rooted at $op$. If
$\setProp(op) = true$ then $\qSub$ is duplicate-insensitive.
\end{Theorem}
\begin{proof}
For convenience we show the inference rules for property $\setProp$ in Table~\ref{tab:top-down-set-supp}.
We have to show that if $\setProp(op) = true$, then the equivalence shown below holds for the subquery $\qSub$ rooted at $op$.

$$\query \equiv \query[\qSub \gets \duplicate(\qSub)]$$

Before proving this claim, we first analyze under which conditions $\setProp(op) = true$. Recall that we initialize $\setProp(op) = true$ for all operators before applying the rules from Table~\ref{tab:top-down-set-supp}. $\setProp(op)$ is set to false for the root of the query (Rule 1). All rules with the exception of Rule 5 (duplicate elimination) either propagate $\setProp(\curOp)$ to a child or set $\setProp(op) = false$ for a child of the current operator. Rule 5 keeps the current state $\icolsProp(R)$. Thus, for an operator $op$ we have $\setProp(op) = true$ iff there is a duplicate elimination operator on every path from $op$ to the root operator of the query. Furthermore, for any such path no window, difference, or aggregation operator precedes the first duplicate elimination operator on the path. To see why this is required observe that Rule 2 which deals with  aggregation  sets  $\setProp$ to $false$ for the child of the aggregation. The same holds for Rule 11 which handles window operators and Rule 9 that deals with difference. The other rules would propagate $false$ to $\setProp(op)$ unless there is a duplicate elimination operator which prevents this. 

Based on this observation, from $\setProp(op) = true$ we can follow that there exists a duplicate elimination operator on each path from $op$ to the root operator of the query. WLOG let $\duplicate_1$, \ldots, $\duplicate_n$ be these duplicate elimination operators. If we can show that for any such operator $\duplicate_i$, the result of the operator is not affected by eliminating duplicates in the result of $op$, then this would imply that the claim $\query \equiv \query[\qSub \gets \duplicate(\qSub)]$ holds. Since, $\duplicate_i$ eliminates duplicates it suffices to show that replacing $op$ with $\duplicate(op)$ only affects the multiplicities of tuples generated by operators on the path from $op$ to $\duplicate_i$ but not what tuples are generated by these operators. Recall that $\supp(R) = \{ t \mid R(t) \geq 1\}$ is the set of tuples that have a non-zero multiplicity in relation $R$ and that $Q(I)$ denotes the result of evaluating query $Q$ over instance $I$. Consider a query $\query = op(\qSub)$ for an unary operator $op$ and let $\query' = op(\duplicate(\qSub))$. Analog define for binary operators $\query = op(\qSub_1, \qSub_2)$  and $\query' = op(\duplicate(\qSub_1), \duplicate(\qSub_2))$.
We have to prove that $\supp(Q(I)) = \supp(\query'(I))$ if the root operator $op$ of $\query$ has $\setProp(op) = true$. We only have to prove this for  operator types for which $\setProp(op) = true$ may hold which are projection, selection, crossproduct, union, and intersection.

\myproofpar{$\query = \selection(\query_1)$} A selection retains all tuples from the input which fulfill the selection condition independent of their multiplicities. Consider the definition of selection:

$$\selection _\theta (R) = \{ t^n|t^n \in R \wedge t \models \theta  \}$$

For any tuple $t$, $\query(I)(t) =n$ for $n \neq 0$ if $R(t) = n$. Then $\query'(I)(t) = 1$. 
Thus, $\supp(Q(I)) = \supp(Q'(I))$ holds.

\myproofpar{$\query = \projection_A(\query_1)$} For each input tuple $t$, the projection outputs $t.A$ with the same multiplicity as in the input. Using an argument analog to the one used for selection, $\query(I)(t) = n$ for some $n \neq 0$ iff $\query'(I)(t) = 1$ and, thus, $\supp(Q(I)) = \supp(Q'(I))$.

\myproofpar{$\query = \query_1 \crossprod \query_2$}
Based on the definition of $\crossprod$, eliminating duplicates in the input is only going to affect the multiplicity of tuples in the result, but will not affect their support. Let $\query_1(I)(t_1) = n$, $\query_2(I)(t_2) = m$, and $t = (t_1,t_2)$. Then $\query(I)(t) = n \cdot m$ and $\query'(I)(t) = 1 \cdot 1 = 1$.

\myproofpar{$\query = \query_1 \union \query_2$}
Applying duplicate elimination to an input of a union or intersection may affect the multiplicities of tuples, but does not affect the support. 
Let $\query_1(I)(t) = n$, $\query_2(I)(t) = m$ where either $n \neq 0$ or $m \neq 0$, then $\query(I)(t) = m+ n \neq 0$ and $\query'(I)(t) \in \{1,2\}$ and, thus, $\query'(I)(t) \neq 0$. 

\myproofpar{$\query = \query_1 \intersection \query_2$}
Let $\query_1(I)(t) = n$, $\query_2(I)(t) = m$ where both $n \neq 0$ and $m \neq 0$, then $\query(I)(t) = min(m,n) \neq 0$ and $\query'(I)(t) = min(1,1) = 1 \neq 0$.
\end{proof}

\begin{Theorem}
  Let $op$ be an operator in a query $\query$.
Any $\aKey \in \keyProp{}(op)$ is a superkey for the output of $op$.
\end{Theorem}

\begin{proof}
Assume that operator $op$ is a parent of operator $op_{1}$ (and of operator $op_{2}$ for binary operators). $\keyProp{}(op)$ is computed based on $\keyProp{}(op_{1})$  (and $\keyProp{}(op_2)$).  
We prove the theorem for each operator type. Some of the inference rules apply operator $\minKey$ which removes keys that are contained in other keys. Since $\minKey$ returns a subset of its inputs it suffices to prove that the input to $\minKey$ is a set of keys to demonstrate that its output is a set of keys.  

\myproofpar{$op = \selection$} Selection returns a subset of its input. Thus, any key that holds for the input has to hold over the output. Furthermore, if a key $k = \{a, b_1, \ldots b_n\}$ holds on $R$ and the selection condition $\theta$ implies an equality $a=c$, then this implies that a key $k' = \{c, b_1, \ldots, b_n\}$ holds on the output of the selection. We prove this implication as follows. Since $k$ and $k'$ only differ in $a$ ($c$) and we know that  $a=c$ for any tuple $t$ in the result of the selection we have $t.k = t.k'$. Recall that a set of attributes $k$ is a super key if two conditions hold: 1) $\forall t, t' \in Q(I): t.k = t'.k \Rightarrow t = t'$ and 2) $\forall t: Q(I)(t) \leq 1$. Since $k$ is a key and $a=c$ for any tuple in the result of the selection, we have $\forall t, t': t.k' = t'.k' \Rightarrow t = t'$. It remains to be shown that the second condition holds. Note that a selection returns a subset of its input. Thus, if $\forall t: Q(I)(t) \leq 1$ holds on $R$ it also has to hold over the output of the selection. Thus, $k'$ is a super key for the output of the selection. 

\myproofpar{$op \in \{\projection, \crossprod\}$} The correctness of the rules for these operators were already proven in~\cite{SS96b}.

\myproofpar{$op = \aggregation$} Since the values of group-by attributes in an output of an aggregation are unique, the set of group-by attributes is a super key for the aggregation's output. Likewise, any subset of the group-by attributes that is a key for $op_1$ is also a key for the aggregation's output.

\myproofpar{$op = \duplicate$} Duplicate elimination returns a subset of its input. Thus, all  keys of the input are also keys of the output.

\myproofpar{$op = \union$} Even if a set $\aKey$ is a key for both inputs, it is not guaranteed that $\aKey$ is a key for the output of the union since a tuple $t$ may exist in both inputs. Since $\keyProp(R \union S) = \emptyset$, the inference rule is trivially correct. 

\myproofpar{$op = \intersection$} 
Intersection returns a subset of the both inputs. Thus, any key that holds for one of the inputs also holds over the result of the intersection.

\myproofpar{$op = \difference$} Set difference returns a subset of its left input, thus, any key that holds for the left input has to hold over the output.

\myproofpar{$op = \win$} For each input tuple $t$, the window operator returns a tuple $(t,x)$ where attribute $x$ stores the result of the aggregation function $f$. Thus, any key that holds over the input of a window operator also holds over the output of the window operator.   
\end{proof}

\begin{figure*}[t]
  \centering

\begin{minipage}{0.46\linewidth}
\begin{equation}\label{eq:pulling-up-provenance-projections}
  \frac{ a \subseteq \schema{\Diamond(\projection_{A}(R))} 
  }{\Diamond(\projection_{A,a \to b}(R)) \to  
     \projection_{\schema{\Diamond(\projection_{A}(R))},a \to b}(\Diamond(\projection_A (R)))}
\end{equation}
\end{minipage}
\hspace{1cm}
\begin{minipage}{0.15\linewidth}  
\begin{equation}\label{eq:duplicate-remove}
    \frac{keys(R) \neq \emptyset}{\duplicate (R) \rightarrow R} 
\end{equation}
\end{minipage}
\begin{minipage}{0.15\linewidth}  
\begin{equation}\label{eq:duplicate-remove-set}
    \frac{set(\duplicate(R))}{\duplicate (R) \rightarrow R} 
\end{equation}
\end{minipage}
\begin{minipage}{0.16\linewidth}  
\begin{equation}\label{eq:remove-redundant-columns1}
 \frac{A=icols(R)}{R \rightarrow \projection_A (R)}
\end{equation}
\end{minipage}
\hspace{1cm}
\begin{minipage}{0.41\linewidth}  
\begin{equation}\label{eq:add-duplicate-removal}
 \frac{G \subseteq \schema R}{\Aggregation{G}{}(R \join_ {b=c} S)  \rightarrow  \Aggregation{G}{}(\Aggregation{G,b}{} (R) \join_{b=c} S)}
\end{equation}
\end{minipage}
\hspace{1cm}
\begin{minipage}{0.45\linewidth}  
\begin{equation}\label{eq:attribute-factoring}
 \frac{e_1 = \eIf{\theta}{a + c}{a}}
{\projection_{e_1,...,e_m}(R) \to \projection_{ a +  \eIf{\theta}{c}{0}, e_2,...,e_m}(R)}
\end{equation}
\end{minipage}\\[1mm]
\begin{minipage}{0.28\linewidth}  
\begin{equation}\label{eq:window-function}
 \frac{x \not\in icols(\Win{f(a)}{x}{G}{O}(R))}{\Win{f(a)}{x}{G}{O}(R) \rightarrow R}
\end{equation}
\end{minipage}
\begin{minipage}{0.67\linewidth}  
\begin{equation}\label{eq:group-by-push-down}
 \frac{a \in \schema R \wedge a \not\in (G \union \{b, c\}) \wedge b \in G \wedge G \subseteq \schema R \wedge \{c\} \in keys(S)}{_{G} \aggregation _{F(a)}(R \join_ {b=c} S)  \rightarrow \Aggregation{G}{f(a)}(R) \join_{b=c} S}
\end{equation}
\end{minipage}
\\
\begin{minipage}{0.5\linewidth}  
\begin{equation}\label{eq:group-by-push-down-preserving-join}
  \frac{a \in \schema{R} \wedge \{c\} \in \keyProp(S) \wedge g \in \schema{S}}
{\Aggregation{g}{f(a)}(R \join_{b=c} S)  \rightarrow \projection_{g,f(A)}(\Aggregation{baq}{f(a)}(R) \join_{b=c} S)}
\end{equation}
\end{minipage}\\
\begin{minipage}{0.31\linewidth}  
\begin{equation}\label{eq:selection-move-around-1}
 \frac{\exists E \in EC(R) \wedge a \in E \wedge b \in E}{R \rightarrow \selection_{a=b}(R)}
\end{equation}
\end{minipage}
\hspace{0.1cm}
\begin{minipage}{0.31\linewidth}  
\begin{equation}\label{eq:selection-move-around-2}
 \frac{\exists E \in EC(R) \wedge a \in E \wedge b \in E}{\selection_{\theta}(R) \rightarrow \selection_{\theta}(\selection_{\theta[a/b]}(R))}
\end{equation}
\end{minipage}
\hspace{0.1cm}
\begin{minipage}{0.30\linewidth}  
\begin{equation}\label{eq:merge-selection}
    \selection_ {\theta_1} (\selection_ {\theta_2} (Q))  \rightarrow          
    \selection_ {\theta_1 \wedge \theta_2}(Q)
\end{equation}
\end{minipage}
\\[2mm]
\begin{minipage}{0.9\linewidth}  
\begin{equation}\label{eq:merge-projection}
   \begin{aligned}
     &\projection_{e_{1} \rightarrow a_{1},...,e_{n} \rightarrow a_{n}}(\projection_{e'_{1} \rightarrow b_{1},...,e'_{m} \rightarrow b_{m}}(Q)) \rightarrow 
        &\projection_{e_{1}[b_{1}/e'_{1},...,b_{m}/e'_{m}] \rightarrow a_{1},...,} 
        &_{e_{n}[b_{1}/e'_{1},...,b_{m}/e'_{m}] \rightarrow a_{n}}(Q)
  \end{aligned}
\end{equation}
\end{minipage}
\\[2mm]
\begin{minipage}{0.8\linewidth}  
\begin{equation}\label{eq:union-icols}
\projection_{A} (R \union S) \rightarrow \projection_{A}(R) \union \projection_{A}(S)[\schema{S}/\schema{R}]
\end{equation}
\end{minipage}

  \caption{Provenance-specific transformation (PAT) rules}
  \label{fig:algebraic-rules-supp}
\end{figure*}

\section{Additional PATs and Correctness Proofs}\label{sec:supp-heuristic}

We did introduce a subset of the PAT rules we support in Sec.~\ref{sec:heuristic}.
We now introduce additional PAT rules (shown in this Appendix, Fig.~\ref{fig:algebraic-rules}), prove the correctness of the full set of rules, and then discuss how these rules address the performance bottlenecks discussed in Sec.~\ref{sec:motivation}. Recall that rules are of the form $\frac{pre}{q \rightarrow q'}$ which has to be read as ``If condition $pre$ holds, then $q$ can be rewritten as $q'$''.

\parttitle{Selection Move-around}\label{sec:PAT-selection-move-around}
Rule~\eqref{eq:selection-move-around-1} and~\eqref{eq:selection-move-around-2} are selection move-around rules.
Selection move-around, a  generalization of the textbook selection-pushdown equivalence, enables us to introduce selections to reduce the size of intermediate results. 
In Rule~\eqref{eq:selection-move-around-1}, if attributes $a$ and $b$ both belong to the same equivalence class of $E \in \ecProp(R)$, then based on Theorem~\ref{thm:top-down-ec} we can introduce a new selection $\selection_{a=b}$ over $R$. For example, if  $\schema{R} = (a,b)$ and $\schema{S} = (c,d)$, then for a query $\selection_{b=5}(R) \join_{b=c} S$ we get $\ecProp(S)=\{\{c,5\},\{d\}\}$. Applying Rule~\eqref{eq:selection-move-around-1} we can replace $S$ with $\selection_{c=5}(S)$. 

For Rule~\eqref{eq:selection-move-around-2} consider two attributes $a$ and $b$ that belong to the same equivalence class $E \in \ecProp(R)$. Furthermore, consider a selection $\selection_\theta(R)$. Since $a \aEquiv b$, we can replace any reference to attribute $a$ with attribute $b$ in $\theta$ (written as $\theta[a/b]$ to get a selection condition $\theta'$ which is equivalent to $\theta$ wrt. to evaluating the query $\query$ that contains this selection.  For example, consider a query $\selection_{b<5}(R)$ where  $\schema{R} = (a,b)$ and  $\ecProp(R)=\{\{a,b\}\}$, then Rule~\eqref{eq:selection-move-around-2} would introduce an additional selection to transform $\query$ into   $\selection_{b<5}(\selection_{a<5}(R))$.

\parttitle{Merge adjacent Projections and Selections}\label{sec:PAT-merge-porj-select}
Rule~\eqref{eq:merge-projection} merges adjacent projections which is a standard relational algebra equivalence. If two projection operators are adjacent, we can merge them into one projection operator by substituting references to attributes in the outer projection with the expressions from the inner projection that define these attributes. The purpose of this rule is to simplify the query and  open up opportunities for further optimization (e.g., removing redundant projections). Note that in constrast to most database systems we do a safety check before applying this rule to avoid a potential exponential blowup in projection expression size.

\begin{Example}
  Consider a query
  $$\query = \projection_{c + c \to d}(\projection_{b + b \to c}(\projection_{a + a \to b}(R)))$$
To merge these projections we have to replace references to attributes with the expression defining them. While in $\query$ every projection references attributes from its input twice, after merging projections we get a projection expression with $2^3$ references to attribute $a$:
  $$\projection_{a+a + a+ a + a+ a+ a+ a}(R)$$
\end{Example}

While the example above may be contrived, we faced such blow-ups in expression size when generating queries that capture the provenance of updates and transactions~\cite{AG17,AG17c}.
Whenever merging projections results in a superlinear increase in expression size, we do not merge the projections. In fact, we will force the database system to materialize the intermediate results to prevent it from merging these projections. 
We use $e[x/y]$ to denote replacing each occurrence of expression $x$ (usually an attribute) in $e$ with expression $y$.
For example, consider the query $\projection_{a + b \to c} (\projection_{a, d + e \to b}(R))$. Merging projections we get:
$\projection_{a + (d + e) \to c}(R)$. In the inner projection, $d + e$
  is renamed to $b$. Hence, if we merge the projections, then $b$ should be
  replaced with $(d + e)$.
  
Rule~\eqref{eq:merge-selection} merges adjacent selections. This is also a standard equivalence rule. When two selections are adjacent, we can replace them with a single selection on the conjunction of the conditions of the two selections. The purpose of this rule is also to simply the query, e.g., after introducing a new selection based on one of the selection move-around rules. 

\parttitle{Pushing Projections through Union}\label{sec:PAT-remove-col-union}
Most standard  projection push-down rules are handled by Rule~\eqref{eq:remove-redundant-columns1}. However, as mention in Appendix~\ref{sec:supp-inference-rule}, $\icolsProp$ does not allow us to push projections to children of a union operator. Thus, we introduce a separate PAT rule (Rule~\eqref{eq:union-icols}) for this purpose.

\begin{Theorem}
The PATs from Fig.~\ref{fig:algebraic-rules-supp} are equivalence preserving.  
\end{Theorem}
\begin{proof}
\myproofpar{Rule~\eqref{eq:pulling-up-provenance-projections}}
The value of attribute $b$ is the same as the value of $a$ because the projection expression determining $b$ is $a \to b$. Since $b$ is not needed to evaluate $\Diamond(\projection_A(R))$, we can delay the computation of $b$ after $\Diamond$ has been evaluated.

\myproofpar{Rule~\ref{eq:duplicate-remove}}
Since $keys(R) \neq \emptyset$, by Def.~\ref{def:def_keys} it follows that no duplicate tuples exist in $R$ ($R(t) = n \rightarrow n \leq 1$). Thus, we get $\query \equiv \query[R \gets \duplicate(R)]$.

\myproofpar{Rule~\ref{eq:duplicate-remove-set}}
Recall that we have proven that if $\setProp(op) = true$ then only the support $\supp$ of the result of $op$ (what tuples are part of the result), but not the multiplicities of tuples in the result, affect the results of ancestors of $op$. Thus, if $\setProp(op) = true$ where $op$ is a duplicate elimination operator, then we can safely remove this operators since this will not affect the support.

\myproofpar{Rule~\ref{eq:remove-redundant-columns1}}
Suppose $A=icols(R)$, by Def.~\ref{def:def_icols} we get $Q[R \gets \projection_A (R)] \equiv Q$.

\myproofpar{Rule~\eqref{eq:attribute-factoring}}
Let ${e_1}' = (A + \eIf{\theta}{c}{0})$.
We distinguish two cases: 1) if $\theta$ holds, then both $e_1$ and ${e_1}'$ evaluate to $A+c$; 2) otherwise both $e_1$ and ${e_1}'$ evaluate to $A$. Thus, these two expressions are equivalent and replacing $e_1$ with ${e_1}'$ in a projection is an equivalence preserving transformation.

\myproofpar{Rule~\ref{eq:window-function}}
From $x \not\in icols(\omega_{f(a) \to x} (R))$ follows $Q[\omega_{f(a) \to x} (R) \gets \projection_{\schema{R}}(\omega_{f(a) \to x} (R))] \equiv Q$. Based on the definition of $\win$ it follows that $t^n \in \projection_{\schema{R}}(\omega_{f(a) \to x} (R)) \leftrightarrow t^n \in R$. Thus, $Q[\omega_{f(a) \to x} (R) \gets R] \equiv Q$.

\myproofpar{Rule~\eqref{eq:selection-move-around-1}}
If $\exists E \in \ecProp(R) \wedge a \in E \wedge b \in E$ then $a \aEquiv b$ and by Def.~\ref{def:def_ec} we get $Q[R \gets \selection_{a=b}(R)] \equiv Q$.

\myproofpar{Rule~\eqref{eq:selection-move-around-2}}
If $\exists E \in \ecProp(R) \wedge a \in E \wedge b \in E$ then $a \aEquiv b$ and by Def.~\ref{def:def_ec} we get $Q[\selection_{\theta}(R) \gets \selection_{a=b}(\selection_\theta(R))]$.
Recall that $\theta[a/b]$ denotes replacing $a$ with $b$ in condition $\theta$.
Since $\selection_{a=b}(\selection_\theta(R)) \equiv \selection_{a=b \wedge \theta}(R) \equiv \selection_{\theta}(\selection_{\theta[a/b]}(R))$ we have $Q[\selection_{\theta}(R) \gets \selection_{\theta}(\selection_{\theta[a/b]}(R))]  \equiv Q$.

\myproofpar{Rule~\eqref{eq:add-duplicate-removal} and~\eqref{eq:group-by-push-down} }
These two rules were introduced in~\cite{chaudhuri1994including}.

\myproofpar{Rule~\eqref{eq:merge-selection},~\eqref{eq:merge-projection} and~\eqref{eq:union-icols}} 
These three rules are fairly standard and generally accepted to be correct.
\end{proof}

\section{Alternative CBO Search Strategies}\label{sec:supp-sty}

In Sec.~7, we introduced our \textbf{\textit{sequential-leaf-traversal}} strategy which  traverses the plan space tree  from the leftmost leaf to the rightmost leaf. For example, consider the 
plan tree shown in Fig.~\ref{fig:plan-tree-example} in Sec.~\ref{sec:cbo}. Our algorithm generates plans in the following sequence  \texttt{[0,0]}, \texttt{[0,1]},\texttt{[1,0,0]}, \texttt{[1,0,1]}, \texttt{[1,1,0]}, \texttt{[1,1,1]}. The \textit{sequential-leaf-traversal} strategy enumerates all possible plans. 
However, if the plan space is large this strategy may spend more time on optimization than on query execution. To address this potential shortcoming, we also explore plan enumeration strategies that only explore parts of the full plan space.

One common approach for dealing with large search spaces is to apply metaheuristics such as simulated annealing and genetic algorithms. Metaheuristics have a long tradition in query optimization, e.g., some systems apply metaheuristics for join enumeration once the number of joins exceeds a threshold~\cite{SM97} or for cost-based transformations~\cite{AL06}. As an example of metaheuristics we implemented the \textbf{\textit{Simulated Annealing}} algorithm. We will discuss this algorithm in Appendix~\ref{sec:simulated-annealing}.

Another option is to apply our \textbf{\textit{adaptive}} strategy introduced in Sec.~\ref{sec:cbo} which balances optimization time and execution time by stopping optimization once a ``good enough'' plan has been found.
However, as already mentioned in Sec.~\ref{sec:cbo} the traversal order of our \textit{sequential-leaf-traversal} strategy is not suited well for the adaptive strategy because it only explores the part of the plan space corresponding to a prefix of the sequence of leafs traversed in left-to-right order. To increase the diversity of plans explored by the \textit{adaptive} strategy, we developed the \textbf{\textit{binary-search-traversal}} strategy which traverses the plan space simulating a binary search over the leaf nodes in left-to-right order instead of traversing the plan space sequentially in left-to-right order. This strategy will be discussed in more detail in Appendix~\ref{sec:binary-search-traversal}. Finally, in Appendix~\ref{sec:2-competitive} we prove that our adaptive strategy is 2-competitive. For completness we also show the \textit{sequential-leaf-traversal} strategy here (Appendix~\ref{sec:supp-sequ-leaf-trav}).

\subsection{Sequential-Leaf-Traversal}
\label{sec:supp-sequ-leaf-trav}

The \textit{sequential-leaf-traversal} strategy traverses the leafs of the plan tree in left-to-right order.
If $p_{cur}$ is the path explored in the previous iteration, then  taking the next available choice as late as possible on the path will lead to the next node at the leaf level.  
Let $p_{next}$ be the prefix of $p_{cur}$ that ends in the new choice to be taken. 
If following $p_{next}$  leads to a path that is longer than $p_{next}$, then after making $len(p_{next})$ choices the first option should be chosen for the remaining choice points.

We use square brackets to denote lists, e.g., $[0,1]$ denotes a list with elements $0$ and $1$. We use $[]$ to represent an empty list.
$L \gets L \cList e$ denotes appending element $e$ to list $L$. Functions $\Call{popHead}{L}$ and $\Call{popTail}{L}$ remove and return the first (respective last) element of list $L$.

\parttitle{The makeChoice Function} Algorithm~\ref{alg:default-callback} shows the default $\Call{makeChoice}{}$ function. 
If possible we pick the next predetermined choice from list $p_{next}$.
If list $p_{next}$ is empty, then we pick the first choice ($0$). In both cases, we append the choice to the current path and the number of available options for the current choice point is appended to list $n_{opts}$.

\parttitle{Determining Choices for the Next Iteration} 
Algorithm~\ref{alg:default-gen-next-option} determines which options to pick in the next iteration. We copy the path from the previous iteration (line 2) and then repeatedly remove elements from the tail of the path and from the list storing the number of options ($n_{ops}$) until we have removed an element $c$ for which at least one more alternative exits ($c + 1 < n_{ops}$). Once we have found such an element we append $c+1$ as the new last element to the path.

\begin{algorithm}[t]
  \caption{Default \textsc{makeChoice} Function}
  \label{alg:default-callback}
  \begin{algorithmic}[1]
    
    \Procedure{makeChoice}{$numChoices$}
      \If {$len(p_{next}) > 0$}
        \State $choice \gets \Call{popHead}{p_{next}}$
      \Else
        \State $choice \gets 0 $
      \EndIf
        \State $p_{cur} \gets p_{cur} \cList choice$
        \State $n_{opts} \gets n_{opts} \cList numChoices$
      \State \textbf{return} choice
    \EndProcedure
        
  \end{algorithmic}
\end{algorithm}

\begin{algorithm}[t]
  \caption{Default \textsc{genNextIterChoices} Function}
  \label{alg:default-gen-next-option}
  \begin{algorithmic}[1]
    
    \Procedure{genNextIterChoices}{ }
      \State $p_{next} \gets p_{cur}$
      \For {$i \in \{ len(p_{next}), \ldots, 1\}$} 
        \State $c \gets \Call{popTail}{p_{next}}$
        \State $nops \gets \Call{popTail}{n_{opts}}$
        \If{$c+1 < nops$}
          \State $c \gets c+1$
          \State $p_{next} \gets p_{next} \cList c$
          \State \textbf{break}
        \EndIf
      \EndFor    
      \State $p_{cur} \gets []$
      \State $n_{opts} \gets []$
    \EndProcedure
    
  \end{algorithmic}
\end{algorithm}

\subsection{Simulated Annealing}\label{sec:simulated-annealing}
\textit{Simulated Annealing} is a metaheuristic, i.e., a randomized, guided search, which tries to find a global optimum in a large search space. The method starts from a random plan and traverses the plan space by randomly applying transformations. In each step, it applies a random transformation to the previous plan $P_{pre}$ to derive a new plan $P_{cur}$  (let $C_{cur}$ and $C_{pre}$ denote the costs of these plans). If $C_{cur} < C_{pre}$, i.e., the current plan has a lower cost than the previous plan, then simulated annealing will set $P_{pre} = P_{cur}$. 
If $C_{cur} \geq C_{pre}$, then the choice of whether to proceed into the direction of the new plan $P_{cur}$ is made probabilistically based on $C_{cur} - C_{pre}$ (how much worse is the new plan) and a parameter called the temperature $temp$ that is decreased over time based on the so-called cooling rate (\textit{cr}).
Initially, the probability to choose an inferior plan is higher to avoid getting stuck in a local minima early on. By decreasing the temperature (and, thus also probability) over time, the approach will converge eventually.
The probability $p$ of updating $P_{pre} = P_{cur}$ is computed as shown below:

$$p =
\begin{cases}
  e^{\frac {c \times (C_{pre} - C_{cur})}{temp \times C_{cur}}} &\mathtext{if} C_{pre} < C_{cur}\\
  1 & \mathtext{otherwise}
\end{cases}
$$

Here, $c$ is a constant used to scale the cost difference appropriately. In each iteration the temperature $temp$ is updated based on the cooling rate $cr$. The speed of convergence and, thus, time spend on optimization is determined by setting the cooling rate. 
We implemented simulated annealing  by implementing appropriate $\Call{continue}{}$, $\Call{genNextIterChoices}{}$, and $\Call{makeChoice}{}$ functions. The $\Call{continue}{}$ function stops exploration once the temperature has reached $0$. The $\Call{genNextIterChoices}{}$ function takes the current plan represented as a path in the plan tree, randomly chooses a position in the path and increases or decreases the current choice in this position. For any additional choice points hit after this position, a random option is chosen. For instance, assume that the current plan corresponds to a path \texttt{[1,0,3,1,3]} in the plan tree. We choose a position between 0 and 4, and then either increment or decrement this position. For instance, we may choose to decrement position 2 resulting in the path \texttt{[1,0,2]}. Note that this path may not end in a leaf node. To expand it to a full path we then randomly choose an option for any additional choice point that is hit after the 3rd choice. For instance, if there are 3 additional choice points that are hit after \texttt{[1,0,2]} and each of them has 2 options to choose from, then we may generate a new plan \texttt{[1,0,2,0,1,0]}. This is implemented in the  $\Call{makeChoice}{}$ function which first selects choices based on the prefix determined by $\Call{genNextIterChoices}{}$ (e.g., \texttt{[1,0,2]} in our example) and then proceeds to randomly choose options.

\subsection{Binary Search Traversal}\label{sec:binary-search-traversal}

We now discuss our \textit{binary-search-traversal} strategy which approximates a binary search over the leaves of a plan tree. The method maintains a queue $todo$ of intervals (pairs of paths in the plan tree stored as lists of length 2). In each iteration a list $[p_{low},p_{high}]$ is poped from this queue and the algorithm computes the prefix $p_{next}$ of a path to a leaf that lies between the leaf nodes at the end of paths $p_{low}$ and $p_{high}$. During an iteration, function  $\Call{makeChoice}{}$ then first exhausts choices from $p_{next}$ and afterwards chooses the middle option, i.e., $\frac{n}{2}$ for $n$ options. Since there may be choice points with an even number of options, the algorithm alternates between rounding up and rounding down.
We now explain the  $\Call{genNextIterChoices}{}$ and $\Call{makeChoice}{}$ functions in more detail.

 Consider $\Call{makeChoice}{}$ shown  as Alg.~\ref{alg:callback}. Here, variable $iter$ is a global variable that stores the current iteration (starting from $1$). In the first iteration ($iter = 1$), the function always chooses the first option ($0$) which generates the leftmost leaf of the plan tree. For instance, for the plan tree shown in Fig.~\ref{fig:bal-tree-example} (top) this is the path marked in red. In the second iteration ($iter = 2$), the function chooses the option leading to the rightmost child of the current node.  For instance, for the plan tree shown in Fig.~\ref{fig:bal-tree-example} (top) this is the path marked in green. In any succeeding iteration, the function takes options from $p_{next}$ which is a prefix the current plan to be chosen. $p_{next}$ is determine by $\Call{genNextIterChoices}{}$ at the end of each iteration.
 Once $p_{next}$ is exhausted, the function takes the middle option computed as $\frac{numChoices-1}{2}$ alternating between rounding up and rounding down. The flag $useLow$ determines whether we should round up or round down. Function $\Call{makeChoice}{}$ iteratively constructs the current plan $p_{cur}$ and keeps a list $n_{opt}$ which stores the number of options available at the choice points that were hit during the current iteration.

 Function $\Call{genNextIterChoices}{}$ shown as Alg.~\ref{alg:gen-next-option} stores the current plan as $p_{low}$ in the first iteration. In the second iteration, it stores $p_{cur}$ in variable $p_{high}$ and then pushes $[p_{low}, p_{high}]$ to queue $todo$. In all iterations except for the first two, the function
 pushes two new intervals $[p_{low},p_{cur}]$ and $[p_{cur},p_{high}]$ unless these intervals do not contain a middle element (e.g., $p_{low} + 1 = p_{cur}$). In all iterations except the first, the function pops an interval from the stack and then calls function $\Call{splitInterval}{}$ to compute $p_{next}$. This function in turn calls function $\Call{approxMiddlePlan}{}$ shown as Alg.~\ref{alg:gen-next-plan}. This function initializes the result with the common prefix of $p_{low}$ and $p_{high}$. Then for the first position $i$ where $p_{low}[i] \neq p_{high}[i]$, the function chooses an option that lies between $p_{low}$ and $p_{high}$. If no such option exists, then depending on the setting of flag $useLow$ it uses $p_{low}[i]$ or $p_{high}[i]$. Afterwards, it determines the first position where it is possible to choose an option higher than $p_{low}$ (lower than $p_{high}$) if such an option exists. Function $\Call{approxMiddlePlan}{}$ may return a plan that is equal to $p_{low}$ or $p_{high}$ in which case $p_{next}$ is set to $p_{low} + 1$.
Here $p_{low} + 1$ denotes the leaf immediately to the right of $p_{low}$ which we can determine in the same way as in our \textit{sequential-leaf-traversal} strategy.
 For example, assume that $p_{low}$ and $p_{high}$ are the plans shown in red and green in Fig.~\ref{fig:bal-tree-example}, respectively. Then if $useLow=0$ function \textsc{genNextIterChoices} would generated the plan shown in purple on the bottom in Fig.~\ref{fig:bal-tree-example}.

 \begin{algorithm}[t]
   \caption{Binary \textsc{makeChoice} Function}
   \label{alg:callback}
   \begin{algorithmic}[1]
    
     \Procedure{makeChoice}{$numChoices$}
      \If {$iter = 1$} \Comment{choose leftmost plan in $1^{st}$ iteration}
         \State $choice \gets 0$
      \ElsIf {$iter = 2$} \Comment{choose rightmost plan in $2^{st}$ iteration} 
         \State $choice \gets numChoices-1$
       \Else
       \If {$ len(p_{next}) > 0$}\Comment{take predetermined choice}
                   \State $choice \gets popHead(p_{next})$
                   \Else \Comment{choose middle option, alternating rounding up and down}
                 		\If {$useLow = 0$}
                 		    \State $choice \gets \lfloor \frac{numChoices - 1}{2} \rfloor$                    
                   		    \State $useLow \gets 1$
                		\Else
                 		    \State $choice \gets \lceil \frac{numChoices - 1}{2} \rceil$
                       	    \State $useLow \gets 0$
                	   \EndIf
         \EndIf
       \EndIf
         \State $p_{cur} \gets p_{cur} \cList choice$
         \State $n_{opts} \gets n_{opts} \cList numChoices$
       \State \textbf{return} choice
     \EndProcedure
        
  \end{algorithmic}
 \end{algorithm}

 \begin{algorithm}[t]
   \caption{Binary \textsc{genNextIterChoices} Function}
   \label{alg:gen-next-option}
   \begin{algorithmic}[1]
    
     \Procedure{genNextIterChoices}{ }
          \If {$iter = 1$}
         	  \State $p_{low} \gets p_{cur}$
          \Else 
             \If {$iter = 2$}
             \State $p_{high} \gets p_{cur}$
                  \State $todo \gets [[p_{low},p_{high}]]$
              \Else
                   \If {$p_{low} + 1 < p_{cur}$}
                       \State $todo \gets todo \cList [p_{low},p_{cur}] $
                   \EndIf
                   \If {$p_{cur} + 1 < p_{high}$}
                       \State $todo \gets todo \cList [p_{cur},p_{high}]$
                   \EndIf
              \EndIf
                  
              \State $[p_{low}, p_{high}] \gets \Call{popHead}{todo}$
              \If {$p_{cur} - 1 \neq p_{low}$}
         	        \State $useLow \gets 0$
         	       
              \Else 
                   \State $useLow \gets 1$
              \EndIf
              \State $\Call{splitInterval}{useLow, p_{low}, p_{high}, n_{opts}}$
         \EndIf
          \State $p_{cur} \gets []$
         \State $n_{opts} \gets []$
     \EndProcedure
    
   \end{algorithmic}
 \end{algorithm}

 \begin{algorithm}[t]
   \caption{Binary \textsc{splitInterval} Function}
   \label{alg:split-interval}
   \begin{algorithmic}[1]
    
     \Procedure{splitInterval}{$useLow$, $p_{low}$, $p_{high}$, $n_{opts}$}
             \State $p_{estNext} \gets \Call{approxMiddlePlan}{useLow, p_{low}, p_{high}, n_{opts}}$
         	 \If {$p_{estNext} = p_{low} \vee p_{estNext} = p_{high}$}
             \State $p_{next} \gets p_{low} + 1$
             \Else
                 \State $p_{next} = p_{estNext}$
             \EndIf
     \EndProcedure
    
   \end{algorithmic}
 \end{algorithm}

 \begin{algorithm}[t]
   \caption{Binary \textsc{approxMiddlePlan} Function}
   \label{alg:gen-next-plan}
   \begin{algorithmic}[1]
    
     \Procedure{approxMiddlePlan}{$useLow$, $p_{low}$, $p_{high}$, $n_{opts}$}
          \State $len_{min} \gets \min(len(p_{low}),len(p_{high}))$ 
          \State $len_{prefix} \gets \Call{findCommonPrefixLen}{p_{low} , p_{high}}$
          \State $result \gets []$

          \For{$i \in \{1,\ldots,len_{prefix}\}$}  \Comment{copy common prefix}
              \State $result \gets result \cList p_{low}[i]$
          \EndFor

          \State $i \gets len_{prefix} + 1$
          \State $c_{low} \gets p_{low}[i]$
          \State $c_{high} \gets p_{high}[i]$
          \State $diff \gets c_{high} - c_{low}$ \Comment{difference between plans at $i$}
          \If {$diff > 1$} \Comment{choose option between $c_{low}$ and $c_{high}$}
          \State $result \gets result \cList \lceil \frac{c_{high} + c_{low}}{2} \rceil$
          \State \Return result
          \ElsIf {$useLow = 0$} \Comment{choose plan left of $p_{high}$}
          \State $result \gets result \cList c_{high}$
          \While {$result[i] = p_{high}[i] \wedge i \leq len(p_{high})$}
          \State $result[i] \gets 0$
          \State $i \gets i + 1$
          \EndWhile
          \State $useLow \gets 1$ 
          \Else \Comment{choose plan right of $p_{low}$}
          \State $result \gets result \cList c_{low}$
          \While {$result[i] = p_{low}[i] \wedge i \leq len(p_{low})$}
          \State $result[i] \gets n_{opt}[i] - 1$
          \State $i \gets i + 1$
          \EndWhile

          \State $useLow \gets 0$
          \EndIf

                \State \Return result
     \EndProcedure
    
   \end{algorithmic}
 \end{algorithm}

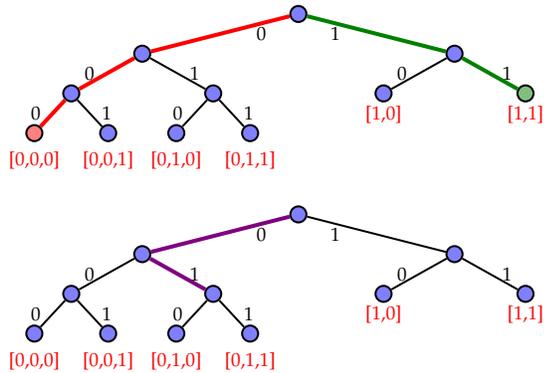
\begin{figure}[t]
  \centering
\resizebox{!}{0.3\columnwidth}{
\begin{tikzpicture}
[every node/.style={black,line width=1pt,circle,draw,fill=blue!50,label position={east},inner sep={1mm}},
every edge/.style={black,line width=1pt},
el/.style={draw=none,fill=none,inner sep={1mm}},
level distance=7mm,
level 1/.style={sibling distance=55mm},
level 2/.style={sibling distance=25mm},
level 3/.style={sibling distance=13mm}
]
  \node[label={[label distance=0.4cm]{}}] {} 
        child {node[label={}] {}
          child {node[label=left:{}] {} 
            child {node[fill=red!50,label={[label distance=-0.3cm]below:{\textcolor{red}{[0,0,0]}}}] {} 
              edge from parent [red,line width=2pt] node[left,el,black] {0}
            }
            child {node[line width=1pt,label={[label distance=-0.3cm,]below:{\textcolor{red}{[0,0,1]}}}] {} 
              edge from parent [black,line width=1pt] node[right,el] {1}
            }
            edge from parent [red,line width=2pt] node[left,el,black] {0}
          }
          child {node[label=right:{}] {} 
            child {node[label={[label distance=-0.3cm]below:{\textcolor{red}{[0,1,0]}}}] {} 
              edge from parent [black,line width=1pt] node[left,el] {0}
            }
            child {node[label={[label distance=-0.3cm]below:{\textcolor{red}{[0,1,1]}}}] {} 
              edge from parent [black,line width=1pt] node[right,el] {1}
            }
            edge from parent [black,line width=1pt] node[right,el] {1}            
          }
          edge from parent [red,line width=2pt] node[right,el,inner sep={4mm}] {0}
        }
        child {node[label=left:{}] {} 
            child {node[label={[label distance=-0.3cm]below:{\textcolor{red}{[1,0]}}}] {} 
              edge from parent [black,line width=1pt] node[left,el] {0}
            }
            child {node[fill=darkgreen!50,label={[label distance=-0.3cm]below:{\textcolor{red}{[1,1]}}}] {} 
              edge from parent [darkgreen,line width=2pt] node[right,el] {1}
            }
            edge from parent [darkgreen,line width=2pt] node[left,el,inner sep={4mm}] {1}
        }
  ;
\end{tikzpicture}
}
\resizebox{!}{0.3\columnwidth}{
\begin{tikzpicture}
[every node/.style={black,line width=1pt,circle,draw,fill=blue!50,label position={east},inner sep={1mm}},
every edge/.style={black,line width=1pt},
el/.style={draw=none,fill=none,inner sep={1mm}},
level distance=7mm,
level 1/.style={sibling distance=55mm},
level 2/.style={sibling distance=25mm},
level 3/.style={sibling distance=13mm}
]
  \node[label={[label distance=0.4cm]{}}] {} 
        child {node[label={}] {}
          child {node[label=left:{}] {} 
            child {node[label={[label distance=-0.3cm]below:{\textcolor{red}{[0,0,0]}}}] {} 
              edge from parent [black,line width=1pt] node[left,el,black] {0}
            }
            child {node[line width=1pt,label={[label distance=-0.3cm,]below:{\textcolor{red}{[0,0,1]}}}] {} 
              edge from parent [black,line width=1pt] node[right,el] {1}
            }
            edge from parent [black,line width=1pt] node[left,el,black] {0}
          }
          child {node[label=right:{}] {} 
            child {node[label={[label distance=-0.3cm]below:{\textcolor{red}{[0,1,0]}}}] {} 
              edge from parent [black,line width=1pt] node[left,el] {0}
            }
            child {node[label={[label distance=-0.3cm]below:{\textcolor{red}{[0,1,1]}}}] {} 
              edge from parent [black,line width=1pt] node[right,el] {1}
            }
            edge from parent [darkpurple,line width=2pt] node[right,el] {1}            
          }
          edge from parent [darkpurple,line width=2pt] node[right,el,inner sep={4mm}] {0}
        }
        child {node[label=left:{}] {} 
            child {node[label={[label distance=-0.3cm]below:{\textcolor{red}{[1,0]}}}] {} 
              edge from parent [black,line width=1pt] node[left,el] {0}
            }
            child {node[label={[label distance=-0.3cm]below:{\textcolor{red}{[1,1]}}}] {} 
              edge from parent [black,line width=1pt] node[right,el] {1}
            }
            edge from parent [black,line width=1pt] node[left,el,inner sep={4mm}] {1}
        }
  ;
\end{tikzpicture}
}

\caption{Plan space tree example for \textit{binary-search-traversal} strategy}
\label{fig:bal-tree-example}
\end{figure}

 \begin{Example}\label{eg:bal-search-traversal-example}
Consider Fig.~\ref{fig:bal-tree-example}, in the first two iterations we generate the plans \texttt{\upshape [[0,0,0],  [1,1]]} and initialize queue $todo$  to the singleton list \texttt{\upshape [[[0,0,0], [1,1]]]}. 
Next, since there is no common prefix and only two options 0 and 1, based on $useLow=0$ the next selected $p_{next}$ is \texttt{\upshape [0,1]} which is extended to plan \texttt{\upshape [0,1,0]}. As a side effect $useLow$ is set to $1$. At the end of this iteration two intervals are pushed on the queue: \texttt{\upshape [[0,0,0], [0,1,0]]} and \texttt{\upshape [[0,1,0], [1,1]]}. We first pop interval \texttt{\upshape [[0,0,0], [0,1,0]]} which contains the common prefix \texttt{\upshape [0]}. Since $useLow = 1$, for the second position we choose 1 and for the last one we choose 0.  We get the new plan \texttt{\upshape [0,1,0]} and set $useLow = 0$. This plan is equal to the current $p_{high}$ and, thus we set $p_{next} = p_{low} + 1$ which is \texttt{\upshape [0,0,1]}. We do not push interval \texttt{\upshape  [[0,0,0], [0,0,1]]} because $p_{low} + 1 = p_{cur}$ and also do not push \texttt{\upshape [[0,0,1], [0,1,0]]} since $p_{cur} + 1 = p_{high}$. In the next iteration we pop interval \texttt{\upshape [[0,1,0], [1,1]]}. These two plans have no common prefix. Because $useLow = 0$ we choose \texttt{\upshape [0]} and then continue to choose the rightmost option until the plan differs from \texttt{\upshape [0,1,0]} resulting in a plan \texttt{\upshape [0,1,1]}. Furthermore $useLow$ is set to $1$. We only push interval \texttt{\upshape [[0,1,1], [1,1]]}. In the next and final iteration we pop the interval we just pushed. Since there is no common prefix and $useLow=1$ we choose \texttt{\upshape [1]} and then proceed to choose option $0$ until the current plan differs from \texttt{\upshape [1,1]} yielding plan \texttt{\upshape [1,0]}. We do not push any new intervals during this iteration. Queue $todo$ is now empty, i.e., the plan space traversal is completed. 
 \end{Example}

\subsection{Balancing Optimization vs. Runtime.}\label{sec:2-competitive}
Recall that in Sec.~\ref{sec:balance-opt} we introduced the adaptive strategy which stops optimization once the time spend on optimization exceeds the estimated cost of the best plan found so far. We claimed that this algorithm is 2-competitive, i.e.,  $T_{opt} + T_{best}$ is  less than 2 times the minimal achievable cost of an algorithm that knows the full sequence of plans that will be generated and their costs upfront and, thus, can choose the stopping point that minimizes $T_{opt} + T_{best}$.

\begin{Theorem}
  The adaptive strategy is 2-competitive.
\end{Theorem}

\begin{proof}
We use $T_{est_{i}}$ to denote the estimated cost of the plan generated in the $i_{th}$ iteration and use $T_{best_{i}} = \min_{1 \leq k \leq i}T_{est_{k}}$ to denote the estimated cost of the best plan we found after $i$ iterations.
We use $T_{opt_{i}}$ to denote the duration of the $i_{th}$ iteration define $T_{opt_{\leq i}} = \sum_{j=1}^{i} T_{opt_{j}}$.
If we stop the algorithm after $i$ iterations then the total time spend is  $T_{i} = T_{opt_{\leq i}} + T_{best_{i}} $. Our goal is to minimize $T_{i}$. Let $min = \argmin_{i \in \mathbb{N}} T_{i}$. 
Our algorithm stops at iteration $our = \argmin_{i \in \mathbb{N}} (T_{opt_{\leq i}} \leq T_{best_{i}})$ and let $T_{opt_{our}}$ denote the time it spend on optimization.  
It is easy to see that $T_{our} \leq 2 \cdot T_{best_{our}} = 2 \cdot T_{opt_{our}}$.
We now  prove that 
$T_{our} \leq 2 \cdot T_{min}$.
\underline{CASE 1}: $our < min$, i.e., our algorithm stops before $min$ iterations. 
Since $our < min$, we also have $T_{opt_{our}} < T_{opt_{\leq min}}$. Thus, $T_{our} \leq 2 \cdot T_{opt_{our}} \leq 2 \cdot T_{opt_{min}} \leq 2 \cdot T_{min}$.
\underline{CASE 2}: $our > min$.
Thus, $T_{best_{our}} \leq T_{best_{min}}$ and we get $T_{our} \leq 2 \cdot T_{best_{our}} \leq 2 \cdot T_{best_{min}} \leq 2 \cdot T_{min}$. 
\underline{CASE 3}: $our = min$. We have $T_{our} = T_{min} \leq 2 \cdot T_{min}$.
\end{proof}

 \section{Notation}\label{sec:supp-notation}

Fig.~\ref{fig:overview-notations} shows a glossary of notation used in this paper.

\begin{figure*}[t]
\centering
\begin{minipage}{1\linewidth}
\renewcommand{\arraystretch}{1.4}
    \centering
          \resizebox{1\linewidth}{!}{
            \begin{minipage}{1.0\linewidth}
              \centering
\begin{tabular}{ |c|p{12cm}|}
\hline
  \chead Notations & \chead Description \\ \hline
  $\aDom$ & Universal domain of constants\\ \hline
  $\schema{Q}$ &  The schema of the result of query Q \\   \hline
  $\query(I)$ & The result of evaluating query $Q$ over database instance $I$ \\ \hline
  $\supp(R)$ & The support of relation $R$, i.e., the set of tuples that appear in $R$ with non-zero multiplicity\\ \hline
  $\schema{R}/\schema{S}$ & Rename the attributes from $\schema{R}$ as $\schema{S}$ \\   \hline
  $\colsOf(\theta) $ & All attributes referenced in a condition $\theta$\\   \hline
  $\colsOf(e)$ & The attributes referenced in expression $e$\\   \hline
  $Q[Q_1 \leftarrow Q_2]$ & The result of substituting subexpression (subgraph) $Q_1$ with $Q_2$ in the algebra graph for query Q \\   \hline
  $Q = op(Q')$ & Operator $op$ is the root of the algebra graph for query $Q$ that $Q'$ is the subquery of $Q$ ``below'' $op$ \\   \hline
  $\Diamond $ & The operator for which we are inferring a property (for bottom- up inference) \\   
  & or a parent of this operator (for top-down inference)  \\ \hline
  $\circledast$ & The root of a query tree \\   \hline
  $\ecClosure $ & Takes a set S of ECs as input and merges ECs if they overlap \\   \hline
  $\minKey(K)$ & Removes keys that are supersets of other keys from a set of keys $K$ \\   \hline
  $\qH(Q)$ & The height of a query Q, e.g., $\qH(Q)=3$ for $Q = \projection_A(\selection_{\theta}(R))$  \\   \hline
  $\qD(op)$ & The depth of operators in a query, e.g., if $Q = \projection_A(\selection_{\theta}(R))$ then $\qD(\projection_A) = 0$, $\qD(\selection_{\theta}) = 1$ and $\qD(R) = 2$ \\   \hline

\end{tabular}
  \end{minipage}
  }

\end{minipage}

\caption{Glossary of notation used in the paper}
\label{fig:overview-notations}

\end{figure*}

\bibliographystyle{IEEEtran}
\bibliography{2018-TKDE-Optimizer-longversion}

 \end{document}